\documentclass[a4paper,UKenglish,cleveref, autoref, thm-restate]{lipics-v2021}

\title{Problems in \NP\ can Admit Double-Exponential Lower Bounds when Parameterized by Treewidth or Vertex Cover\thanks{An extended abstract of parts of this paper was presented in~\cite{icalppaper}.}} 

\titlerunning{Problems in \NP\ can Admit Double-Exponential Lower Bounds} 

\author{Florent Foucaud}{Université Clermont Auvergne, CNRS, Mines Saint-Étienne, Clermont Auvergne INP, LIMOS, 63000 Clermont-Ferrand, France\and \url{https://perso.limos.fr/ffoucaud}}{florent.foucaud@uca.fr}{https://orcid.org/0000-0001-8198-693X}{ANR project GRALMECO (ANR-21-CE48-0004), French government IDEX-ISITE initiative 16-IDEX-0001 (CAP 20-25), International Research Center "Innovation Transportation and Production Systems" of the I-SITE CAP 20-25.}

\author{Esther Galby}{Department of Computer Science and Engineering, Chalmers University of Technology and University of Gothenburg, Gothenburg, Sweden}{galby@chalmers.se}{https://orcid.org/0009-0004-5398-2770}{}

\author{Liana Khazaliya}{Technische Universit\"{a}t Wien, Vienna, Austria \and \url{https://www.ac.tuwien.ac.at/people/lkhazaliya/}}{lkhazaliya@ac.tuwien.ac.at}{https://orcid.org/my-orcid?orcid=0009-0002-3012-7240}{Vienna Science and Technology Fund (WWTF) [10.47379/ICT22029]; Austrian Science Fund (FWF) [Y1329]; European Union's Horizon 2020 COFUND programme [LogiCS@TUWien, grant agreement No.\ 101034440].}

\author{Shaohua Li}{School of Computer Science and Engineering, Central South University, Changsha, China}{shaohua.li@csu.edu.cn}{https://orcid.org/0000-0001-8079-6405}{National Natural Science Foundation of China under Grant 62472449.}

\author{Fionn {Mc Inerney}}{Telef\'{o}nica Scientific Research, Barcelona, Spain \and \url{https://sites.google.com/view/fionn-mc-inerney/home?pli=1}}{fmcinern@gmail.com}{https://orcid.org/0000-0002-5634-9506}{Smart Networks and Services Joint Undertaking (SNS JU) under the EU's Horizon Europe and innovation programme under Grant Agreement No. 101139067 (ELASTIC).}

\author{Roohani Sharma}{University of Bergen, Bergen, Norway}{r.sharma@uib.no}{https://orcid.org/0000-0003-2212-1359}{}

\author{Prafullkumar Tale}{Indian Institute of Science Education and Research Pune, Pune, India \and \url{https://pptale.github.io/}}{prafullkumar@iiserb.ac.in}{https://orcid.org/0000-0001-9753-0523}{}

\authorrunning{F. Foucaud, E. Galby, L. Khazaliya, S. Li, F. {Mc Inerney}, R. Sharma, and P. Tale} 

\Copyright{Esther Galby, Liana Khazaliya, Fionn {Mc Inerney}, Roohani Sharma, and Prafullkumar Tale} 



\nolinenumbers 

\hideLIPIcs

\usepackage{complexity}
\usepackage{mathtools}
\usepackage{bm}
\usepackage{xspace}
\usepackage{mathdots}
\usepackage{float}
\usepackage{needspace}
\usepackage[framemethod=tikz]{mdframed}

\newcommand{\mdfull}{\textsc{Metric Dimension}\xspace}
\newcommand{\smdfull}{\textsc{Strong Metric Dimension}\xspace}

\newcommand{\gsfull}{\textsc{Geodetic Set}\xspace}

\newcommand{\OO}{\mathcal{O}}

\newcommand{\calF}{\mathcal{F}}
\newcommand{\calO}{\mathcal{O}}

\newcommand{\true}{\texttt{True}}
\newcommand{\false}{\texttt{False}}
\newcommand{\bit}{\textsf{bin}}

\newcommand{\bitrep}{\textsf{bit-rep}}

\newcommand{\setrep}{\textsf{set-rep}}
\newcommand{\bitrepnullifier}{\textsf{nullifier}}
\newcommand{\bits}{\textsf{bits}}
\DeclareMathOperator{\md}{md}

\DeclareMathOperator{\dist}{dist}
\DeclareMathOperator{\gs}{gs}
\newcommand{\vc}{\mathtt{vc}}
\newcommand{\fvs}{\mathtt{fvs}}
\newcommand{\td}{\mathtt{td}}
\newcommand{\tw}{\mathtt{tw}}
\newcommand{\pw}{\mathtt{pw}}
\newcommand{\diam}{\mathtt{diam}}
\newcommand{\yes}{\textsc{Yes}}
\newcommand{\no}{\textsc{No}}

\newcommand{\ETH}{\textsf{ETH}}
\newcommand{\smd}{\mathtt{smd}}
\newcommand{\glb}{\textsf{glb}}
\newcommand{\pndt}{\textsf{pndt}}
\newcommand{\conport}{\textsf{con-port}}

\newtheoremstyle{noparentheses}
    {}{}{\itshape}{}%
    {\bfseries}{.}{ }%
    {\thmname{#1}\thmnumber{ #2}\thmnote{ {\mdseries #3}}}
\theoremstyle{noparentheses} 
\newtheorem*{theoremnp*}{Theorem}

\usepackage{color}

\newcommand{\defproblem}[3]{
  \vspace{3mm}
\noindent\fbox{
  \begin{minipage}{0.96\textwidth}
  \begin{tabular*}{\textwidth}{@{\extracolsep{\fill}}lr} #1 \\ \end{tabular*}
  {\bf{Input:}} #2  \\
  {\bf{Question:}} #3
  \end{minipage}
  }
  \vspace{3mm}
}

\usepackage[colorinlistoftodos,textsize=tiny]{todonotes}

\usepackage{tikz}
\usetikzlibrary{positioning,backgrounds,patterns,calc}

\tikzset{
  circ/.style = {circle,draw,fill,inner sep=1.3pt}
}

\newtheorem{reduction rule}{Reduction Rule}

\EventEditors{John Q. Open and Joan R. Access}
\EventNoEds{2}
\EventLongTitle{42nd Conference on Very Important Topics (CVIT 2016)}
\EventShortTitle{CVIT 2016}
\EventAcronym{CVIT}
\EventYear{2016}
\EventDate{December 24--27, 2016}
\EventLocation{Little Whinging, United Kingdom}
\EventLogo{}
\SeriesVolume{42}
\ArticleNo{23}

\begin{document}

\maketitle

\begin{abstract}

	Treewidth serves as an important parameter that, when bounded, yields tractability for a wide class of problems. 
	For example, graph
	problems expressible in Monadic Second Order (MSO) logic and \textsc{Quantified SAT} or, more generally, \textsc{Quantified CSP}, are fixed-parameter tractable parameterized by the treewidth {of the input's (primal) graph} 
	plus the length of the MSO-formula~[Courcelle, Information \& Computation 1990]  and the quantifier rank~[Chen, ECAI 2004], respectively. 
	The algorithms generated by these (meta-)results have running times whose dependence on treewidth is 
	a tower of exponents. 
	A {conditional lower bound} by Fichte, Hecher, and Pfandler~[LICS 2020] shows that, for \textsc{Quantified SAT},
	the height of this tower is equal to the number of quantifier alternations.
	These types of
	lower bounds, which show that at least {\em double-exponential} factors in the running time are {\em necessary}, exhibit the extraordinary level of computational hardness for such problems,
	and are 
	rare in the current literature: there are only a handful 
	of such lower bounds (for treewidth and vertex cover parameterizations) and all of them are for problems that are $\#$\NP-complete, $\Sigma_2^p$-complete, $\Pi_2^p$-complete, or complete for even higher levels of the polynomial hierarchy.
	
	Our results demonstrate,
	for the first time, that it is not necessary to go higher up in the polynomial hierarchy to achieve double-exponential lower bounds: we derive double-exponential lower bounds in the treewidth {$(\textsf{tw})$} and 
	the vertex cover number {$(\textsf{vc})$}, 
	for natural, important, and well-studied {\em \NP-complete} graph problems. 
	Specifically, we design a {\em technique} to obtain such lower bounds and show its {\em versatility} by applying it to three different problems:
	\textsc{Metric Dimension}, \textsc{Strong Metric Dimension},
	and \textsc{Geodetic Set}.
	We prove that these
	problems do not admit $2^{2^{o(\textsf{tw})}} \cdot n^{\mathcal{O}(1)}$-time algorithms, even on bounded diameter graphs, unless the \ETH\ fails (here, $n$ is the number of vertices in the graph).
	In fact, for \textsc{Strong Metric Dimension}, the double-exponential {lower bound holds}
	even {for} the {vertex cover number}.
	We further complement all our
	lower bounds with matching (and sometimes {non-trivial})
	upper bounds. 
	
	For the conditional lower bounds, 
	we design and use a novel, yet simple technique 
	based on \emph{Sperner families} of sets.
	We believe that the amenability of our technique will lead to obtaining such lower bounds for many other problems in~\NP.
\end{abstract}

\newpage
\tableofcontents
\newpage


\section{Introduction}
\label{sec:intro}
Many interesting computational problems turn out to be intractable. In these cases, identifying parameters under which the problems become tractable is desirable.
In the area of parameterized complexity, 
treewidth is a cornerstone parameter since 
a large class of problems become tractable on graphs of bounded treewidth.

Courcelle's celebrated theorem~\cite{Courcelle90}
states that the class of graph problems expressible in Monadic Second-Order
Logic (MSOL) of constant size is fixed-parameter tractable (\FPT) when parameterized by the treewidth of the graph. 
That is, such problems admit algorithms whose running time is of the form $f(\tw)\cdot \poly(n)$, where $\tw$ is the treewidth of the input, $n$ is the size of the input, and 
 $f$ is a function that depends only on $\tw$. 
 Similarly, a result by Chen~\cite{DBLP:conf/ecai/Chen04} shows that the {\sc Quantified SAT} ({\sc Q-SAT}) problem can also be solved in time $f(\tw) \cdot \poly(n)$, where $\tw$ is the treewidth of the primal graph of the input formula and $f$ is a function that depends only on $\tw$ and the number of quantifier alternations in the input formula. 
{\sc Q-SAT} is a generalization of {\sc SAT} that allows universal and existential quantifications over the variables. 
Note that {\sc Q-SAT} with $k$ quantifier alternations is $\Pi_k^p$-complete or $\Sigma_k^p$-complete.
Unfortunately, in both of the aforementioned results,
the function $f$
is a tower of exponents whose height depends roughly on the size of the MSOL and input formulas, respectively. 
For {\sc Q-SAT}, the height of this tower equals the number of quantifier alternations in the {\sc Q-SAT} instance~\cite{DBLP:conf/ecai/Chen04}.
 
Over the years, the focus shifted to making such \FPT\ algorithms as efficient as possible. Thus, a natural question is to ask when this higher-exponential 
dependence on treewidth is necessary.
 There is a rich literature that provides (conditional)
 lower bounds on this dependency for many problems,
 and these bounds are commonly of the form $2^{o(\tw)}$ or, in some unusual cases, $2^{o(\tw \log \tw)}$ (e.g., \cite{6108160,doi:10.1137/16M1104834}) and even $2^{o(\poly(\tw))}$ (e.g., \cite{10.1016/j.ic.2017.04.009,DBLP:conf/mfcs/Pilipczuk11}). Most notably, these lower bounds are far from the tower of exponents upper bounds given by the (meta-)results discussed above.
In this work, we develop a simple technique that allows to prove double-exponential dependence on the treewidth $\tw$ and the vertex cover number $\vc$, two of the most fundamental graph parameters. Notably, these are the first such results for problems in \NP, and 
we believe that the amenability of our technique will lead to many more similar results for other problems in \NP.

Indeed, after a preprint of this paper appeared on arxiv, our technique was also used to prove double-exponential dependence on $\vc$ for an \NP-complete machine learning problem~\cite{CCMR23}
and double-exponential dependence on the solution size and $\tw$ for \NP-complete \emph{identification problems} like 
\textsc{Test Cover} and \textsc{Locating-Dominating Set}~\cite{chakraborty2024tight}.

\subparagraph{Double-exponential lower bounds: treewidth and vertex cover parameterizations.} 
Fichte, Hecher, and Pfander~\cite{DBLP:conf/lics/FichteHP20} recently proved that, assuming the Exponential Time Hypothesis\footnote{The Exponential Time Hypothesis roughly states that $n$-variable {\sc 3-SAT} cannot be solved in time $2^{o(n)}$.} (\ETH), {\sc Q-SAT} with $k$ quantifier alternations cannot be solved in time significantly better than a tower of exponents of height $k$ in the treewidth. 
This exemplifies an interesting but expected trait of this problem: 
its complexity, in terms of the height of the exponential tower in $\tw$, increases with each quantifier alternation. 
It strengthened the result that appeared in~\cite{LM17}, where conditional double-exponential lower bounds for $\exists \forall${\sc SAT} and $\forall \exists${\sc SAT} were given. 
The results in~\cite{LM17} also yield a double-exponential lower bound in $\vc$ 
of the primal graph for both problems.
Besides these results, there are only a handful of other problems known to 
require higher-exponential dependence in the treewidth
of the input graph (or the primal graph of the input formula).
Specifically, the 
$\Pi^p_2$-complete \textsc{$k$-Choosability} problem and the $\Sigma_3^p$-complete \textsc{$k$-Choosability Deletion} problem
admit a double-exponential and a triple-exponential lower bound in treewidth~\cite{MM16}, respectively.
Recently, the $\Sigma^2_p$-complete problems \textsc{Cycle HitPack} and \textsc{$H$-HitPack}, for a fixed graph $H$, were shown to admit tight algorithms that are double-exponential in the treewidth~\cite{DBLP:journals/corr/abs-2402-14927}.
Further, the $\Sigma^2_p$-complete problem \textsc{Core Stability} was shown to admit a tight double-exponential lower bound in the treewidth, even on graphs of bounded degree~\cite{HKL24}.
Lastly, the $\#$\NP-complete counting problem
\textsc{Projected Model
  Counting} admits a double-exponential lower bound in $\tw$~\cite{DBLP:journals/ai/FichteHMTW23,DBLP:conf/sat/FichteHMW18}. For other double-exponential lower bounds, see~\cite{DBLP:journals/algorithmica/AchilleosLM12,DBLP:journals/siamcomp/CyganPP16,DBLP:journals/talg/FominGLSZ19,HKL24,JKL23,DBLP:journals/toct/KnopPW20,KLMPS24,kunnemann_et_al:LIPIcs.ICALP.2023.131,DBLP:conf/sat/LampisMM18,DBLP:conf/fsttcs/Lokshtanov0SX21,DBLP:journals/dmtcs/PilipczukS20,tale2024} and Section~\ref{sec:related-work-double-exp}.
 
All the double- (or higher) exponential lower bounds in treewidth mentioned so far are for problems that are $\#$\NP-complete, $\Sigma_2^p$-complete, $\Pi_2^p$-complete, or complete for even higher levels of the polynomial hierarchy. 
To quote~\cite{MM16}: 
\emph{``$\Pi^p_2$-completeness of these
problems already gives sufficient explanation why double- [$\dots$]
exponential dependence on treewidth is needed. [$\dots$] the quantifier
alternations in the problem definitions are the common underlying
reasons for being in the higher levels of the polynomial hierarchy and
for requiring unusually large dependence on treewidth.''}

As mentioned above, we develop a technique that allows to demonstrate, for the first time, that it is not necessary to go to higher levels of the polynomial hierarchy to achieve double-exponential lower bounds in the treewidth or
the vertex cover number of the graph.

\medskip

\begin{mdframed}[backgroundcolor=green!10]
We prove that three natural and well-studied {\NP-complete} problems
admit double-exponential lower bounds in $\tw$ or $\vc
 $, under the \ETH.
These are the first problems in \NP\ known to admit such lower bounds.\footnotemark[2]
\end{mdframed}
\footnotetext[2]{While it may be possible to artificially engineer a graph problem or graph representation of a problem in \NP\ that admits such lower bounds (although, to the best of our knowledge, this has not been done), we emphasize that this is not the case for these three natural and well-established graph problems in \NP.}

\subparagraph*{\boldsymbol{\NP}-complete metric-based graph problems.}
We study three \emph{metric-based graph problems}.
These problems are \textsc{Metric Dimension}, \textsc{Strong Metric Dimension}, and \textsc{Geodetic Set}, and they arise from network
design and network monitoring.
Apart from serving as examples for double-exponential dependence on treewidth and the amenability of our technique, these problems are of interest in their own right, and possess a rich literature both in the algorithms and discrete mathematics communities (see Section~\ref{sec:related-work-metric}). 
Their non-local nature has posed interesting algorithmic challenges and 
our results, 
as we explain later, supplement the already vast literature on the structural parameterizations of these problems.
Below we define the three above-mentioned problems formally, and particularly focus on \mdfull\ as it is the most popular and well-studied of the three. 

\defproblem{\mdfull}{A graph $G$ and a positive integer $k$.}{Does there exist $S \subseteq V(G)$ such that $|S| \leq k$ and, for any pair of vertices $u,v\in V(G)$,
there exists a vertex $w\in S$ with $d(w,u)\neq d(w,v)$?}

The \mdfull problem dates back to the 70s~\cite{HM76,Slater75}. As in geolocation
problems, the aim is to distinguish the vertices of a graph via their distances to a solution set. 
\mdfull was first shown to be \NP-complete in general graphs in Garey and
Johnson's book~\cite[GT61]{GJ79}, and this was later extended to many
restricted graph classes~\cite{DiazPSL17,ELW15,FoucaudMNPV17b},
including graphs of diameter~$2$~\cite{FoucaudMNPV17b} and graphs of
pathwidth~24~\cite{LM21}. 
In a seminal paper, \mdfull was proven to be
\W[2]-hard parameterized by the solution size $k$, even in subcubic
bipartite graphs~\cite{HartungN13}. This drove the subsequent meticulous study of the problem under structural
parameterizations.
  
In particular, the complexity of \mdfull parameterized by treewidth
remained an intriguing open problem for a long time.
Recently, it was shown that \mdfull\ is para-\NP-hard parameterized by
pathwidth ($\pw$)~\cite{LM21} (an earlier
result~\cite{BP21} showed that it is \W[1]-hard for pathwidth). 
A subsequent paper showed that the problem is \W[1]-hard
parameterized by the combined parameter feedback vertex set number
($\fvs$) plus pathwidth of the graph~\cite{GKIST23}. See Section~\ref{sec:related-work-metric} for more related work on \mdfull. 

We conclude this part with the definitions of
the remaining two problems,
both of which are known to be \NP-Complete~\cite{floCALDAM20,OP07}.
\textsc{Geodetic Set} is also
\W[1]-hard parameterized by the solution size, feedback vertex set number, and pathwidth, combined~\cite{KK22}.

\defproblem{\smdfull}{A graph $G$ and a positive integer $k$.}{Does there exist $S \subseteq V(G)$ such that $|S| \leq k$ and, for any pair of vertices $u,v\in V(G)$,
there exists a vertex $w\in S$ such that either $u$ lies on some shortest path between $v$ and $w$, or $v$ lies on some shortest path between $u$ and $w$?}

\defproblem{\gsfull}{A graph $G$ and a positive integer $k$.}{Does there exist $S \subseteq V(G)$ such that $|S| \leq k$ and, for any vertex $u \in V(G)$, there are two
vertices $s_1, s_2 \in S$ such that a shortest path from $s_1$ to
$s_2$ contains $u$?}

\subparagraph*{Our technical contributions.}
As \mdfull\ and \gsfull\ are \NP-complete on bounded diameter graphs 
\emph{or} on bounded treewidth graphs, 
we study their parameterized complexity with $\tw+\diam$ as the parameter
and prove the following results.

\medskip

\begin{mdframed}[backgroundcolor=yellow!10]
\begin{enumerate}
\item \mdfull\ and \gsfull\ do not admit algorithms running in time $2^{f(\diam)^{o(\tw)}} \cdot n^{\OO(1)}$, for any computable function $f$, unless the \ETH\ fails.
(Sections~\ref{sec:lower-bound-diam-treewidth-MD},~\ref{sec:lower-bound-diam-treewidth})

\smallskip

\item \smdfull\ does not even admit an algorithm with a running time of $2^{2^{o(\vc)}} \cdot n^{\OO(1)}$, unless the \ETH\ fails.
This also implies the problem does not admit a kernelization algorithm that outputs an instance with $2^{o(\vc)}$ {\em vertices}, unless the \ETH\ fails
(Section~\ref{sec:strong-met-dim-lower-bound-vc}).

\end{enumerate}
\end{mdframed}

\smallskip

The above lower bounds for $\tw+\diam$, in particular, imply
that \mdfull and \gsfull on graphs of bounded diameter cannot admit
$2^{2^{o(\tw)}} \cdot n^{\OO(1)}$-time algorithms,
unless the \ETH\ fails. 
The reduction in
Section~\ref{sec:lower-bound-diam-treewidth-MD} also works for $\fvs$ and $\td$ for \mdfull, and
the reduction in
Section~\ref{sec:lower-bound-diam-treewidth} works for $\td$ for \gsfull.

We show that all our lower bounds are tight by providing  
algorithms (kernelization algorithms, respectively) with matching running times (guarantees, respectively). 

\medskip

\begin{mdframed}[backgroundcolor=yellow!10]
\begin{enumerate}
\item \mdfull\ and \gsfull\ admit algorithms running in time $2^{\diam^{\OO(\tw)}} \cdot n^{\OO(1)}$. (Sections~\ref{subsec:algo-tw-diam-MD}, \ref{subsec:algo-tw-diam-GD})

\smallskip

\item \smdfull\ admits an algorithm running in time $2^{2^{\OO(\vc)}} \cdot n^{\OO(1)}$ and a kernel with $2^{\OO(\vc)}$ vertices.
(Section~\ref{subsec:algo-vc-SMD})
\end{enumerate}
\end{mdframed}

\smallskip

The (kernelization) algorithms for the $\vc$ parameterization are very simple, whereas the algorithms for the $\tw+\diam$ parameter are highly non-trivial and require showing interesting locality properties in the instance. 
Further, for our $\tw+\diam$ parameterized algorithms, the (double-exponential) dependency of treewidth in the running time is unusual (and rightly so, as exhibited by our lower bounds),
as most natural graph problems in \NP\ for which a dedicated algorithm 
(i.e., \emph{not} relying on Courcelle's theorem) parameterized by
treewidth is known, can be solved in time $2^{\calO(\tw)} \cdot
n^{\calO(1)}$, $2^{\calO(\tw \cdot \log(\tw))} \cdot n^{\calO(1)}$ or $2^{\calO(\poly(\tw))} \cdot n^{\calO(1)}$.

Finally, our reductions rely on a novel, yet simple technique based on {\em Sperner families} of sets that allows to encode particular {\sc SAT} relations across large sets of
variables and clauses into relatively small vertex-separators. As mentioned before, we believe that this technique is the key to obtaining such lower bound results for other problems in \NP. In particular, as witnessed by our results, our technique has the additional features that it even allows to prove such lower bounds in very restricted cases, such as bounded diameter graphs, and is not specific to any one structural parameter, as it also works for, e.g., the feedback vertex set number and treedepth. We elaborate on our technique in the next section.

\section{Technical Overview}\label{subsec:technique}

In this section, we present an overview of our lower bound techniques.
We first exhibit our technique to obtain the double-exponential lower bounds in its most general setting.
Then, we continue with the problem-specific tools we developed that are required for the reductions.


The first integral part of our technique is to reduce from a variant of {\sc 3-SAT} known as {\sc 3-Partitioned-3-SAT} that was introduced in~\cite{DBLP:journals/corr/abs-2302-09604}.
In this problem, the input is a formula~$\psi$ in $3$-\textsc{CNF} form,
together with a partition of the set of its variables into three
disjoint sets $X^{\alpha}$, $X^{\beta}$, $X^{\gamma}$,
with $|X^{\alpha}| = |X^{\beta}| = |X^{\gamma}| = n$, and
such that no clause contains more than one variable from each of
$X^{\alpha}$,  $X^{\beta}$, and $X^{\gamma}$.
The objective is to determine whether $\psi$ is satisfiable.
Unless the \ETH\ fails, \textsc{3-Partitioned-3-SAT} does not admit
an algorithm running in time $2^{o(n)}$~\cite[Theorem 3]{DBLP:journals/corr/abs-2302-09604}.

\begin{figure}[t]
    \centering
        \includegraphics[scale=0.475]{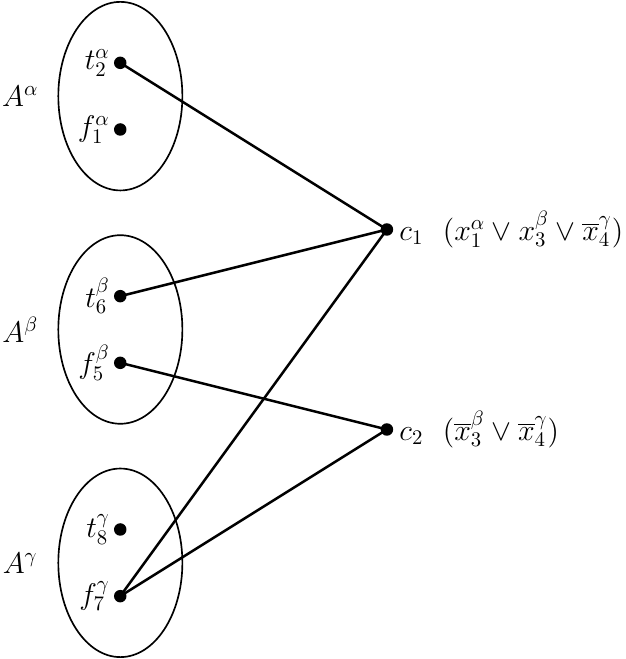}
        \hspace{0.05cm}
        \vline
        \hspace{0.15cm}
        \includegraphics[scale=0.475]{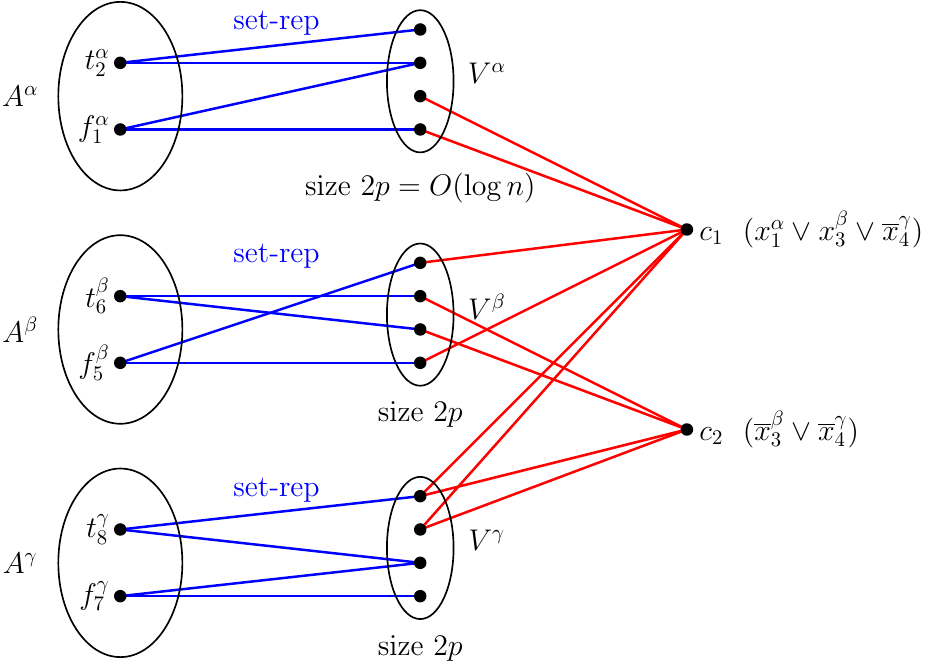}
    \caption{{\bf Graph representations of {\sc \bf 3-Partitioned-3-SAT}.} (Left) incidence graph representation. (Right) representation with small separators using our technique. Note, for example, that $x_1^{\alpha}$ appears as a positive literal in the clause $C_1$. Thus, on the left, $t_2^{\alpha}$ is the only literal vertex in $A^{\alpha}$ incident to $c_1$, while on the right, $t_2^{\alpha}$ is the only literal vertex in $A^{\alpha}$ that does not share a common neighbor with $c_1$ in $V^{\alpha}$. The edges from $c_2$ to each vertex in $V^{\alpha}$ are omitted for clarity.}
       \label{fig:technique}
\end{figure}

Typical reductions from satisfiability problems to graph problems usually entail representing the satisfiability problem by its incidence graph, in which each variable is represented by two vertices corresponding to its positive and negative literals.
In this representation, a clause vertex is adjacent to a literal vertex if and only if it contains that literal in $\psi$ (see~\cref{fig:technique} (left) for an illustration).
However, this naive approach does not lead to any structural parameters of the incidence graph being of bounded size.
The {\em core idea} of our technique is to instead represent the relationships between clause and literal vertices via edges from these two sets of vertices to ``small'' separators (three separators in the case of {\sc 3-Partitioned-3-SAT}) that encode these relationships.

Formally, this is achieved as follows. For a positive integer $p$, define $\calF_p$ as the collection of
subsets of $[2p]$ that contains exactly $p$ integers.
We critically use the fact that no set in $\calF_p$
is contained in any other set in $\calF_p$
(such a collection of sets are called a \emph{Sperner family}).
Let $\ell$ be a positive integer such that $\ell \leq \binom{2p}{p}$.
We define
$\setrep: [\ell] \mapsto \calF_p$ as a one-to-one function
by arbitrarily assigning a set in $\calF_p$ to an integer in $[\ell]$.
By the asymptotic estimation of the central binomial coefficient, $\binom{2p}{p}\sim \frac{4^p}{\sqrt{\pi \cdot p}}$ ~\cite{Sperner}.
To get the upper bound of $p$, we scale down the asymptotic function and have $\ell \leq \frac{4^p}{2^p}=2^p$.
Thus, $p=\OO(\log \ell)$. 

Let $\psi$ be an instance of {\sc 3-Partitioned-3-SAT} on $3n$ variables, and let $p$ be the smallest integer such that $2n \leq \binom{2p}{p}$. In particular, $p=\OO(\log n)$. Define
$\setrep: [2n] \mapsto \calF_p$ as above. Rename the variables in $X^{\alpha}$ to $x^{\alpha}_i$ for all $i\in [n]$.
For each variable $x^{\alpha}_i$, add two vertices $t^{\alpha}_{2i}$ and $f^{\alpha}_{2i-1}$ corresponding to the positive and negative literals of $x^{\alpha}_i$, respectively.
Let $A^{\alpha} = \{t^{\alpha}_{2i},  f^{\alpha}_{2i - 1} |\ i \in [n] \}$.
Add a {\em validation portal} with $2p$ vertices, denoted by $V^{\alpha}=\{v^{\alpha}_1,\ldots,v^{\alpha}_{2p}\}$.
For each $i\in [n]$, add the edge $t^{\alpha}_{2i}v^{\alpha}_{p'}$ for each $p'\in \setrep(2i)$.
Similarly, for each $i\in [n]$, add the edge $f^{\alpha}_{2i-1}v^{\alpha}_{p'}$ for each $p'\in \setrep(2i-1)$.
Repeat the above steps for $\beta$ and $\gamma$.

Now, for each clause $C_j$ ($j\in [m]$) in $\psi$, add a clause vertex $c_j$.
Let $\delta\in \{\alpha,\beta,\gamma\}$. For all $i\in [n]$ and $j\in [m]$, if the variable $x^{\delta}_i$ appears as a positive (negative, respectively) literal in the clause $C_j$ in $\psi$, then add the edge $c_jv^{\delta}_{p'}$ for each $p'\in [2p]\setminus \setrep(2i)$ ($p'\in [2p]\setminus \setrep(2i-1)$, respectively).
For all $j\in [m]$, if no variable from $X^{\delta}$ appears in $C_j$ in $\psi$, then make $c_j$ adjacent to all the vertices in $V^{\delta}$.
See~\cref{fig:technique} (right) for an illustration.

As a clause contains at most one variable from $X^{\delta}$ in $\psi$, $c_j$ and $t^{\delta}_{2i}$ ($f^{\delta}_{2i-1}$, respectively) do not share a common neighbor in $V^{\delta}$ if and only if the clause $C_j$ contains $x^{\delta}_i$ as a positive (negative, respectively) literal in $\psi$.
For the reductions, we use this representation of the relationship between clause and literal vertices. Since $p=\OO(\log n)$, this ensures that $\tw(G)=\OO(\log n)$, which we exploit along with the fact that, unless the \ETH\ fails, \textsc{3-Partitioned-3-SAT} does not admit an algorithm running in time $2^{o(n)}$.
 
\subsection{Basic Tools for Lower Bounds}
 
For brevity, we focus on \textsc{Metric Dimension} and explain our problem-specific tools in this context.
We use two such simple tools: the \emph{bit representation gadget} and the \emph{set representation gadget}. The set representation gadget is the problem-specific implementation of the above technique, and it uses the bit representation gadget.

Before going further, we need to define some terms related to \mdfull. 
The set $S$ defined in the problem statement of \textsc{Metric Dimension} is called a {\em resolving set} of $G$.
A subset of vertices $S'\subseteq V(G)$ {\it resolves} a pair of vertices $u,v\in V(G)$ if there exists a vertex $w \in S'$ such that $d(w,u)\neq d(w,v)$.
Lastly, a vertex $u\in V(G)$ is {\it distinguished} by a subset of vertices $S'\subseteq V(G)$ if, for any $v\in V(G)\setminus \{u\}$, there exists a vertex $w\in S'$ such that $d(w,u)\neq d(w,v)$.

\begin{figure}[t]
    \centering
        \includegraphics[scale=1.23]{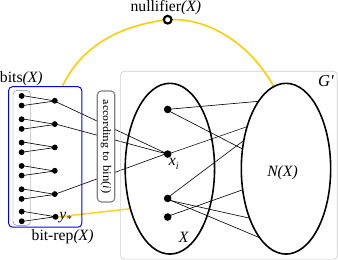} \hspace{0.4cm}
    \caption{\textbf{Set Identifying Gadget}. The blue box represents $\bitrep(X)$ and the yellow lines represent that $\bitrepnullifier(X)$ is adjacent to each vertex in $(\bitrep(X)\setminus \bits(X))\cup N(X)$, and $y_{\star}$ is adjacent to each vertex in $X$. Also, $G'$ is not necessarily restricted to the graph induced by the vertices in $X\cup N(X)$.}
       \label{fig:set-identifying-gadget-tech}
\end{figure}

\subparagraph*{Bit Representation Gadget 
to Identify Sets.}
\label{subsubsec:gadget-set}
Suppose we are given a graph $G'$ and a subset $X\subseteq V(G')$ of its vertices.
Further, suppose that we want to add a vertex set $X^+$ to $G'$ to obtain a new graph $G$ with the following properties. We want that each vertex in $X \cup X^+$ is distinguished by vertices in $X^+$ that must be in any resolving set $S$ of $G$, and that no vertex in $X^+$ can resolve any ``critical pair'' of vertices in $G$. Roughly, a pair of vertices is critical if it forces certain ``types'' of vertices to be in any resolving set $S$ of $G$, and the selection of the specific vertices of those types depends on the solution to the problem being reduced from (which, in our case, is \textsc{3-Partitioned-3-SAT}~\cite{DBLP:journals/corr/abs-2302-09604}).  
We refer to the graph induced by the vertices of $X^+$, along with the edges connecting $X^+$ to $G'$, as the 
\emph{Set Identifying Gadget} for the set $X$.
Given a graph $G'$ and a non-empty subset $X\subseteq V(G')$ of its vertices, to construct such a graph $G$, we add vertices and edges to $G'$ as follows (see \cref{fig:set-identifying-gadget-tech}):
\begin{itemize}
\item The vertex set $X^+$ that we are aiming to add is the union of a set $\bitrep(X)$ and a special vertex denoted by $\bitrepnullifier(X)$.
\item First, let $X=\{x_i\mid i\in [|X|]\}$, and set $q := \lceil \log(|X|+ 2) \rceil+1$.
We select this value for $q$ to \emph{(1)} uniquely represent each integer in $[|X|]$ by its bit-representation in binary (note that we start from $1$ and not $0$), \emph{(2)} ensure that the only vertex whose bit-representation contains all $1$'s is $\bitrepnullifier(X)$, and \emph{(3)} reserve one spot for an additional vertex $y_{\star}$.
\item For every $i \in [q]$, add three vertices $y^a_i,  y_i, y^b_i$, and add the path $(y^a_i, y_i, y^b_i)$.
\item Add $3$ vertices $y^a_{\star},  y_{\star}, y^b_{\star}$ and the path $(y^a_{\star}, y_{\star}, y^b_{\star})$. 
Add edges to make $\{y_i\mid\ i \in [q] \}\cup\{y_{\star}\}$ a clique.
Make $y_{\star}$ adjacent to each vertex in $X$.
Let $\bitrep(X)=\{y_i, y^a_i, y^b_i\mid i\in [q]\}\cup \{y_{\star}, y^a_{\star}, y^b_{\star}\}$ and denote its subset by $\bits(X)=\{y^a_i, y^b_i\mid i\in [q]\}\cup \{y^a_{\star}, y^b_{\star}\}$.
\item For every integer $j \in [|X|]$, let $\bit(j)$ denote the binary
representation of $j$ using $q$ bits.
Connect $x_j$ with $y_{i}$
if the $i^{th}$ bit (going from left to right) in $\bit(j)$ is $1$.
\item Add a vertex, denoted by $\bitrepnullifier(X)$,
and connect it to each vertex in $\{y_i\mid i \in [q] \}\cup\{y_{\star}\}$.
\item For every vertex $u \in V(G)\setminus (X\cup X^+)$ such that
$u$ is adjacent to some vertex in $X$, add an edge between
$u$ and $\bitrepnullifier(X)$.
We add this vertex to ensure that
vertices in $\bitrep(X)$ do not resolve critical pairs
in $V(G)$.
\end{itemize}

\subparagraph{Set Representation Gadget.}
\label{subsec:prelim-3-Par-3-SAT-Met-Dim-diam-tw-tech}
We define $\setrep: [\ell] \mapsto \calF_p$ as in \cref{subsec:technique}, and recall that $p=\OO(\log \ell)$. 
Suppose we have a ``large'' collection of vertices,  say $A = \{a_1, a_2, \dots, a_{\ell} \}$,
and a ``large'' collection of critical pairs
$C = \{\langle c^{\circ}_1, c^{\star}_1\rangle,  \langle c^{\circ}_2, c^{\star}_2\rangle,  \dots,
\langle c^{\circ}_m,  c^{\star}_m\rangle\}$.
Moreover,  we are given an injective function $\phi:[m] \mapsto [\ell]$.
The objective is to design a gadget such that only $a_{\phi(q)} \in A$
can resolve a critical pair $\langle c^{\circ}_q, c^{\star}_q\rangle \in C$
for any $q \in [m]$, while keeping the treewidth of this part of
the graph of order $\calO(\log(|A|))$. With this in mind, we do the following.

\begin{itemize}
\item Add vertices and edges to identify the set $A$ and to add critical pairs in $C$ (for each critical pair in $C$, both vertices share the same bit-representation in the Set Identifying Gadget for $C$).
\item Add a \emph{validation portal}, a clique on $2p$ vertices, denoted by $V = \{v_1, v_2, \dots,  v_{2p}\}$,
and vertices and edges to identify it.
\item For every $i \in [\ell]$ and for every $p' \in \setrep(i)$,
add the edge $(a_i,  v_{p'})$.
\item For every critical pair $\langle c^{\circ}_q, c^{\star}_q \rangle$,
make $c^{\circ}_q$ adjacent to every vertex in $V$,  and
add every edge of the form $(c^{\star}_q, v_{p'})$ for $p' \in [2p] \setminus \setrep(\phi(q))$.
Note that the vertices in $V$ that are indexed using integers
in $\setrep(\phi(q))$ are \emph{not} adjacent with $c^{\star}_q$.
\end{itemize}
See \cref{fig:set-rep-core-tech} for an illustration.
Now, consider a critical pair 
$\langle c^{\circ}_q, c^{\star}_q \rangle$ and
suppose $i = \phi(q)$.
\begin{itemize}
\item 
By the construction, 
$N(a_i) \cap N(c^{\circ}_q) \neq \emptyset$,
whereas $N(a_i) \cap N(c^{\star}_q) = \emptyset$.
Hence, $a_i$ resolves the critical pair
$\langle c^{\circ}_q, c^{\star}_q \rangle$ 
as $d(a_i,  c_q^{\circ}) = 2$ and $d(a_j, c_q^{\star}) > 2$.

\item For any other vertex in $A$, say $a_j$,  $\setrep(j) \setminus \setrep(i)$ is a non-empty set.
So, there are paths from $a_j$ to  $c^{\circ}_q$ and $a_j$ to $c^{\star}_q$
    through vertices in $V$ with indices in $\setrep(j) \setminus \setrep(i)$.
    This implies that $d(a_j,  c_q^{\circ}) = d(a_j, c_q^{\star}) = 2$ and $a_j$ cannot 
    resolve the pair $\langle c^{\circ}_q, c^{\star}_q \rangle$.
\end{itemize}

\begin{figure}[t]
    \centering
        \includegraphics[scale=1.2]{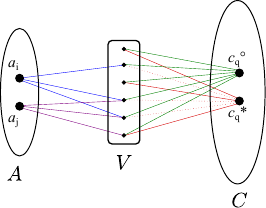}
    \caption{\textbf{Set Representation Gadget}.
    Let $\phi(q) = i$, i.e., only $a_i$ in $A$ can resolve
    the critical pair $\langle c^{\circ}_q, c^{\star}_q \rangle$.
    Let the vertices in $V$ be indexed from top to bottom and let $\setrep(i) = \{2, 4, 5\}$.
    By construction, the only vertices in $V$ that $c^{\star}_q$ is \emph{not} adjacent to are $v_2$, $v_4$, and $v_5$ (this is highlighted by red-dotted edges).
    Thus,  $\dist(a_i,  c^{\circ}_q) = 2$ and $\dist(a_i, c^{\star}_q) > 2$,  and hence,
    $a_i$ resolves $\langle c^{\circ}_q, c^{\star}_q \rangle$.
    For any other vertex in $A$, say $a_j$, $\setrep(j) \setminus \setrep(i)$ is non-empty, and thus,
    $a_j$ cannot resolve $\langle c^{\circ}_q, c^{\star}_q \rangle$.
    \label{fig:set-rep-core-tech}}
\end{figure}

\subsection{Sketch of the Lower Bound Proof for \textsc{Metric Dimension}}

With these tools in hand, 
we present an overview of the
reduction from \textsc{$3$-Partitioned-$3$-SAT} used to prove Theorem~\ref{thm:lower-bound-diam-tw}, which we restate 
here for convenience.

\begin{theoremnp*}[\textbf{\ref{thm:lower-bound-diam-tw}}]
Unless the \ETH\ fails, \mdfull does not admit an algorithm 
running in time $2^{f(\diam)^{o(\tw)}} \cdot n^{\OO(1)}$ for any computable function 
$f:\mathbb{N} \mapsto \mathbb{N}$.
\end{theoremnp*}


The reduction in the proof of Theorem~\ref{thm:lower-bound-diam-tw}
takes as input an instance $\psi$ of \textsc{3-Partitioned-3-SAT}
on $3n$ variables and returns $(G, k)$ as an instance of \textsc{Metric Dimension}
such that $\tw(G) = \calO(\log(n))$ and $\diam(G) = \calO(1)$.
In the following, we mention a crude
outline of the reduction, omitting 
some technical details.
For the formal proof, please refer
to Section~\ref{sec:lower-bound-diam-treewidth-MD}.

\subsubsection{Reduction}

\begin{itemize}
\item 
We rename the variables in $X^{\alpha}$
to $x^{\alpha}_{i}$ for $i \in [n]$.
For every variable $x^{\alpha}_i$, we add a critical pair 
$\langle x^{\alpha, \circ}_i, x^{\alpha, \star}_i \rangle$ of vertices.
We denote $X^{\alpha} = \{x^{\alpha, \circ}_i, x^{\alpha, \star}_i\  |\ i \in [n] \}$.

\item For each variable $x^{\alpha}_i$,
we add the vertices $t^{\alpha}_{2i}$ and $f^{\alpha}_{2i - 1}$. 
Let $A^{\alpha} = \{t^{\alpha}_{2i},  f^{\alpha}_{2i - 1} |\ i \in [n] \}$.

\item 
For every $i \in [n]$,  
we add the edges $(x^{\alpha, \circ}_i, t^{\alpha}_{2i})$ and $(x^{\alpha, \circ}_i, f^{\alpha}_{2i-1})$
which will ensure that
any resolving set contains at least one vertex in 
$\{t^{\alpha}_{2i},  f^{\alpha}_{2i - 1}, x^{\alpha, \circ}_i, x^{\alpha, \star}_i\}$ for every $i \in [n]$.

\item 
Let $p$ be the smallest integer such that $2n \leq \binom{2p}{p}$.
In particular, $p=\OO(\log n)$.
Define $\setrep: [2n] \mapsto \calF_{p}$ as in \cref{subsec:technique}.

\item We add a \emph{validation portal},
a clique on $2p$ vertices,  denoted by 
$V^{\alpha} = \{v^{\alpha}_1, v^{\alpha}_2,  \dots, v^{\alpha}_{2p} \}$.

\item 
For each $i\in [n]$, we add the edge $(t^{\alpha}_{2i},  v^{\alpha}_{p'})$
for every $p' \in \setrep(2i)$.
Similarly, for each $i\in [n]$,
we add the edge $(f^{\alpha}_{2i - 1},  v^{\alpha}_{p'})$
for every $p' \in \setrep(2i - 1)$.
\end{itemize}

We repeat the above steps to construct
$X^{\beta},  A^{\beta},  V^{\beta}$, 
$X^{\gamma},  A^{\gamma},  V^{\gamma}$.

\begin{itemize}
\item For every clause $C_q$ in $\psi$, we introduce a pair 
$\langle c^{\circ}_q,  c^{\star}_q \rangle$ of vertices.
Let $C$ be the collection of vertices in such pairs.

\item We add edges across $C$ and the portals as follows.
Consider a clause $C_q$ in $\psi$ and the corresponding
critical pair $\langle c^{\circ}_q,  c^{\star}_q \rangle$ in $C$.
Let $\delta \in \{\alpha, \beta, \gamma\}$.
As $\psi$ is an instance of \textsc{$3$-Partitioned-$3$-SAT}, at most one variable in $X^{\delta}$ appears in $C_q$, say $x^{\delta}_{i}$ for some $i \in [n]$.
We add all edges of the form 
$(v^{\delta}_{p'},  c^{\circ}_q)$ for every $p' \in [2p]$.
If $x^{\delta}_{i}$ appears as a positive literal in $C_q$, then
we add the edge $(v^{\delta}_{p'},  c^{\star}_q)$
for every $p' \in [2p] \setminus \setrep(2i)$
(which corresponds to $t^{\delta}_{2i}$).
If $x^{\delta}_{i}$ appears as a negative literal in $C_q$, then
we add the edge $(v^{\delta}_{p'},  c^{\star}_q)$
for every $p' \in [2p] \setminus \setrep(2i  - 1)$
(which corresponds to $f^{\delta}_{2i - 1}$).
Note that if $x^{\delta}_{i}$ appears as a positive (negative, respectively) literal in $C_q$,
then the vertices in $V^{\delta}$ whose indices are in $\setrep(2i)$ ($\setrep(2i - 1)$, respectively) 
are \emph{not adjacent} to $c^{\star}_q$.
If no variable in $X^{\delta}$ appears in $C_q$, 
then we make each vertex in $V^{\delta}$ adjacent to both $c^{\circ}_q$ and $c^{\star}_q$.
\end{itemize}
For all the sets mentioned above,
we add vertices and edges
to identify them as shown in \cref{fig:reduction-overview-diam-tw-tech} (for each critical pair, both vertices share the same bit-representation in their Set Identifying Gadget).
This concludes the construction of $G$.
The reduction returns $(G,  k)$
as an instance of \textsc{Metric Dimension} 
for some appropriate value of $k$.

\begin{figure}[t]
    \centering
        \includegraphics[scale=1.1]{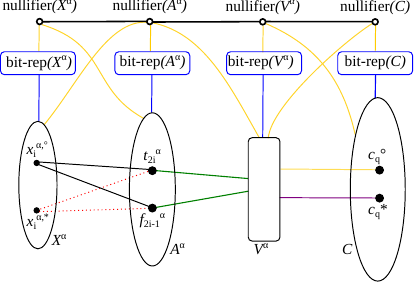}
    \caption{\textbf{Reduction for proof of Theorem~\ref{thm:lower-bound-diam-tw}.}
    For any $X\in \{X^{\alpha},A^{\alpha},V^{\alpha},C\}$, the yellow line between a vertex and a blue box containing $\bitrep(X)$ indicates that vertex is connected to every vertex in $\bitrep(X)\setminus \bits(X)$. The remainder of the yellow lines represent that vertex is connected to every vertex in the set the edge goes to.
    Green edges denote adjacencies with respect to $\setrep$,
    e.g.,  $t^{\alpha}_{2i}$ is adjacent to $v_j\in V^{\alpha}$ if $j\in \setrep(2i)$.
    Purple lines also indicate adjacencies with respect to $\setrep$, but in a complementary
    way, i.e., if $x_i\in c_q$, then, for every $p'\in [2p]\setminus \setrep(2i)$, we have $(v_{p'}^{\alpha}, c_q^{\star})\in E(G)$, and if $\overline{x}_i\in c_q$, then, for all $p'\in [2p]\setminus \setrep(2i-1)$, we have $(v_{p'}^{\alpha}, c_q^{\star})\in E(G)$.}
    \label{fig:reduction-overview-diam-tw-tech}
\end{figure}

\subsubsection{Correctness of the Reduction}

We give an informal description of the proof of correctness of the reverse direction
here.
Fix $\delta\in \{\alpha,\beta,\gamma\}$. For all $i\in [n]$, the only vertices that can resolve the critical pair 
$\langle x^{\delta, \circ}_i,  x^{\delta, \star}_i \rangle$
are the vertices in $\{x^{\delta, \circ}_i,  x^{\delta, \star}_i\} \cup \{t^{\delta}_{2i},  f^{\delta}_{2i - 1}\}$.
This fact and the budget $k$ ensure that any resolving set of $G$ contains 
exactly one vertex from 
$\{t^{\delta}_{2i},  f^{\delta}_{2i - 1}\} \cup \{x^{\delta, \circ}_i,  x^{\delta, \star}_i\}$ for all $i\in [n]$.
This naturally corresponds to
an assignment of the variable $x^{\delta}_i$ if a vertex from $\{t^{\delta}_{2i},  f^{\delta}_{2i - 1}\}$ is in the resolving set.
However, if a vertex from $\{x^{\delta, \circ}_i,  x^{\delta, \star}_i\}$ is in the resolving set, then we can 
see this as giving an arbitrary assignment to the variable $x^{\delta}_i$.
Suppose the clause $C_q$ contains the variable $x^{\delta}_{i}$ as a positive literal.
By the construction, 
every vertex in $V^{\delta}$ that is adjacent to $t^{\delta}_{2i}$ is \emph{not}
adjacent to $c_q^{\star}$.
However, $c_q^{\circ}$ is adjacent to every vertex in $V^{\delta}$.
Hence, $d(t^{\delta}_{2i},c_q^{\circ})=2$, whereas 
$d(t^{\delta}_{2i},c_q^{\star})>2$.
Thus, $t^{\delta}_{2i}$ resolves the critical pair $\langle c_q^{\circ}, c_q^{\star}\rangle$.
Consider any other vertex in $A^{\delta}$, say $t^{\delta}_{2j}$.
Since  $\setrep(2i)$ is not a subset of $\setrep(2j)$ (as both have the same cardinality),
there is at least one integer,  say $p'$, in $\setrep(2j) \setminus \setrep(2i)$. 
The vertex $v^{\delta}_{p'} \in V^{\delta}$ is adjacent to $t^{\delta}_{2j}$,  $c_q^{\circ}$,
and $c_q^{\star}$.
Hence,  $t^{\delta}_{2j}$ cannot resolve the critical pair $\langle c_q^{\circ}, c_q^{\star}\rangle$
as both these vertices are at distance $2$ from it.
Also, as $\psi$ is an instance of \textsc{$3$-Partitioned-$3$-SAT},
$C_q$ contains at most one variable in $X^{\delta}$, which is $x^{\delta}_i$ in this case.
This also helps to encode the fact that at most one
vertex from $A^{\delta}$ should be able to resolve
the critical pair $\langle c_q^{\circ}, c_q^{\star}\rangle$.
Since vertices in $X^{\delta}$ cannot resolve critical pairs
$\langle c^{\circ}_q,  c^{\star}_q \rangle$ in $C$, then
finding a resolving set in $G$ corresponds to finding a satisfying
assignment for $\psi$.

\subsubsection{Lower Bounds Obtained from the Reduction}

Let $Z=\{V^{\delta}\cup X^+~|~X \in \{X^{\delta}, A^{\delta},  V^{\delta}, C\}, \delta \in \{\alpha, \beta, \gamma\}\}$. Note that $|Z| = \calO(\log (n))$
and $G - Z$ is a collection of $P_3$'s and isolated vertices.
Hence,  $\tw(G)$, $\fvs(G)$, and $\td(G)$ are upper bounded by~$\calO(\log (n))$.
Also, $G$ has constant diameter.
Thus, if there is an algorithm for \textsc{Metric Dimension} 
that runs in time~$2^{f(\diam)^{o(\tw)}}$ (or $2^{f(\diam)^{o(\fvs)}}$ or $2^{f(\diam)^{o(td)}}$),
then there is an algorithm solving \textsc{$3$-Partitioned-$3$-SAT} in time~$2^{o(n)}$, contradicting the \ETH.

\section{Related Work}\label{sec:related-work}

\subsection{Double-Exponential Lower Bounds}\label{sec:related-work-double-exp}

It is long known that certain algorithmic tasks cannot be solved in
less than double-exponential time. In the realm of classical
complexity, this is captured by the complexity class
2-EXPTIME, with some problems that are complete for that class, for
example, \textsc{Presburger Arithmetic}~\cite{10.1007/978-3-7091-9459-1_5},
the \textsc{Asynchronous Reactive Module Synthesis} arising from linear temporal
logic~\cite{DBLP:conf/icalp/PnueliR89} or \textsc{Planning with Partial Observability}~\cite{DBLP:conf/aips/Rintanen04a}.

In the realm of parameterized complexity as well, it is known that
certain problems have a fixed-parameter-tractable running time that
requires a double-exponential dependency in the parameter, such as
\textsc{Edge Clique Cover}~\cite{DBLP:journals/siamcomp/CyganPP16},
\textsc{Multi-Team Formation}~\cite{DBLP:conf/fsttcs/Lokshtanov0SX21},
\textsc{Modal
  Satisfiability}~\cite{DBLP:journals/algorithmica/AchilleosLM12}, and
\textsc{Distinct Vectors}~\cite{DBLP:journals/dmtcs/PilipczukS20}.

When it comes to structural parameterized algorithms, treewidth is
one of the main success stories of the field, as many problems are
\FPT\ when parameterized by the treewidth of the input. This is
witnessed by the famous theorem of Courcelle that states that all problems
expressible in MSOL can be solved by a linear-time \FPT\
algorithm~\cite{Courcelle90}, and by a similar result of Chen for the
{\sc Quantified SAT} and {\sc Quantified CSP}
problems~\cite{DBLP:conf/ecai/Chen04} (here, we consider the treewidth
of the primal graph of the input relational structure). Unfortunately,
the dependence in treewidth is notoriously ``galactic'': a tower of
exponentials whose height depends on the number of quantifier
alternations in the MSOL formula, and in the SAT instance,
respectively. Moreover, Chen~\cite{DBLP:conf/ecai/Chen04} showed that
the height of this tower is equal to the number of quantifier
alternations in the {\sc Quantified SAT} instance.

However, note that {\sc Quantified SAT} and \textsc{MSO Model
  Checking on Trees} (even for FO formulas) are PSPACE-complete
problems~\cite{arora2009computational,DBLP:journals/jacm/FrickG01}. There
are few natural problems that have been shown to admit (at least)
double-exponential lower bounds with respect to treewidth. The
$\Pi^p_2$-complete \textsc{$k$-Choosability} problem and the
$\Sigma_3^p$-complete \textsc{$k$-Choosability Deletion} problem admit
a double-exponential and a triple-exponential lower bound in
treewidth~\cite{MM16}, respectively. Double-exponential lower bounds
for the $\Sigma_2^p$-complete $\exists \forall${\sc SAT} and
$\Pi^p_2$-complete $\forall \exists${\sc SAT} problems were shown
in~\cite{LM17}. 
Recently, the $\Sigma^2_p$-complete problems \textsc{Cycle HitPack} and \textsc{$H$-HitPack}, for a fixed graph $H$, were shown to admit tight algorithms that are double-exponential in the treewidth~\cite{DBLP:journals/corr/abs-2402-14927}.
Further, the $\Sigma^2_p$-complete problem \textsc{Core Stability} was shown to admit a tight double-exponential lower bound in the treewidth, even on graphs of bounded degree~\cite{HKL24}.
Lastly, the $\#$\NP-complete counting problem
\textsc{Projected Model Counting} admits a double-exponential lower
bound in
$\tw$~\cite{DBLP:journals/ai/FichteHMTW23,DBLP:conf/sat/FichteHMW18}. Similar
lower bounds were obtained in~\cite{DBLP:conf/sat/LampisMM18} for problems from artificial
intelligence (abstract argumentation, abduction, circumscription, and
the computation of minimal unsatisfiable sets in unsatisfiable
formulas), with these problems also lying in the second level of the polynomial
hierarchy.

With respect to the vertex cover parameter (of the primal graph of the
input relational structure), the $\Sigma^p_2$-complete problem
\textsc{$\exists \forall$-CSP} also requires a double-exponential
dependency in its running time~\cite{LM17}.

A double-exponential lower bound is also known for \textsc{Coloring}
with respect to the smaller parameter
\emph{cliquewidth}~\cite{DBLP:journals/talg/FominGLSZ19}. However, in contrast
with the aforementioned ones, this problem is only \XP\ and not \FPT\ parameterized by the cliquewidth.
As another example, 
we know that \textsc{ILP Feasibility} admits a double-exponential
lower bound when parameterized by the \emph{dual treedepth} of the 
input matrix~\cite{DBLP:journals/toct/KnopPW20}.

%

\subsection{Metric Graph Problems}\label{sec:related-work-metric}

Metric graph problems are defined using either distance values or
shortest paths in the graph. Metric-based graph problems are
ubiquitous in computer science, for example, the classic
\textsc{(Single-Source) Shortest Path}, \textsc{(Graphic) Traveling
  Salesperson Problem} or \textsc{Steiner Tree} fall into this
category. Those are fundamental problems, often stemming from
applications in network design, for which a lot of algorithmic research
has been done. Among these, metric-based graph packing and
covering problems such as, for example, \textsc{Distance
  Domination}~\cite{JKST19} or \textsc{Scattered Set}~\cite{KLP22},
have recently gained a lot of attention. Their non-local nature leads
to non-trivial algorithmic properties that differ from most classic
graph problems with a more local nature. This is the case in
particular for treewidth-based algorithms. In this paper, we focus on
three problems arising from network design (\gsfull) and network
monitoring (\mdfull and \smdfull). These problems have 
far-reaching applications, as exemplified by, e.g., the recent work~\cite{BDM23}, in which
it was shown that enumerating minimal solution sets for the metric dimension problem in (general) graphs and the geodetic set
problem in split graphs is equivalent to enumerating minimal transversals of hypergraphs, arguably
the most important open problem in algorithmic enumeration.

\subparagraph*{\mdfull.}
\mdfull was introduced in the 1970s independently by Harary and
Melter~\cite{HM76} and Slater~\cite{Slater75} as a network monitoring
problem. \mdfull and its variants (see, e.g.,~\cite{BGLM08, BMM+19,
  ERY15, FGLTY22, harary1993, KarpovskyCL98, sebo04, Slater87}) are
very well-studied and have numerous applications such as in graph
isomorphism testing~\cite{B80}, network discovery~\cite{BEE+06}, image
processing~\cite{MT84}, chemistry~\cite{J93}, graph
reconstruction~\cite{DBLP:conf/esa/Mathieu021} or
genomics~\cite{TL19}. In fact, \mdfull was first shown to be
\NP-complete in general graphs in Garey and Johnson's
book~\cite{GJ79}, and this was later extended to
unit disk graphs~\cite{HW13}, split graphs, bipartite graphs, co-bipartite graphs, 
and line graphs of bipartite graphs~\cite{ELW15}, bounded-degree planar 
graphs~\cite{DiazPSL17}, and interval and permutation graphs of 
diameter~$2$~\cite{FoucaudMNPV17b}.
On the tractable side, \mdfull admits linear-time algorithms on 
trees~\cite{Slater75}, cographs~\cite{ELW15}, {chain graphs~\cite{FHHMS15}}, cactus block 
graphs~\cite{HEW16}, and bipartite distance-hereditary graphs~\cite{M22}, and a polynomial-time algorithm on
outerplanar graphs~\cite{DiazPSL17}.

Due to the \NP-hardness results, the focus has now
shifted to studying its parameterized complexity, in search of
tractable instances. In a seminal paper, it was proven that \mdfull is
\W[2]-hard parameterized by the solution size $k$, even in subcubic
bipartite graphs~\cite{HartungN13}. This paper was the driving
motivation behind the subsequent meticulous study of \mdfull under
structural parameterizations. Several different parameterizations have
been studied for this problem, that we now elaborate on (see also~\cite[Figure~1]{GKIST23}).

In terms of structural parameterizions for \mdfull, 
through careful design, kernelization, and/or meta-results, 
it was proven that there is an \XP\ algorithm parameterized by the 
feedback edge set number in~\cite{ELW15}, and \FPT\ algorithms 
parameterized by the max leaf number in~\cite{E15}, the modular-width 
and the treelength plus the maximum degree in~\cite{BelmonteFGR17}, 
the treedepth and the clique-width plus the diameter in~\cite{GHK22}, 
and the distance to cluster (co-cluster, respectively) in~\cite{GKIST23}. 
Recently, an \FPT\ algorithm parameterized by the treewidth in chordal 
graphs was given in~\cite{BDP23}. 
On the negative side, \mdfull is \W[1]-hard parameterized by the pathwidth 
on graphs of constant degree~\cite{BP21}, para-\NP-hard parameterized 
by the pathwidth~\cite{LM21}, and \W[1]-hard parameterized by 
the combined parameter feedback vertex set number plus pathwidth~\cite{GKIST23}. Lastly, it is not computable (unless the \ETH\ fails) in time $2^{o(n)}$ on bipartite graphs, and in time $2^{o(\sqrt{n})}$ on planar bipartite graphs~\cite{BIT20}.

\subparagraph*{\smdfull.}

Albeit less well-studied than \mdfull, the \smdfull problem, which was
introduced by Seb\H{o} and Tannier in 2004~\cite{sebo04} as a
strengthening of \mdfull, enjoys interesting applications in
coin-weighing problems and other areas of algorithms and
combinatorics. Here, we are given a graph $G$ and an integer $k$, and
we are seeking a solution set $S$ of size at most $k$, called a
\emph{strong resolving set}, such that, for any pair of vertices
$u,v\in V(G)$, there exists a vertex $w\in S$ such that either $u$
lies on some shortest path between $v$ and $w$, or $v$ lies on some
shortest path between $u$ and $w$. The size of the smallest such set
$S$ is called the \emph{strong metric dimension} of $G$. It is clear
from the definition that any strong resolving set is also a resolving
set.

In the seminal paper
introducing the problem, it was used to design an efficient algorithm
for the graph problem \textsc{Connected Join Existence}~\cite{sebo04}.
Interestingly, it was shown in~\cite{OP07} that the problem on an
instance $(G,k)$ can be reduced (in polynomial time) to an instance
$(G',k)$ of \textsc{Vertex Cover} where $V(G)=V(G')$ and the edges of
$G'$ join a set of suitably defined critical pairs (see
also~\cite{DBLP:journals/dam/KuziakPRY18} for further studies of this
reduction). Consequently, algorithmic results known for \textsc{Vertex
  Cover} can be applied to \smdfull: in particular, the problem is
\FPT\ when parameterized by the solution size. On the other hand, it was
shown in~\cite{DBLP:journals/dam/DasGuptaM17} that many hardness
results known for \textsc{Vertex Cover} can also be transferred to
\smdfull.

\subparagraph*{\gsfull.}

\gsfull was introduced in 1993 by Harary, Loukakis, and Tsouros
in~\cite{harary1993}. It can be seen as a network design problem,
where one seeks to determine the optimal locations of public
transportation hubs in a road network, while minimizing the total
number of such hubs~\cite{floCALDAM20}.
Other applications are mentioned in~\cite{ekim2012}. More generally,
\gsfull is part of the area of geodesic convexity in graphs: see, e.g., the paper~\cite{farber1986} or the book~\cite{bookGC}.

As is often the case with metric-based problems, \gsfull is
computationally hard, even for very structured graphs. Its \NP-hardness
was claimed in the seminal paper~\cite{harary1993} (see~\cite{JCMCC96}
for the earliest explicit proofs). This is known to hold even for
graphs that belong to various structured input graph classes, such as
interval graphs~\cite{floISAAC20}, co-bipartite
graphs~\cite{ekim2012}, line graphs~\cite{floCALDAM20}, graphs of
diameter~2~\cite{floCALDAM20}, and subcubic (planar bipartite) grid
graphs of arbitrarily large girth~\cite{floISAAC20,DBLP:journals/tcs/ChakrabortyGR23} (see
also~\cite{atici2002,bueno2018,dourado2008,dourado2010} for various
earlier hardness results).  \gsfull can be solved in polynomial time
on split graphs~\cite{dourado2010,JCMCC96} and, more generally,
well-partitioned chordal graphs~\cite{wellpart}, outerplanar
graphs~\cite{mezzini2018}, ptolemaic graphs~\cite{farber1986},
cographs~\cite{dourado2010} and, more generally, distance-hereditary
graphs~\cite{dh}, block-cactus graphs~\cite{ekim2012}, solid grid graphs~\cite{floISAAC20,DBLP:journals/tcs/ChakrabortyGR23}, and proper
interval graphs~\cite{ekim2012}.

The parameterized complexity of \gsfull was first addressed by
Kellerhals and Koana in~\cite{KK22}. They observed that the reduction
from~\cite{dourado2010} implies that the problem is \W[2]-hard when 
parameterized by the solution size (even for chordal
bipartite graphs). The above-mentioned hardness results for structural
graph classes motivated the authors of~\cite{KK22} to investigate
structural parameterizations of \gsfull. They proved the problem to
be W[1]-hard for the parameters solution size, feedback vertex set
number, and pathwidth, combined~\cite{KK22}. On the positive side,
they showed that \gsfull is \FPT\ for the parameters treedepth,
modular-width (more generally, clique-width plus diameter), and
feedback edge set number~\cite{KK22}. The problem is also \FPT\ on
chordal graphs when parameterized by the treewidth~\cite{floISAAC20}.




The approximability of \gsfull was also studied. Its
minimization variant is \NP-hard to approximate within a factor of
$o(\log n)$, even for diameter~2 graphs~\cite{floCALDAM20} and
subcubic bipartite graphs of arbitrarily large girth~\cite{DIT21}. It
can be approximated in polynomial time within a factor of $n^{1/3}\log
n$~\cite{floCALDAM20} (but the best possible approximation factor is
unknown).



\section{Preliminaries}
\label{sec:preliminaries}

In this paper, all logarithms are to the base $2$. For an integer $a$, we let $[a] = \{1,\ldots,a\}$.

\subparagraph*{Graph theory.} 
We use standard graph-theoretic notation and refer the reader 
to~\cite{D12} for any undefined notation. 
For an undirected graph $G$, the sets $V(G)$ and $E(G)$ denote its 
set of vertices and edges, respectively.
Two vertices $u,v\in V(G)$ are {\it adjacent} or {\it neighbors} if 
$(u, v)\in E(G)$. 
The {\it open neighborhood} of a vertex $u\in V(G)$, denoted by $N(u):=N_G(u)$, is the set of vertices that are neighbors of $u$. 
The {\it closed neighborhood} of a vertex $u\in V(G)$ is denoted by $N[u]:=N_G[u]:=N_G(u)\cup \{u\}$.
For any $X \subseteq V(G)$ and $u\in V(G)$, $N_X(u) = N_G(u) \cap X$. 
Any two vertices $u,v\in V(G)$ are {\it true twins} if $N[u] = N[v]$, and are {\it false twins} if $N(u) = N(v)$. 
Observe that if $u$ and $v$ are true twins, then $(u,v) \in E(G)$, but if they are only false twins, then $(u,v) \not \in E(G)$.
For a subset $S$ of $V(G)$, we say that the vertices in $S$ are true (false, respectively) twins if, for any $u,v\in S$, $u$ and $v$ are true (false, respectively) twins.
The {\it distance} between two vertices $u,v\in V(G)$ in $G$, denoted by $d(u,v):=d_G(u,v)$, is the length of a $(u,v)$-shortest path in $G$. 
For a subset $S$ of $V(G)$, we define $N[S] = \bigcup_{v \in S} N[v]$ and $N(S) = N[S] \setminus S$.
For a subset $S$ of $V(G)$, we denote the graph obtained by deleting $S$ from $G$ by $G - S$.
We denote the subgraph of $G$ induced on the set $S$ by $G[S]$.
For a graph $G$, a set $X \subseteq V(G)$ is a \emph{vertex cover} of $G$ if $V(G) \setminus X$ is an independent set. We denote by $\vc(G)$  the size of a minimum vertex cover in $G$. When $G$ is clear from the context, we simply say $\vc$.
For a graph $G$, a set $X \subseteq V(G)$ is a \emph{feedback vertex set} of $G$ if $V(G) \setminus X$ is an acyclic graph. We define the notation of \emph{the feedback vertex set number} in the analogous way.

\subparagraph*{Tree decompositions.} A \emph{tree decomposition} of a graph $G$ is a pair $(T,\mathcal{X})$, where $T$ is a tree and $\mathcal{X}:=\{X_i:i\in V(T)\}$ is a collection of subsets of $V(G)$, called \emph{bags}, satisfying the following conditions: (i)~$\bigcup_{i\in V(T)}X_i=V(G)$, (ii)~for every edge $(u,v)\in E(G)$, there is a bag that contains both $u$ and $v$, and (iii)~for every vertex $v\in V(G)$, the set of nodes of $T$ whose bags contain $v$ induces a (connected) subtree of $T$.

The maximum size of a bag minus one is called the \emph{width} of $T$. The minimum width of a tree decomposition of $G$ is the \emph{treewidth} of $G$.

We consider a {\em rooted} tree decomposition by fixing a root of $T$ and orienting the tree edges from the root toward the leaves. A rooted tree decomposition is \emph{nice} (see~\cite{niceTW}) if each node $i$ of $T$ has at most two children and falls into one of the four types: 
\begin{itemize}
\item {\em Join} node: $i$ has exactly two children $i_1$ and $i_2$ with $X_i=X_{i_1}=X_{i_2}$.
\item {\em Introduce} node: $i$ has a unique child $i'$ with $X_{i'}=X_{i}\setminus \{v\}$, where $v\in V(G) \setminus X_{i'}$.
\item {\em Forget} node: $i$ has a unique child $i'$ with $X_i=X_{i'}\setminus \{v\}$, where $v\in X_{i'}$.
\item {\em Leaf} node: $i$ is a leaf of $T$ with $|X_i|=1$.
\end{itemize}
For a node $i$ of $T$, we denote by $T_i$ the subtree of $T$ rooted at $i$, and by $G_i$, the subgraph of $G$ induced by the vertices of the bags in $T_i$.

For a graph $G$, a set $S \subseteq V(G)$ is a \emph{separator} for two non-adjacent vertices $x,y \in V(G)$ if $x$ and $y$ belong to two different connected components of $G-X$.

\subparagraph*{Parameterized Complexity.}
An instance of a parameterized problem $\Pi$ comprises an input $I$, which is an input of the classical instance of the problem, and an integer $\ell$, which is called the parameter.
A problem $\Pi$ is said to be \emph{fixed-parameter tractable} or in \FPT\ if given an instance $(I,\ell)$ of $\Pi$, we can decide whether or not $(I,\ell)$ is a \yes-instance of $\Pi$ in  time $f(\ell)\cdot |I|^{\OO(1)}$,
for some computable function $f$ whose value depends only on $\ell$. 

A {\em kernelization} algorithm for $\Pi$ is a polynomial-time algorithm that takes as input an instance $(I,\ell)$ of $\Pi$ and returns an {\em equivalent} instance $(I',\ell')$ of $\Pi$ with $|I'|, \ell' \leq f(\ell)$, where $f$ is a function that depends only on the initial parameter $\ell$. If such an algorithm exists for $\Pi$, we say that $\Pi$ admits a kernel of {\em size} $f(\ell)$. If $f$ is a polynomial or exponential function of $\ell$, we say that $\Pi$ admits a polynomial or exponential kernel, respectively. 
If $\Pi$ is a graph problem, then $I$ contains a graph, say $G$, and $I'$ contains a graph, say $G'$. In this case, we say that $\Pi$ admits a kernel with $f(\ell)$ vertices if the number of vertices of $G'$ is at most $f(\ell)$. 

It is typical to describe a kernelization algorithm as a series of reduction rules.
A \emph{reduction rule} is a polynomial time algorithm that takes as an input an instance of a problem and outputs another (usually reduced) instance.
A reduction rule said to be \emph{applicable} on an instance if the output instance is different from the input instance.
A reduction rule is \emph{safe} if the input instance is a \yes-instance if and only if the output instance is a \yes-instance.
For more on parameterized complexity and related terminologies, we refer the reader to the recent book by Cygan et al.~\cite{DBLP:books/sp/CyganFKLMPPS15}.

\subparagraph*{(Strong) Metric Dimension.}
A subset of vertices $S\subseteq V(G)$ {\it resolves} a pair of vertices $u,v\in V(G)$ if there exists a vertex $w \in S$ such that $d(w,u)\neq d(w,v)$.
A subset of vertices $S\subseteq V(G)$ is a {\it resolving set} of $G$ if it resolves all pairs of vertices $u,v\in V(G)$.
A vertex $u\in V(G)$ is {\it distinguished} by a subset of vertices $S\subseteq V(G)$ if, for any $v\in V(G)\setminus \{u\}$, there exists a vertex $w\in S$ such that $d(w,u)\neq d(w,v)$.
For an ordered subset of vertices $S=\{s_1,\dots,s_k\}\subseteq V(G)$ and a single vertex $u\in V(G)$, the {\it distance vector} of $S$ with respect to $u$ is $r(S|u):=(d(s_1,u),\dots,d(s_k,u))$.
The next observation is used throughout the paper.

\begin{observation}\label{obs:twins}
Let $G$ be a graph. 
For any (true or false) twins $u,v \in V(G)$ and any $w \in V(G) \setminus \{u,v\}$, $d(u,w) = d(v,w)$, 
and so, for any resolving set $S$ of $G$, $S\cap \{u,v\} \neq \emptyset$.
\end{observation}

\begin{proof}
As $w\in V(G) \setminus \{u,v\}$, and $u$ and $v$ are (true or false) twins, the shortest $(u,w)$- and $(v,w)$-paths contain a vertex of $N:=N(u)\setminus\{v\}=N(v)\setminus\{u\}$, and $d(u,w)=d(v,w)$. Hence, any resolving set $S$ of $G$ contains at least one of $u$ and $v$.
\end{proof}

A vertex $s \in V(G)$ \emph{strongly resolves} a pair of vertices $u,v \in V(G)$ if there exists a shortest path from $u$ to $s$ containing $v$, or a shortest path from $v$ to $s$ containing $u$. A subset $S \subseteq V(G)$ is a \emph{strong resolving set} is every pair of vertices in $V(G)$ is strongly resolved by a vertex in $S$.

\subparagraph*{Geodetic Set.}
A subset $S \subseteq V(G)$ is a \emph{geodetic set} if for every $u \in V(G)$, the following holds: there exist $s_1,s_2 \in S$ such that $u$ lies on a shortest path from $s_1$ to $s_2$.
The following simple observation is used throughout the paper. Recall that a vertex is \emph{simplicial} if its neighborhood forms a clique.

\begin{observation}[\cite{CHZ02}]\label{obs:simplicial}
If a graph $G$ contains a simplicial vertex $v$, then $v$ belongs to any geodetic set of $G$.
\end{observation}
\begin{proof}
Observe that $v$ does not belong to any shortest path between any pair $x,y$ of vertices (both distinct from $v$).
\end{proof}

This gives the following observation as an immediate corollary.

\begin{observation}\label{obs:pendant}
  If a graph $G$ contains a degree-$1$ vertex $v$, then $v$ belongs to any geodetic set of $G$.
\end{observation}

\subparagraph*{\textsc{3-Partitioned-3-SAT}.}
Most of our lower bound proofs consist of reductions from the \textsc{3-Partitioned-3-SAT} problem, a version of \textsc{3-SAT} introduced
in~\cite{DBLP:journals/corr/abs-2302-09604} and defined as follows.

\defproblem{\textsc{3-Partitioned-3-SAT}}{A formula $\psi$ in $3$-\textsc{CNF} form,
together with a partition of the set of its variables into three
disjoint sets $X^{\alpha}$, $X^{\beta}$, $X^{\gamma}$,
with $|X^{\alpha}| = |X^{\beta}| = |X^{\gamma}| = n$, and
such that no clause contains more than one variable from each of
$X^{\alpha},  X^{\beta}$, and $X^{\gamma}$.}{Determine whether $\psi$ is satisfiable.}

The authors of~\cite{DBLP:journals/corr/abs-2302-09604} also proved the following.

\begin{proposition}[{\cite[Theorem 3]{DBLP:journals/corr/abs-2302-09604}}]
\label{prop:3-SAT-to-3-Partition-3-SAT-lampis}
Unless the \ETH\ fails,  \textsc{3-Partitioned-3-SAT} does not admit
an algorithm running in time $2^{o(n)}$.
\end{proposition}

We will also use the following restricted
version of the above problem.

\defproblem{\textsc{Exact-3-Partitioned-3-SAT}}{A formula $\psi$ in $3$-\textsc{CNF} form,
together with a partition of the set of its variables into three
disjoint sets $X^{\alpha}$, $X^{\beta}$, $X^{\gamma}$,
with $|X^{\alpha}| = |X^{\beta}| = |X^{\gamma}| = n$, and
every clause contains exactly one variable from each of
$X^{\alpha},  X^{\beta}$, and $X^{\gamma}$.}{Determine whether $\psi$ is satisfiable.}

For completeness, 
we repeat the polynomial-time reduction in~\cite{DBLP:journals/corr/abs-2302-09604} from \textsc{$3$-SAT}
to \textsc{$3$-Partitioned-$3$-SAT} that increases the number of variables and
clauses by a constant factor.
Importantly, we make a simple change 
to adapt the proof for 
\textsc{Exact-$3$-Partitioned-$3$-SAT}.

\begin{proposition}{\cite[Theorem 3]{DBLP:journals/corr/abs-2302-09604}}
\label{prop:3-SAT-to-3-Partition-3-SAT}
Unless the \ETH\ fails, \textsc{3-Partitioned-3-SAT} or 
\textsc{Exact-3-Partitioned-3-SAT}
does not admit
an algorithm running in time $2^{o(n)}$.
\end{proposition}
\begin{proof}
Let $\psi$ be a \textsc{$3$-SAT} formula of 
$m$ clauses and $n$ variables.
We can assume, without loss of generality,
that every variable is used in some clause
and every clause contains at least two literals.
Suppose $X = \{x_1,\dots, x_n\}$
is the set of variables in $\psi$.
We construct an equivalent instance $\psi'$
of \textsc{$3$-Partitioned-$3$-SAT} as follows:
\begin{itemize}
\item For every $i \in [n]$, we introduce
three variables $x^{\alpha}_i$, $x^{\beta}_i$,
and $x^{\gamma}_i$, corresponding to 
the variable $x_i$, to $\psi'$.
\item For every clause, e.g., 
$C = (x_i \lor \neg x_j \lor x_{\ell})$, we
introduce the clause
$C' = (x^{\alpha}_i \lor \neg x^{\beta}_j \lor x^{\gamma}_{\ell})$ 
to $\psi'$.
In an analogous way,
for every clause $C = (x_i \lor x_j)$, we
introduce $C' = (x^{\alpha}_i \lor x^{\beta}_j)$.
\item 
For every $i \in [n]$, we 
introduce the clauses $(\neg x_i^{\alpha} \lor x_i^{\beta})$, $(\neg x_i^{\beta} \lor x^{\gamma}_{i})$, and
$(x^{\alpha}_i \lor \neg x^{\gamma}_i)$.
\end{itemize}

Define $X^{\alpha} = \{x^{\alpha}_i \mid i \in [n]\}$, and 
$X^{\beta}$, $X^{\gamma}$ in the analogous way.
Note that $\psi'$ is a valid instance of 
\textsc{$3$-Partitioned-$3$-SAT}
as its variable set is divided into three
equal parts, $X^{\alpha}$, $X^{\beta}$, and 
$X^{\gamma}$, and each clause contains 
at most one variable from each 
of these parts.
To see that $\psi$ and $\psi'$ are equivalent
instances, consider a satisfying assignment 
$\pi: X \mapsto \{\true, \false\}$ for $\psi$.
Consider the assignment 
$\pi': X^{\alpha} \cup X^{\beta} \cup X^{\gamma} \mapsto \{\true, \false\}$
defined as follows: 
$\pi'(x^{\alpha}_i) = \pi'(x^{\beta}_i) =
\pi'(x^{\gamma}_i) = \pi(x_i)$
for all $i \in [n]$.
It is easy to verify
that the assignment $\pi'$ is a satisfying 
assignment for $\psi'$.
In the reverse direction,
consider a satisfying assignment
$\pi': X^{\alpha} \cup X^{\beta} \cup X^{\gamma} \mapsto \{\true, \false\}$ for $\psi'$.
Note that the clauses added in
the third step above are all
satisfied if and only if the variables 
$x^{\alpha}_i$, $x^{\beta}_i$, and $x^{\gamma}_i$
share the same assignment, i.e., either all
are $\true$ or all are $\false$.
Hence, $\pi'(x^{\alpha}_i) = \pi'(x^{\beta}_i) = \pi'(x^{\gamma}_i)$.
It is easy to see that $\pi: X \mapsto \{\true, \false\}$,
where $\pi(x_i) = \pi'(x^{\alpha}_{i})$ for
all $i \in [n]$, is a satisfying
assignment for $\psi$.
As the number of variables in $\psi'$ is at 
most $3$ times the number of variables
in $\psi$, 
if \textsc{3-Partitioned-3-SAT} admits
an algorithm running in time $2^{o(n)}$, so
does \textsc{3-SAT}, which contradicts
the \ETH.
This completes the first part of the 
proposition.

To prove the second part, we 
add the following steps to the above reduction.
\begin{itemize}
\item We add the variables $x^{\alpha}_0$,
$x^{\beta}_0$, $x^{\gamma}_0$ to 
$X^{\alpha}$, $X^{\beta}$, and $X^{\gamma}$,
respectively.
\item We add the following clauses:
\begin{itemize}
\item $(\neg x_0^{\alpha} \lor \neg x_0^{\beta} \lor \neg x_0^{\gamma})$,
\item $(x_0^{\alpha} \lor \neg x_0^{\beta} \lor \neg x_0^{\gamma})$,
$(\neg x_0^{\alpha} \lor x_0^{\beta} \lor \neg x_0^{\gamma})$,
$(\neg x_0^{\alpha} \lor \neg x_0^{\beta} \lor x_0^{\gamma})$,
and
\item $(\neg x_0^{\alpha} \lor x_0^{\beta} \lor  x_0^{\gamma})$,
$(x_0^{\alpha} \lor \neg x_0^{\beta} \lor x_0^{\gamma})$, $(x_0^{\alpha} \lor x_0^{\beta} \lor \neg x_0^{\gamma})$.
\item For every clause that 
has only two literals, we
add $x^{\gamma}_{0}$.
\end{itemize}
\end{itemize}
By the construction above,
each clause that had only two
literals contained literals
corresponding to variables in
$X^{\alpha}$ and $X^{\beta}$.
Thus, $\psi'$ is a valid
instance of 
\textsc{Exact-3-Partitioned-3-SAT}.
Now, it suffices to note that any satisfying assignment $\pi'$ for $\psi'$ must set $x_0^{\alpha}$, $x_0^{\beta}$, and $x_0^{\gamma}$ to \false. 
Then, the other arguments are similar to those
mentioned in the above paragraph.
\end{proof}
\section{\mdfull: Lower Bound Regarding Diameter plus Treewidth}
\label{sec:lower-bound-diam-treewidth-MD}


The aim of this section is to prove
the following theorem.

\begin{theorem}
\label{thm:lower-bound-diam-tw}
Unless the \ETH\ fails, \mdfull does not admit an algorithm 
running in time $2^{f(\diam)^{o(\tw)}} \cdot n^{\OO(1)}$ for any computable function 
$f:\mathbb{N} \mapsto \mathbb{N}$.
\end{theorem}

To this end, we present a reduction from \textsc{$3$-Partitioned-$3$-SAT}
to \textsc{Metric Dimension}.
The reduction takes as input an instance $\psi$ of \textsc{$3$-Partitioned-$3$-SAT}
on $3n$ variables (see \Cref{sec:preliminaries} for a definition of this problem), and returns $(G, k)$ as an instance of \textsc{Metric Dimension}
such that $\tw(G) = \calO(\log(n))$ and $\diam(G) = \calO(1)$.
Before presenting the reduction, we first introduce some preliminary tools.

\subsection{Preliminary Tools}
\label{subsec:prelim-3-Par-3-SAT-Met-Dim-diam-tw}

\begin{figure}[t]
    \centering
        \includegraphics[scale=.91]{./images/set-identifying-gadget2.pdf} \hspace{0.7cm}
        \includegraphics[scale=.91]{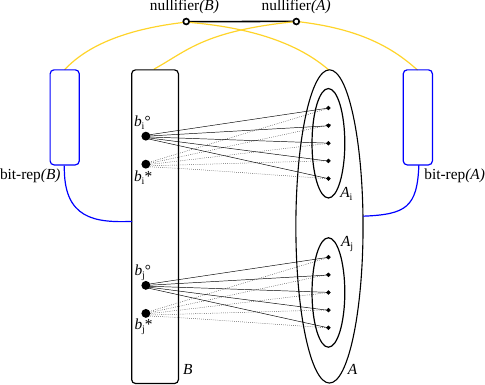}
    \caption{ \textbf{Set Identifying Gadget (left)}. The blue box represents $\bitrep(X)$ and the yellow lines represent that every vertex in $\bitrep(X)\setminus \bits(X)$ is adjacent to $\bitrepnullifier(X)$, $\bitrepnullifier(X)$ is adjacent to every vertex in $N(X)$, and $y_{\star}$ is adjacent to every vertex in $X$. Note that $G'$ is not necessarily restricted to the graph induced by the vertices in $X\cup N(X)$.
    \textbf{Vertex Selector Gadget (right)}.
    For $X \in \{B,  A\}$,  the blue box represents $\bitrep(X)$, the blue link represents the connection
    with respect to the binary representation, and the yellow line represents that $\bitrepnullifier(X)$
    is connected to every vertex in $\bitrep(X)\setminus \bits(X)$.
    The dotted lines highlight the absence of edges.}
    \label{fig:set-identifying-gadget}
    \label{fig:vertex-selector-gadget}
\end{figure}

\subsubsection{Set Identifying Gadget}
\label{subsubsec:gadget-set-MD}

Suppose that we are given a graph $G'$ and a subset $X\subseteq V(G')$ of its vertices.
Further, suppose that we want to add a vertex set $X^+$ to $G'$ in order to obtain a new graph $G$ with the following properties. We want that each vertex in $X \cup X^+$ will be distinguished by vertices in $X^+$ that must be in any resolving set $S$ of~$G$, and no vertex in $X^+$ can resolve any ``critical pair'' of vertices in $V(G)$ (critical pairs will be defined in the next subsection).
 
We refer to the graph induced by the vertices of $X^+$, along with the edges connecting $X^+$ to $G'$, as the Set Identifying Gadget for the set $X$.

Given a graph $G'$ and a non-empty subset $X\subseteq V(G')$ of its vertices, to construct such a graph $G$, we add vertices and edges to $G'$ as follows:
\begin{itemize}
\item The vertex set $X^+$ that we are aiming to add is the union of a set $\bitrep(X)$ and a special vertex denoted by $\bitrepnullifier(X)$.
\item First, let $X=\{x_i\mid i\in [|X|]\}$, and set $q := \lceil \log(|X|+ 2) \rceil+1$.
We select this value for $q$ to \emph{(1)} uniquely represent each integer in $[|X|]$ by its bit-representation in binary (note that we start from $1$ and not $0$), \emph{(2)} ensure that the only vertex whose bit-representation contains all $1$'s is $\bitrepnullifier(X)$, and \emph{(3)} reserve one spot for an additional vertex $y_{\star}$.
\item For every $i \in [q]$, add three vertices $y^a_i,  y_i, y^b_i$, and add the path $(y^a_i, y_i, y^b_i)$.
\item Add three vertices $y^a_{\star},  y_{\star}, y^b_{\star}$, and add the path $(y^a_{\star}, y_{\star}, y^b_{\star})$. 
Add all the edges to make $\{y_i\mid\ i \in [q] \}\cup\{y_{\star}\}$ into a clique.
Make $y_{\star}$ adjacent to each vertex $v\in X$.
We denote $\bitrep(X)=\{y_i, y^a_i, y^b_i\mid i\in [q]\}\cup \{y_{\star}, y^a_{\star}, y^b_{\star}\}$ and its subset $\bits(X)=\{y^a_i, y^b_i\mid i\in [q]\}\cup \{y^a_{\star}, y^b_{\star}\}$ for convenience in a later case analysis.
\item For every integer $j \in [|X|]$, let $\bit(j)$ denote the binary
representation of $j$ using $q$ bits.
Connect $x_j$ with $y_{i}$
if the $i^{th}$ bit (going from left to right) in $\bit(j)$ is $1$.
\item Add a vertex, denoted by $\bitrepnullifier(X)$,
and make it adjacent to every vertex in $\{y_i\mid i \in [q] \}\cup\{y_{\star}\}$.
One can think of the vertex $\bitrepnullifier(X)$ as the only vertex whose bit-representation contains all $1$'s.
\item For every vertex $u \in V(G)\setminus (X\cup X^+)$ such that
$u$ is adjacent to some vertex in $X$, add an edge between
$u$ and $\bitrepnullifier(X)$.
We add this vertex to ensure that
vertices in $\bitrep(X)$ do not resolve critical pairs
in $V(G)$.
\end{itemize}
This completes the construction of $G$. The properties of $G$ are not proven yet, but just given as an intuition behind its construction. See Figure~\ref{fig:set-identifying-gadget} for an illustration.

\subsubsection{Gadget to Add Critical Pairs}
\label{subsubsec:gadget-critical-pairs}

Any resolving set needs to resolve \emph{all} pairs of vertices in
the input graph.
As we will see, some pairs, which we call critical pairs, are harder
to resolve than others.
In fact, the non-trivial part will be to
resolve all of the critical pairs.

Suppose that we need to have $m \in \mathbb{N}$ critical pairs in a graph $G$, say $\langle c^\circ_i,   c^\star_i \rangle$ for every $i \in [m]$.
Define $C := \{c^{\circ}_i,  c^{\star}_i\mid i \in [m]\}$.
We then add $\bitrep(C)$ and $\bitrepnullifier(C)$ as mentioned above (taking $C$ as the set $X$), but
the connection across $\{c^{\circ}_i, c^\star_i\}$ and $\bitrep(C)$ is defined by $\bit(i)$, i.e.,
connect both~$c^{\circ}_i$ and $c^\star_i$ with the $j$-th vertex of $\bitrep(C)$
if the $j^{th}$ digit (going from left to right) in $\bit(i)$ is $1$.
Hence,  $\bitrep(C)$ can resolve any pair of the form
$\langle c^{\circ}_i,  c^{\star}_{\ell} \rangle$,  $\langle c^{\circ}_i,  c^{\circ}_{\ell} \rangle$,  or
$\langle c^{\star}_i,  c^{\star}_{\ell} \rangle$ as long as $i \neq \ell$.
As before, $\bitrep(C)$ can also resolve all pairs with one vertex in
$C \cup \bitrep(C) \cup \{\bitrepnullifier(C)\}$, but no critical pair of vertices. Again, when these facts
will be used, they will be proven formally.

\subsubsection{Vertex Selector Gadgets}
\label{subsubsec:vertex-selector}

Suppose that we are given a collection of sets $A_1, A_2, \dots, A_q$ of vertices
in a graph $G$, and
we want to ensure that any resolving set of $G$ includes at least one vertex
from $A_i$ for every $i \in [q]$.
In the following,  we construct a gadget that achieves a slightly weaker objective.
\begin{itemize}
\item Let $A = \underset{i \in [q]}{\bigcup} A_i$.
Add a set identifying gadget for $A$ as mentioned in Subsection~\ref{subsubsec:gadget-set-MD}.
\item For every $i \in [q]$,  add two vertices $b^{\circ}_i$ and $b^{\star}_i$.
Use the gadget mentioned in Subsection~\ref{subsubsec:gadget-critical-pairs}
to make all the pairs of the form
$\langle b^{\circ}_i, b^{\star}_i \rangle$ critical pairs.

\item For every $a \in A_i$,  add an edge $(a, b^{\circ}_i)$.
We highlight that we do not make $a$ adjacent to $b^{\star}_i$
by a dotted line in Figure~\ref{fig:vertex-selector-gadget}.
Also, add the edges $(a,  \bitrepnullifier(B))$,  $(b^{\circ}_i,  \bitrepnullifier(A))$, $(b^{\star}_i,  \bitrepnullifier(A))$, and $(\bitrepnullifier(A), \bitrepnullifier(B))$.
\end{itemize}
This completes the construction.

Note that the only vertices that can resolve a critical pair $\langle b^{\circ}_i, b^{\star}_i \rangle$,
apart from $b^{\circ}_i$ and $b^{\star}_i$, are the vertices in $A_i$.
Hence, every resolving set contains at least one vertex in
$\{b^{\circ}_i, b^{\star}_i\} \cup A_i$. Again, when used, these facts will be proven formally.

\subsubsection{Set Representation}
\label{subsubsec:set-representation-MD}

For a positive integer $p$,  define $\calF_p$ as the collection of
subsets of $[2p]$ that contains exactly $p$ integers.
We critically use the fact that no set in $\calF_p$
is contained in any other set in $\calF_p$
(such a collection of sets are called a \emph{Sperner family}).
Let $\ell$ be a positive integer such that $\ell \leq \binom{2p}{p}$.
We define
$\setrep: [\ell] \mapsto \calF_p$ as a one-to-one function
by arbitrarily assigning a set in $\calF_p$ to an integer in $[\ell]$.
By the asymptotic estimation of the central binomial coefficient, $\binom{2p}{p}\sim \frac{4^p}{\sqrt{\pi \cdot p}}$ \cite{Sperner}.
To get the upper bound of $p$, we scale down the asymptotic function and have $\ell \leq \frac{4^p}{2^p}=2^p$.
Thus, $p=\OO(\log \ell)$. 

We mention an application of such a function in the context of
\textsc{Metric Dimension}.
Suppose that we have a ``large'' collection of vertices,  say $A = \{a_1, a_2, \dots, a_{\ell} \}$,
and a ``large'' collection of critical pairs
$C = \{\langle c^{\circ}_1, c^{\star}_1\rangle,  \langle c^{\circ}_2, c^{\star}_2\rangle,  \dots,
\langle c^{\circ}_m,  c^{\star}_m\rangle\}$.
Moreover,  we are given an injective function $\phi:[m] \mapsto [\ell]$.
The objective is to design a gadget such that only $a_{\phi(q)} \in A$
can resolve a critical pair $\langle c^{\circ}_q, c^{\star}_q\rangle \in C$
for any $q \in [m]$, while keeping the treewidth of this part of
the graph of order $\calO(\log(|A|))$.
We add the following vertices and edges in order to achieve this objective.

\begin{itemize}
\item Add vertices and edges as mentioned in Subsection~\ref{subsubsec:gadget-set-MD}
and in Subsection~\ref{subsubsec:gadget-critical-pairs},  respectively,
to identify the set $A$ and to add critical pairs in $C$.
\item Add a \emph{validation portal}, a clique on $2p$ vertices, denoted by $V = \{v_1, v_2, \dots,  v_{2p}\}$,
and vertices and edges to identify it.
\item For every $i \in [\ell]$ and for every $p' \in \setrep(i)$,
add the edge $(a_i,  v_{p'})$.
\item For every critical pair $\langle c^{\circ}_q, c^{\star}_q \rangle$,
make $c^{\circ}_q$ adjacent to every vertex in $V$,  and
add every edge of the form $(c^{\star}_q, v_{p'})$ for $p' \in [2p] \setminus \setrep(\phi(q))$.
Note that the vertices in $V$ that are indexed using integers
in $\setrep(\phi(q))$ are \emph{not} adjacent with $c^{\star}_q$.
\end{itemize}

See Figure~\ref{fig:set-rep-core} for an illustration.
\begin{figure}[t]
    \centering
        \includegraphics[scale=1.25]{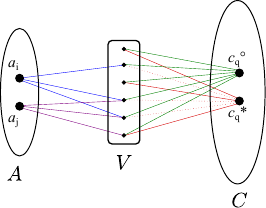}
    \caption{A toy example to illustrate the application of $\setrep$.
    See Subsection~\ref{subsec:prelim-3-Par-3-SAT-Met-Dim-diam-tw}.
    Suppose that $\phi(q) = i$, i.e., we want to add gadgets such that only $a_i$ in $A$ can resolve
    the critical pair $\langle c^{\circ}_q, c^{\star}_q \rangle$.
    Suppose that the vertices in $V$ are indexed from top to bottom and $\setrep(i) = \{2, 4, 5\}$.
    By the construction, the only vertices in $V$ that $c^{\star}_q$ is \emph{not} adjacent to are $v_2$, $v_4$, and $v_5$ (this fact is highlighted with red-dotted edges).
    Thus,  $\dist(a_i,  c^{\circ}_q) = 2$ and $\dist(a_i, c^{\star}_q) > 2$,  and hence,
    $a_i$ resolves the critical pair $\langle c^{\circ}_q, c^{\star}_q \rangle$.
    For any other vertex in $A$, say $a_j$, $\setrep(j) \setminus \setrep(i)$ is a non-empty set.
    Hence,  there are shortest paths from $a_j$ to  $c^{\circ}_q$, and $a_j$ to $c^{\star}_q$
    through the vertices in $V$ with indices in $\setrep(j) \setminus \setrep(i)$.
    This implies that $\dist(a_j,  c^{\circ}_q) = \dist(a_j, c^{\star}_q) = 2$ and $a_j$ cannot
    resolve the pair $\langle c^{\circ}_q, c^{\star}_q \rangle$.
    The sets $\bitrep(X)$ and $\bitrepnullifier(X)$ are omitted for $X\in \{A, V, C\}$.
    \label{fig:set-rep-core}}
\end{figure}

\subsection{Reduction}\label{sec:reduction-tw}

Consider an instance $\psi$ of \textsc{3-Partitioned-3-SAT},
with $X^{\alpha}, X^{\beta}, X^{\gamma}$ the partition of the variable set.
From $\psi$, we construct the graph $G$ as follows.
We describe the construction of the part of the graph $G$ that corresponds to $X^{\alpha}$, with the parts corresponding to $X^{\beta}$ and $X^{\gamma}$ being analogous.
We rename the variables in $X^{\alpha}$
to $x^{\alpha}_{i}$ for $i \in [n]$.

\begin{itemize}
\item
For every variable $x^{\alpha}_i$, we add a pair
$\langle x^{\alpha, \circ}_i, x^{\alpha, \star}_i \rangle$ of vertices.
We add vertices and edges as mentioned in Subsection~\ref{subsubsec:gadget-critical-pairs}
to make all pairs of the form $\langle x^{\alpha, \circ}_i, x^{\alpha, \star}_i \rangle$
critical in the graph $G$.
We denote
$X^{\alpha} = \{x^{\alpha, \circ}_i, x^{\alpha, \star}_i\  |\ i \in [n] \}$
as the collection of vertices in the critical pairs.
We remark that we do \emph{not} convert $X^{\alpha}$ into a clique.

\item For every variable $x^{\alpha}_i$,
we add the vertices $t^{\alpha}_{2i}$ and $f^{\alpha}_{2i - 1}$.
Formally,
$A^{\alpha} = \{t^{\alpha}_{2i},  f^{\alpha}_{2i - 1} |\ i \in [n] \}$,
and hence, $|A^{\alpha}| = 2n$.
We add vertices and edges as mentioned in
Subsection~\ref{subsubsec:gadget-set-MD} in order to identify the
set $A^{\alpha}$ in $G$.

\item We would like that
any resolving set contains at least one vertex in
$\{t^{\alpha}_{2i},  f^{\alpha}_{2i - 1}\}$ for every $i \in [n]$, but instead we add the construction mentioned in 
Subsection~\ref{subsubsec:vertex-selector} that achieves the slightly weaker objective as mentioned there.
As before,  instead of adding two new vertices, we use
$\langle x^{\alpha, \circ}_i,  x^{\alpha,  \star}_i \rangle$ as the necessary critical pair.
Formally,  for every $i \in [n]$,
we add the edges $(x^{\alpha, \circ}_i, t^{\alpha}_{2i})$ and $(x^{\alpha, \circ}_i, f^{\alpha}_{2i-1})$.
We add edges to make $\bitrepnullifier(X^{\alpha})$ adjacent to every vertex
in $A^{\alpha}$, and
$\bitrepnullifier(A^{\alpha})$ adjacent to every vertex in $X^{\alpha}$.
Also, we add the edge $(\bitrepnullifier(X^{\alpha}), \bitrepnullifier(A^{\alpha}))$.

\item
Let $p$ be the smallest positive integer such that $2n \leq \binom{2p}{p}$.
In particular, $p=\OO(\log n)$.
Moreover,  define $\setrep: [2n] \mapsto \calF_{p}$
as mentioned in Subsection~\ref{subsec:prelim-3-Par-3-SAT-Met-Dim-diam-tw}.

\item We add a \emph{validation portal},
a clique on $2p$ vertices,  denoted by
$V^{\alpha} = \{v^{\alpha}_1, v^{\alpha}_2,  \dots, v^{\alpha}_{2p} \}$.
We add vertices and edges to identify $V^{\alpha}$ as mentioned
in Subsection~\ref{subsubsec:gadget-set-MD}.
We add the edge $(\bitrepnullifier(V^{\alpha}),  \bitrepnullifier(A^{\alpha}))$
and make $\bitrepnullifier(A^{\alpha})$ adjacent to every vertex in $V^{\alpha}$.
We note that we \emph{do not} add edges across
$\bitrepnullifier(V^{\alpha})$
 and $A^{\alpha}$.

\item We add edges across $A^{\alpha}$ and the validation portal as follows:
for each $i \in [n]$,
we add the edge $(t^{\alpha}_{2i},  v^{\alpha}_{p'})$
for every $p' \in \setrep(2i)$. Similarly, for each $i\in [n]$, we add the
edge $(f^{\alpha}_{2i - 1},  v^{\alpha}_{p'})$
for every $p' \in \setrep(2i - 1)$.

We repeat the above steps to construct
$X^{\beta},  A^{\beta},  V^{\beta}$,
$X^{\gamma},  A^{\gamma},  V^{\gamma}$,
and their related vertices and edges.

\item For every clause $C_q$ in $\psi$, we introduce a pair
$\langle c^{\circ}_q,  c^{\star}_q \rangle$ of vertices.
We add vertices and edges to make each pair of the form
$\langle c^{\circ}_q,  c^{\star}_q \rangle$ a critical pair
as mentioned in Subsection~\ref{subsubsec:gadget-critical-pairs}.
Let $C$ be the collection of the vertices in such pairs.

\item We add edges across $C$ and the portals as follows.
Consider a clause $C_q$ in $\psi$ and the corresponding
critical pair $\langle c^{\circ}_q,  c^{\star}_q \rangle$ in $C$.
Suppose $\delta \in \{\alpha, \beta, \gamma\}$.
As $\psi$ is an instance of \textsc{$3$-Partitioned-$3$-SAT},
there is at most one variable in $X^{\delta}$ that appears in $C_q$.
Suppose that variable is $x^{\delta}_{i}$ for some $i \in [n]$.

We add all edges of the form
$(v^{\delta}_{p'},  c^{\circ}_q)$ for every $p' \in [2p]$.
If $x^{\delta}_{i}$ appears as a positive literal in $C_q$, then
we add the edge $(v^{\delta}_{p'},  c^{\star}_q)$
for every $p' \in [2p] \setminus \setrep(2i)$
(which corresponds to $t^{\delta}_{2i}$).
If $x^{\delta}_{i}$ appears as a negative literal in $C_q$, then
we add the edge $(v^{\delta}_{p'},  c^{\star}_q)$
for every $p' \in [2p] \setminus \setrep(2i  - 1)$
(which corresponds to $f^{\delta}_{2i - 1}$).
We remark that if $x^{\delta}_{i}$ appears as a positive (negative, respectively) literal in $C_q$, then the vertices in $V^{\delta}$
whose indices are in $\setrep(2i)$ ($\setrep(2i - 1)$, respectively) are \emph{not adjacent}
to $c^{\star}_q$ .
If there is no variable in $X^{\delta}$ that appears in $C_q$, 
then we make every vertex in $V^{\delta}$ adjacent to both $c^{\circ}_q$ and $c^{\star}_q$.
Finally, we add the edge $(\bitrepnullifier(V^{\delta}),  \bitrepnullifier(C))$.
See Figure~\ref{fig:reduction-overview-diam-tw}.
\end{itemize}

\begin{figure}[t]
    \centering
        \includegraphics[scale=1.25]{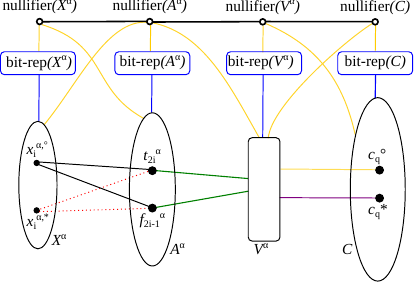}
    \caption{Overview of the reduction.
    For any set $X \in \{X^{\alpha}, A^{\alpha},  V^{\alpha}, C\}$,  the blue rectangle attached to it
    via the blue edge represents $\bitrep(X)$, and
    the yellow line between a vertex and $\bitrep(X)$ indicates that vertex is connected to every vertex in $\bitrep(X)\setminus \bits(X)$. The remainder of the yellow lines represent that vertex is connected to every vertex in the set the edge goes to.
    Note that $\bitrepnullifier(V^{\alpha})$ is \emph{not} adjacent to any
    vertex in $A^{\alpha}$.
    Green edges denote adjacencies with respect to $\setrep$,
    i.e.,  $t^{\alpha}_{2i}$ is adjacent to $v_j\in V^{\alpha}$ if $j\in \setrep(2i)$.
    The same holds for $f_{2i-1}$ for all $i\in[n]$.
    Purple lines also indicate adjacencies with respect to $\setrep$, but in a complementary
    way, i.e., if $x_i\in c_q$, then, for all $p'\in [2p]\setminus \setrep(2i)$, we have that $(v_{p'}^{\alpha}, c_q^{\star})\in E(G)$, and if $\overline{x}_i\in c_q$, then, for all $p'\in [2p]\setminus \setrep(2i-1)$, we have that $(v_{p'}^{\alpha}, c_q^{\star})\in E(G)$.}
    \label{fig:reduction-overview-diam-tw}
\end{figure}
This concludes the construction of $G$.
The reduction returns $(G,  k)$
as an instance of \textsc{Metric Dimension} where
\begin{equation*}
k= 3 \cdot (n + (\lceil \log(|X^{\alpha}|/2+2)\rceil + 1)
+ (\lceil\log(|A^{\alpha}|+2)\rceil + 1) + (\lceil \log(|V^{\alpha}|+2)\rceil + 1)) + \lceil\log(|C|/2+2)\rceil + 1.
\end{equation*} 

\subsection{Correctness of the Reduction}
\label{subsec:correctness-3-Par-3-SAT-Met-Dim-diam-tw}

Suppose,  given an instance $\psi$ of \textsc{$3$-Partitioned-$3$-SAT}, that
the reduction returns $(G, k)$ as an instance of \textsc{Metric Dimension}.
We first prove the following lemma, which will be helpful in proving the correctness of the reduction.

\begin{lemma} For any resolving set $S$ of $G$ and for all $X\in \{C\} \cup \{X^{\delta},A^{\delta},V^{\delta} \mid \delta\in \{\alpha,\beta,\gamma\}\}$,
\begin{enumerate}
\item $S$ contains at least one vertex from each pair of false twins in $\bits(X)$.
\item Vertices in $\bits(X) \cap S$ resolve any non-critical pair of vertices
$\langle u, v\rangle$ when $u \in X\cup X^+$ and $v\in V(G)$.  
\item Vertices in $X^+ \cap S$ \emph{cannot} resolve any critical
pair of vertices $\langle x_i^{\delta',\circ},  x_i^{\delta',\star}\rangle$ nor $\langle c_q^{\circ},  c_q^{\star}\rangle$ for all $i\in [n]$, $\delta'\in \{\alpha,\beta,\gamma\}$, and $q\in [m]$.     
\end{enumerate}\label{lemma:set-id-tw}
\end{lemma}
\begin{proof}
\begin{enumerate}
\item By~\cref{obs:twins}, the statement follows for all $X\in \{C\} \cup \{X^{\delta},A^{\delta},V^{\delta} \mid \delta\in \{\alpha,\beta,\gamma\}\}$.

\item For all $X\in \{C\} \cup \{X^{\delta},A^{\delta},V^{\delta} \mid \delta\in \{\alpha,\beta,\gamma\}\}$, note that $\bitrepnullifier(X)$ is distinguished by 
$S \cap \bits(X)$ since it is the only vertex in $G$ that is at distance~$2$ from every vertex in $\bits(X)$. We now do a case analysis for the remaining non-critical pairs of vertices $\langle u, v \rangle$ assuming that $\bitrepnullifier(X)\notin \{u,v\}$ (also, suppose that both $u$ and $v$ are not in $S$, as otherwise, they are obviously distinguished):
\begin{description}
\item[Case i: $u, v \in X\cup X^+$.]
\hfill
\begin{description}
\item[Case i(a): $u, v\in X$ or $u, v\in \bitrep(X)\setminus \bits(X)$.] In the first case, let $j$ be the digit in the binary representation of the subscript of $u$ that is not equal to the $j^{\text{th}}$ digit in the binary representation of the subscript of $v$ (such a $j$ exists since $\langle u, v \rangle$ is not a critical pair). In the second case, without loss of generality, let $u=y_i$ and $v=y_j$. 
By the first item of the statement of the lemma (1.), without loss of generality, $y_j^{a}\in S\cap \bits(X)$.
Then, in both cases, $d(y_j^a, u)\neq d(y_j^a, v)$.
\item[Case i(b): $u \in X$ and $v\in \bitrep(X)$.]
Without loss of generality, $y_{\star}^{a}\in S\cap \bits(X)$ (by 1.). 
Then, $d(y_{\star}^{a}, u)=2$ and, for all $v\in \bits(X)\setminus\{y_{\star}^b\}$, $d(y_{\star}^{a}, v)=3$. 
Without loss of generality, let $y_i$ be adjacent to $u$ and let $y_i^a\in S\cap \bits(X)$ (by 1.).
Then, for $v=y_{\star}^b$, $3=d(y_i^a, v)\neq d(y_i^a, u)=2$.
If $v\in \bitrep(X)\setminus \bits(X)$, then, without loss of generality, $v=y_j$ and $y_j^a\in S\cap \bits(X)$ (by 1.), and $1=d(y_j^a, v)< d(y_j^a, u)$.
\item[Case i(c): $u, v\in \bits(X)$.]
Without loss of generality, $u=y_i^b$ and $y_i^a\in S$ (by 1.). 
Then, $2=d(y_i^a, u)\neq d(y_i^a, v)=3$.
\item[Case i(d): $u \in \bits(X)$ and $v\in \bitrep(X)\setminus \bits(X)$.]
Without loss of generality, $v=y_i$ and $y_i^a\in S$ (by 1.). 
Then, $1=d(y_i^a, v) < d(y_i^a, u)$.
\end{description}
\item[Case ii: $u \in X\cup X^+$ and $v \in V(G)\setminus (X\cup X^+)$.]
For each $u \in X\cup X^+$, there exists $w\in \bits(X)\cap S$ such that $d(u, w)\leq 2$, while, for each $v \in V(G)\setminus (X\cup X^+)$ and $w\in \bits(X)\cap S$, we have $d(v, w)\geq 3$.
\end{description}
\item For all $X\in \{X^{\delta},A^{\delta},V^{\delta} \mid \delta\in \{\alpha,\beta,\gamma\}\}$, $u\in X^+$, $v \in \{c_q^{\circ}, c_q^{\star}\}$, and $q\in [m]$,
we have that $d(u, v)=d(u, \bitrepnullifier(V^{\delta}))+1$.
Further, for $X=C$ and all $u\in X^+$ and $q\in [m]$, either $d(u,c_q^{\circ})=d(u,c_q^{\star})=1$, $d(u,c_q^{\circ})=d(u,c_q^{\star})=2$, or $d(u,c_q^{\circ})=d(u,c_q^{\star})=3$ by the construction in Subsection~\ref{subsubsec:gadget-critical-pairs} and since $\bitrep(X)\setminus \bits(X)$ is a clique. 
Hence, for all $X\in \{C\} \cup \{X^{\delta},A^{\delta},V^{\delta} \mid \delta\in \{\alpha,\beta,\gamma\}\}$, vertices in $X^+ \cap S$ cannot resolve a pair of vertices $\langle c_q^{\circ}, c_q^{\star}\rangle$ for any $q\in [m]$.

\smallskip

For all $\delta\in \{\alpha,\beta,\gamma\}$, if $v\in X^{\delta}$, then, for all $X\in \{C\} \cup \{X^{\delta'},A^{\delta'},V^{\delta'} \mid \delta'\in \{\alpha,\beta,\gamma\}\}$ such that $\delta\neq \delta'$, and $u\in X^+$, we have that $d(u,v)=d(u, \bitrepnullifier(A^{\delta}))+1$. Similarly, for all $\delta\in \{\alpha,\beta,\gamma\}$, if $v\in X^{\delta}$, then, for all $X\in \{A^{\delta}, V^{\delta}\}$ and $u\in X^+$, we have that $d(u,v)=d(u, \bitrepnullifier(A^{\delta}))+1$. Lastly, for each $\langle x_i^{\delta,\circ},  x_i^{\delta,\star}\rangle$, $\delta\in \{\alpha,\beta,\gamma\}$, and $i\in [n]$, if $X=X^{\delta}$, then, for all $u\in X^+$, either $d(u,x_i^{\delta,\circ})=d(u,x_i^{\delta,\star})=1$, $d(u,x_i^{\delta,\circ})=d(u,x_i^{\delta,\star})=2$, or $d(u,x_i^{\delta,\circ})=d(u,x_i^{\delta,\star})=3$ by the construction in Subsection~\ref{subsubsec:gadget-critical-pairs} and since $\bitrep(X)\setminus \bits(X)$ is a clique. \qedhere
\end{enumerate}
\end{proof}

\begin{lemma}
\label{lemma:3-Part-3-SAT-Met-Dim-diam-tw-forward}
If $\psi$ is a satisfiable \textsc{$3$-Partitioned-$3$-SAT} formula, then $G$ admits a resolving set of size $k$.
\end{lemma}

\begin{proof}
Suppose $\pi: X^{\alpha} \cup X^{\beta} \cup X^{\gamma} \mapsto \{\true, \false\}$
is a satisfying assignment for $\psi$.
We construct a resolving set $S$ of size $k$ for $G$ using this assignment.

For every $\delta \in \{\alpha, \beta, \gamma\}$ and $i \in [n]$,
if $\pi(x^{\delta}_i)=\true$, then let $t^{\delta}_{2i}\in S$, and otherwise, let $f^{\delta}_{2i-1}\in S$.
For every $X \in \{B^{\delta},  A^{\delta},  V^{\delta}, C\}$ and $\delta \in \{\alpha, \beta, \gamma\}$,
add one vertex from each pair of false twins in $\bits(X)$ to $S$.
Note that the size of $S$ is $k$.

In the remaining part of the proof, we show that $S$ is a resolving set of $G$.
First, we prove that all critical pairs are resolved by $S$ in the following claim.

\begin{claim}\label{clm:critpairs-tw}
All critical pairs are resolved by $S$.
\end{claim}
\begin{claimproof}
For each $i\in [n]$ and $\delta\in \{\alpha, \beta, \gamma\}$, the critical pair $\langle x^{\delta, \circ}_i, x^{\delta, \star}_i \rangle$
is resolved by the vertex $S\cap A^{\delta}$ by the construction.
For each $q\in [m]$, the clause $C_q$ is satisfied by the assignment $\pi$.
Thus, there is a variable, say $x_i$ in $C_q$, that satisfies $C_q$ according to $\pi$.
If $x_i$ appears positively in $C_q$, then $t_{2i}^{\delta}\in S$ resolves the critical pair $\langle c_q^{\circ}, c_q^{\star} \rangle$ since
$d(t_{2i}^{\alpha},c_q^{\circ})=2<d(t_{2i}^{\alpha},c_q^{\star})$ by the construction.
Similarly, if $x_i$ appears negatively in $C_q$, then $f_{2i-1}^{\delta}\in S$ resolves the critical pair $\langle c_q^{\circ}, c_q^{\star} \rangle$ since $d(f_{2i-1}^{\alpha},c_q^{\circ})=2<d(f_{2i-1}^{\alpha},c_q^{\star})$ by the construction.
Thus, every critical pair $\langle c_q^{\circ}, c_q^{\star} \rangle$ is resolved by $S$.
\end{claimproof}

Then, every vertex pair in $V(G)$ is resolved by $S$ by \cref{clm:critpairs-tw} in conjunction with the second item of the statement of \cref{lemma:set-id-tw}.
\end{proof}

\begin{lemma}
\label{lemma:3-Part-3-SAT-Met-Dim-diam-tw-backword}
If $G$ admits a resolving set of size $k$, then $\psi$ is a satisfiable \textsc{$3$-Partitioned-$3$-SAT} formula.
\end{lemma}
\begin{proof}
Assume that $G$ admits a resolving set $S$ of size $k$.
First, we prove some properties regarding $S$.
By the first item of the statement of \cref{lemma:set-id-tw}, for each $\delta \in \{\alpha, \beta, \gamma\}$, we have
\begin{equation*}
\begin{split}
&|S\cap \bits(X^{\delta})| \geq \lceil\log(|X^{\delta}|/2+2)\rceil+1, \quad
|S\cap \bits(V^{\delta})| \geq \lceil\log(|V^{\delta}|+2)\rceil+1, \\
&|S\cap \bits(A^{\delta})| \geq \lceil\log(|A^{\delta}|+2)\rceil+1, \quad
|S\cap \bits(C)| \geq \lceil\log(|C|/2+2)\rceil + 1.
\end{split}
\end{equation*}
Hence, any resolving set $S$ of $G$ already has size at least
$$3 \cdot ((\lceil \log(|X^{\alpha}|/2+2)\rceil + 1)
+ (\lceil\log(|A^{\alpha}|+2)\rceil + 1) + (\lceil \log(|V^{\alpha}|+2)\rceil + 1)) + \lceil\log(|C|/2+2)\rceil + 1.$$

Now, for each $\delta\in\{\alpha, \beta, \gamma\}$ and $i\in[n]$, consider the critical pair $\langle x_i^{\delta, \circ}, x_i^{\delta, \star}\rangle$.
By the construction mentioned in Subsection~\ref{subsubsec:gadget-critical-pairs}, only $v\in \{t^{\delta}_{2i},  f^{\delta}_{2i - 1}, x_i^{\delta, \circ}, x_i^{\delta, \star}\}$ resolves a pair $\langle x_i^{\delta, \circ}, x_i^{\delta, \star} \rangle$.
Indeed, for all $X\in \{C\} \cup \{X^{\delta'},A^{\delta'},V^{\delta'} \mid \delta'\in \{\alpha,\beta,\gamma\}\}$, no vertex in $X^+$ can resolve such a pair by the third item of the statement of \cref{lemma:set-id-tw}. Also, for all $X\in \{A^{\delta''},A^{\delta}\setminus \{t^{\delta}_{2i},  f^{\delta}_{2i - 1}\},V^{\delta'},C\}$, $\delta' \in \{\alpha, \beta, \gamma\}$, $\delta'' \in \{\alpha, \beta, \gamma\}$ such that $\delta\neq \delta''$, and $u\in X$, we have that $d(u,x_i^{\delta, \circ})=d(u,x_i^{\delta, \star})=d(u,\bitrepnullifier(A^{\delta}))+1$. Hence, since any resolving set $S$ of $G$ of size at most $k$ can only admit at most another $3n$ vertices, we get that equality must in fact hold in every one of the aforementioned inequalities, and any resolving set $S$ of $G$ of size at most $k$ contains one vertex from $\{t^{\delta}_{2i},  f^{\delta}_{2i - 1}, x_i^{\delta, \circ}, x_i^{\delta, \star}\}$ for all $i\in [n]$ and $\delta\in \{\alpha, \beta, \gamma\}$. Hence, any resolving set $S$ of $G$ of size at most $k$ is actually of size exactly $k$.

Next, we construct an assignment $\pi: X^{\alpha}\cup X^{\beta}\cup X^{\gamma} \rightarrow \{\true, \false\}$ in the following way.
For each $\delta\in \{\alpha, \beta, \gamma\}$ and $i\in [n]$, if $t_{2i}^{\delta}\in S$, then set $\pi(x^{\delta}_i):=\true$, and if $f_{2i-1}^{\delta}\in S$, then set $\pi(x^{\delta}_i):=\false$.
For any $i\in [n]$ and $\delta\in\{\alpha, \beta, \gamma\}$, if $S\cap \{t^{\delta}_{2i},  f^{\delta}_{2i - 1}\}=\emptyset$, then one of $ x_i^{\delta, \circ}, x_i^{\delta, \star}$ is in $S$, and we can use an arbitrary assignment of the variable $x_i^{\delta}$.

We prove that the constructed assignment $\pi$ satisfies every clause in $C$.
Since $S$ is a resolving set, it follows that, for every clause $c_q\in C$,
there exists $v\in S$ such that $d(v, c_q^{\circ})\neq d(v, c_q^{\star})$.
Note that, for any $v$ in $\bits(A^{\delta}), \bits(X^{\delta}), \bits(V^{\delta})$ for any $\delta\in \{\alpha, \beta, \gamma\}$ or in $\bits(C)$, we have $d(v, c_i^{\circ})= d(v, c_i^{\star})$ by the third item of the statement of \cref{lemma:set-id-tw}.
Further, for any $v\in X^{\delta}$ and any $\delta\in \{\alpha, \beta, \gamma\}$, we have that $d(v, c_q^{\circ})=d(v, c_q^{\star})=d(v,\bitrepnullifier(V^{\delta}))+1$.
Thus, $v\in S\cap \bigcup\limits_{\delta\in \{\alpha, \beta, \gamma\}} A^{\delta}$.
Without loss of generality, suppose that $c_q^{\circ}$ and $c_q^{\star}$ are resolved by
$t_{2i}^{\alpha}$.
So, $d(t_{2i}^{\alpha}, c_i^{\circ}) \neq d(t_{2i}^{\alpha}, c_i^{\star})$.
By the construction, the only case where $d(t_{2i}^{\alpha}, c_i^{\circ}) \neq d(t_{2i}^{\alpha}, c_i^{\star})$ is when
$C_q$ contains a variable $x_i\in X^{\alpha}$ and $\pi(x_i)$ satisfies $C_q$.
Thus, we get that the clause $C_q$ is satisfied by the assignment $\pi$.

Since $S$ resolves all pairs $\langle c_q^{\circ}, c_q^{\star}\rangle$ in $V(G)$, then the assignment $\pi$ constructed above indeed satisfies every clause $c_q$, completing the proof.
\end{proof}

\begin{proof}[Proof of \Cref{thm:lower-bound-diam-tw}.]
In Subsection~\ref{sec:reduction-tw}, we presented a reduction that takes an instance $\psi$ of \textsc{$3$-Partitioned-$3$-SAT} and returns an equivalent instance $(G,k)$ of \textsc{Metric Dimension} (by Lemmas~\ref{lemma:3-Part-3-SAT-Met-Dim-diam-tw-forward} and \ref{lemma:3-Part-3-SAT-Met-Dim-diam-tw-backword}) in polynomial time.
Now, consider the set $$Z=\{V^{\delta}\cup X^+~|~X \in \{X^{\delta}, A^{\delta},  V^{\delta}, C\}, \delta \in \{\alpha, \beta, \gamma\}\}.$$
It is easy to verify that $|Z| = \calO(\log (n))$
and $G - Z$ is a collection of $P_3$'s and isolated vertices.
Hence,  $\tw(G)$, $\fvs(G)$, and $\td(G)$ are upper bounded by $\calO(\log (n))$.
It is also easy to see that the diameter of the graph is bounded by a constant.
Hence, if there is an algorithm for \textsc{Metric Dimension} 
that runs in time $2^{f(\diam)^{o(\tw)}}$ (or $2^{f(\diam)^{o(\fvs)}}$ or $2^{f(\diam)^{o(td)}}$),
then there is an algorithm solving \textsc{$3$-Partitioned-$3$-SAT} running in time $2^{o(n)}$, which by \cref{prop:3-SAT-to-3-Partition-3-SAT} contradicts the \ETH.
\end{proof}

\section{\gsfull: Lower Bound Regarding Diameter plus Treewidth}
\label{sec:lower-bound-diam-treewidth}


The aim of this section is to prove
the following theorem.
\begin{theorem}
\label{thm:lower-bound-diam-tw-GS}
Unless the \ETH\ fails, \gsfull does not admit an algorithm 
running in time $2^{f(\diam)^{o(\tw)}} \cdot n^{\OO(1)}$ for any computable function 
$f:\mathbb{N} \mapsto \mathbb{N}$.
\end{theorem}

As in the previous section,
we present a different reduction from \textsc{$3$-Partitioned-$3$-SAT} (see \Cref{sec:preliminaries})
to \gsfull. 
The reduction takes as input an instance $\psi$ of \textsc{$3$-Partitioned-$3$-SAT}
on $3n$ variables and returns $(G, k)$ as an instance of \gsfull
such that $\tw(G) = \calO(\log(n))$ and $\diam(G) = \calO(1)$.
We rely on the tool of set representation introduced in Section~\ref{subsubsec:set-representation-MD}. For convenience, we recall it in the next subsection and describe how we apply it in the reduction to prove Theorem~\ref{thm:lower-bound-diam-tw-GS}.

\subsection{Preliminary Tool: Set Representation}
For a positive integer $p$, define $\calF_p$ as the collection of
subsets of $[2p]$ that contains exactly $p$ integers.
We critically use the fact that no set in $\calF_p$
is contained in any other set in $\calF_p$
(such a collection of sets is called a \emph{Sperner family}).
Let $\ell$ be a positive integer such that $\ell \leq \binom{2p}{p}$.
We define
$\setrep: [\ell] \mapsto \calF_p$ as a one-to-one function
by arbitrarily assigning a set in $\calF_p$ to an integer in $[\ell]$.
By the asymptotic estimation of the central binomial coefficient, $\binom{2p}{p}\sim \frac{4^p}{\sqrt{\pi \cdot p}}$ \cite{Sperner}.
To get the upper bound of $p$, we scale down the asymptotic function and have $\ell \leq \frac{4^p}{2^p}=2^p$.
Thus, $p=\OO(\log \ell)$. 

We will apply the existence of such a function in the context of
\gsfull.
Suppose we have a ``large'' collection of vertices,  say $A = \{a_1, a_2, \dots, a_{\ell} \}$,
and a ``large'' collection of vertices
$C = \{c_1, c_2, \dots, c_m\}$.
Moreover, we are given a function $\phi:[m] \mapsto [\ell]$.
The basic idea is to design gadgets such that $c_q$ is only covered by the shortest path from $a_{\phi(q)} \in A$ to $c^b_q$ ($c^b_q$ is forced to be chosen in the geodetic set)
for any $q \in [m]$, while keeping the treewidth of this part of
the graph of order $\calO(\log(|A|))$. To do so, we create a ``small'' intermediate set $V$ (of size $\calO(\log(|A|))$) through which will go the shortest paths between vertices in $A$ and $C$, and we connect $a_i$ to the vertices of $V$ corresponding to the bit-representation of $\setrep(i)$, and $c_q$ (with $i=\phi(q)$) to all the other vertices of $V$. In this way, the construction will ensure that $c_q$ is covered by a shortest path between $a_{\phi(q)}$ and $c^b_q$, but is not covered by any other shortest path between a vertex of $A$ and a vertex of $C$.
We give the details in the following subsection.

\subsection{Reduction}
Consider an instance $\psi$ of \textsc{3-Partitioned-3-SAT},
with $X^{\alpha}, X^{\beta}, X^{\gamma}$ the partition of the variable set.
From $\psi$, we construct the graph $G$ as follows.
We describe the construction of $X^{\alpha}$,
with the constructions for $X^{\beta}$ and $X^{\gamma}$ being analogous. See Figure~\ref{fig:reduction-diameter-tw} for an illustration. 
We rename the variables in $X^{\alpha}$
to $x^{\alpha}_{i}$ for $i \in [n]$.

\begin{figure}[t]
	\centering
	\includegraphics[scale=0.9]{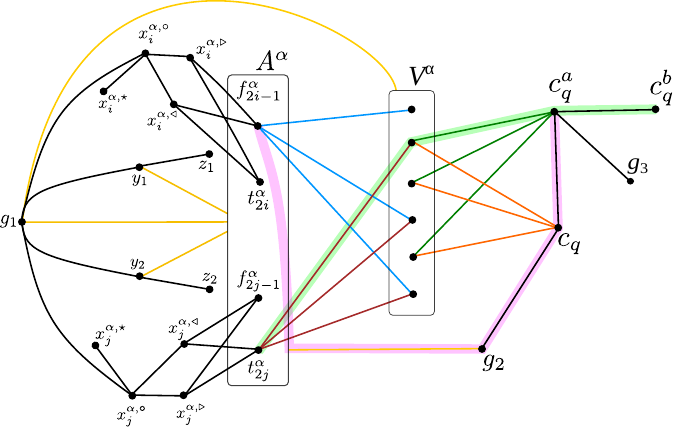}
	\caption{Overview of the reduction. We only draw $A^{\alpha}$ and $V^{\alpha}$ here, as
		$A^{\beta}$, $A^{\gamma}$, $V^{\beta}$, and $V^{\gamma}$ are similar. The yellow lines joining $g_1$, $g_2$, $y_1$, and $y_2$ to sets indicate that the corresponding vertex is adjacent to all the vertices of the corresponding set. Suppose that $f^{\alpha}_{2i-1}$ and $t^{\alpha}_{2j}$ are in the geodetic set and $\overline{x}_i$ appears in the clause $c_q$. The thick green path is a shortest path between $t^{\alpha}_{2j}$ and $c^b_q$ which does not cover $c_q$. The thick violet path plus the edge $(c^a_q, c^b_q)$ is a shortest path between $f^{\alpha}_{2i-1}$ and $c^b_q$ covering $c_q$.}
	\label{fig:reduction-diameter-tw}
\end{figure}

\begin{itemize}
\item
For every variable $x^{\alpha}_i$, we add the vertices $t^{\alpha}_{2i}$ and $f^{\alpha}_{2i - 1}$.
Formally,
$A^{\alpha} = \{t^{\alpha}_{2i},  f^{\alpha}_{2i - 1} ~|\ i \in [n] \}$,
and hence, $|A^{\alpha}| = 2n$.

\item
For every variable $x^{\alpha}_i$,  we add four vertices:
$x^{\alpha, \triangleleft}_i, x^{\alpha, \triangleright}_i, x^{\alpha, \circ}_i, x^{\alpha, \star}_i$.
We make $x^{\alpha, \triangleleft}_i$ and $x^{\alpha, \triangleright}_i$ adjacent to both $t^{\alpha}_{2i}$ and $f^{\alpha}_{2i - 1}$.
We make $x^{\alpha, \circ}_i$ adjacent to both $x^{\alpha, \triangleleft}_i$ and $x^{\alpha, \triangleright}_i$.
We make $x^{\alpha, \star}_i$ adjacent to $x^{\alpha, \circ}_i$.


\item
We add the vertices $y_1,y_2,z_1,z_2$.
We make $y_1$ and $y_2$ adjacent to every vertex of $A^{\alpha}$.
We make $y_i$ adjacent to $z_i$ for $i\in \{1,2\}$.
Note that $y_1,y_2,z_1,z_2$ are common to $X^{\beta}$ and $X^{\gamma}$.

\item
We add the vertex $g_1$ and make it adjacent to $y_1$, $y_2$, and $x^{\alpha, \circ}_i$ for each $i\in [n]$.
Note that $g_1$ is common to $X^{\beta}$ and $X^{\gamma}$.
We add edges between $g_1$ and every vertex of $A^{\alpha}$.

\item
Let $p$ be the smallest positive integer such that $2n \leq \binom{2p}{p}$.
In particular, $p=\OO(\log n)$.
We add a \emph{validation portal},
a clique on $2p$ vertices,  denoted by
$V^{\alpha} = \{v^{\alpha}_1, v^{\alpha}_2,  \dots, v^{\alpha}_{2p} \}$.
For each $\delta \in \{\alpha,\beta,\gamma\}$, we 
add edges between $g_1$ and every vertex of $V^{\delta}$.

\item
For every clause $C_q$ in $\psi$, we introduce three vertices: $c_q, c^a_q, c^b_q$.
We add the edges $(c_q,c^a_q)$ and $(c^a_q,c^b_q)$.

\item
Define $\setrep: [2n] \mapsto \calF_p$ as an arbitrary injective function,
where $\calF_p$ is the Sperner family (and $p$ is as defined two items above).
Add the edge $(t^{\alpha}_{2i},v^{\alpha}_{p'})$ for every $p'\in \setrep(2i)$ and the edge $(f^{\alpha}_{2i-1},v^{\alpha}_{p'})$ for every $p'\in \setrep(2i-1)$.
If the variable $x_i^{\alpha}$ appears positively in the clause $C_q$, then we add the edges $(c_q, v^{\alpha}_{p'})$ and $(c^a_q, v^{\alpha}_{p'})$
for every $p' \in [2p] \setminus \setrep(2i)$.
If the variable $x_i^{\alpha}$ appears negatively in the clause $C_q$, then we add the edges $(c_q, v^{\alpha}_{p'})$ and $(c^a_q, v^{\alpha}_{p'})$
for every $p' \in [2p] \setminus \setrep(2i-1)$.

\item
Add a vertex $g_2$ and make $g_2$ adjacent to every vertex of $A^{\alpha}$ and every vertex of $\{c_q: q\in [m]\}$.
Note that $g_2$ is common to $X^{\beta}$ and $X^{\gamma}$.

\item Add a vertex $g_3$ and make it adjacent to every vertex of $\{c^a_q: q\in [m]\}$. Note that $g_3$ and the vertices of $\{c_q, c^a_q, c^b_q: q\in [m]\}$ are common to $X^{\beta}$ and $X^{\gamma}$.
\end{itemize}

This concludes the construction of $G$.
The reduction returns $(G,  k)$ as an instance of \gsfull where $k = 6n+m+2$.

\subsection{Correctness of the Reduction}

Suppose, given an instance $\psi$ of \textsc{$3$-Partitioned-$3$-SAT}, that the reduction above returns $(G, k)$ as an instance of \textsc{Geodetic Set}.

\begin{lemma}
\label{lemma:3-Part-3-SAT-Geo-Set-diam-tw-forward}
If $\psi$ is a satisfiable \textsc{$3$-Partitioned-$3$-SAT} formula, then $G$ admits a geodetic set of size $k$.
\end{lemma}

\begin{proof}
Suppose that $\pi: X^{\alpha} \cup X^{\beta} \cup X^{\gamma} \mapsto \{\true, \false\}$
is a satisfying assignment for $\psi$.
We construct a geodetic set $S$ of size $k$ for $G$ using this assignment.

For every $\delta \in \{\alpha, \beta, \gamma\}$ and $i \in [n]$,
if $\pi(x^{\delta}_i)=\true$, then let $t^{\delta}_{2i}\in S$, and otherwise, $f^{\delta}_{2i-1}\in S$.
We also put $z_1,z_2$, $x^{\delta, \star}_i$, and $c^b_q$ into $S$ for all $i\in [n]$, $\delta\in \{\alpha, \beta, \gamma\}$, and $q\in [m]$.
Note that $|S|=k$.

Now, we show that $S$ is indeed a geodetic set of $G$.
First, $y_1,y_2,z_1,z_2,g_1$, and all the vertices of $A^{\alpha},A^{\beta},A^{\gamma}$ are covered by a shortest path between $z_1$ and $z_2$.
Then, for each $\delta \in \{\alpha,\beta,\gamma\}$ and $i\in [n]$, $x^{\delta, \triangleleft}_i$, $x^{\delta, \triangleright}_i$, $x^{\delta, \circ}_i$, and $x^{\delta, \star}_i$ are covered by a shortest path between $S\cap \{t^{\delta}_{2i},f^{\delta}_{2i-1}\}$ and $x^{\delta, \star}_i$. The vertex $g_3$ is covered by any shortest path between $c_q^b$ and $c_{q'}^b$, where $C_q$ and $C_{q'}$ are two clauses of $\psi$.
Suppose that $\pi(x^{\delta}_i)$, for some $i\in [n]$ and $\delta\in \{\alpha,\beta,\gamma\}$, satisfies some clause $C_q$.
By our construction, if $x^{\delta}_i$ appears positively (negatively, respectively) in $C_q$,
then $t^{\delta}_{2i}$ ($f^{\delta}_{2i-1}$, respectively) and $c^b_q$ are at distance four since $t^{\delta}_{2i}$ ($f^{\delta}_{2i-1}$, respectively) and $c^a_q$ have no common neighbor in $V^{\delta}$.
Moreover, there is a shortest path from $t^{\delta}_{2i}$ ($f^{\delta}_{2i-1}$, respectively) to $c^b_q$ of length four, covering $g_2,c_q,c^a_q$, and $c^b_q$; 
there is also a shortest path from $t^{\delta}_{2i}$ ($f^{\delta}_{2i-1}$, respectively) to~$c^b_q$ of length four, covering $v^{\delta}_{j}, v^{\delta}_{h},c^a_q$, and $c^b_q$,
where $v^{\delta}_{j}\in V^{\delta}$ is a vertex adjacent to $t^{\delta}_{2i}$ ($f^{\delta}_{2i-1}$, respectively) and $v^{\delta}_{h}$ is any vertex of $V^{\delta}$ that is not adjacent to $t^{\delta}_{2i}$ ($f^{\delta}_{2i-1}$, respectively).
Thus, every vertex of $V^{\delta}$ for $\delta \in \{\alpha,\beta,\gamma\}$ is covered by a shortest path between two vertices of $S$.
Since every clause of $\psi$ is satisfied by~$\pi$, it follows that every vertex of $\{c_q,c^a_q,c^b_q: q\in [m]\}$ is covered by a shortest path between two vertices of~$S$.
As a result, $S$ is a geodetic set of $G$.
\end{proof}

\begin{lemma}
\label{lemma:3-Part-3-SAT-Geo-Set-diam-tw-backword}
If $G$ admits a geodetic set of size $k$, then $\psi$ is a satisfiable \textsc{$3$-Partitioned-$3$-SAT} formula.
\end{lemma}
\begin{proof}
Suppose that $G$ has a geodetic set $S$ of size at most $k$.
By Observation~\ref{obs:pendant}, $z_1,z_2$, $x^{\delta, \star}_i$, and $c^b_q$ for all $i\in [n]$, $\delta\in \{\alpha, \beta, \gamma\}$, and $q\in [m]$ must be in any geodetic set $S$ of $G$.

\begin{claim}
\label{clm:3-Part-3-SAT-Geo-Set-diam-tw-backward-1}
For each $i\in [n]$ and $\delta \in \{\alpha,\beta,\gamma\}$,
exactly one of $t^{\delta}_{2i}$ and $f^{\delta}_{2i-1}$ must be in $S$.
\end{claim}
\begin{claimproof}
Since $S$ is a geodetic set, for each $i\in [n]$ and $\delta \in \{\alpha,\beta,\gamma\}$
$x^{\delta, \triangleleft}_i$ and $x^{\delta, \triangleright}_i$ must be covered by shortest paths between two vertices of $S$.
If $t^{\delta}_{2i}\in S$ ($f^{\delta}_{2i-1}\in S$, respectively), $x^{\delta, \triangleleft}_i$ and $x^{\delta, \triangleright}_i$ 
are covered by shortest paths between $t^{\delta}_{2i}\in S$ ($f^{\delta}_{2i-1}\in S$, respectively) and $x^{\delta, \star}_i$.
Suppose that, for some $i'\in [n]$ and $\delta' \in \{\alpha,\beta,\gamma\}$,
neither of $t^{\delta'}_{2i'}$ and $f^{\delta'}_{2i'-1}$ is in $S$. Moreover, if neither of $x^{\delta', \triangleleft}_{i'}$ and $x^{\delta', \triangleright}_{i'}$ is in $S$, then, due to the edges incident with $g_1$, no vertices in $S$ have a shortest path containing any of these two vertices. Similarly, if only one of $x^{\delta', \triangleleft}_{i'}$ and $x^{\delta', \triangleright}_{i'}$ is in $S$, then the other is not covered by $S$. Thus, if neither of $t^{\delta'}_{2i'}$ and $f^{\delta'}_{2i'-1}$ is in $S$, then both $x^{\delta', \triangleleft}_{i'}$ and $x^{\delta', \triangleright}_{i'}$ must be in $S$. 
Since $k-|\{z_1,z_2\} \cup \{x^{\delta, \star}_i: i\in [n], \delta\in \{\alpha, \beta, \gamma\}\} \cup \{c^b_q: q\in [m]\}|=3n$,
we conclude that exactly one of $t^{\delta}_{2i}$ and $f^{\delta}_{2i-1}$ must be in $S$ for each $i\in [n]$ and $\delta \in \{\alpha,\beta,\gamma\}$.
\end{claimproof}

By \cref{clm:3-Part-3-SAT-Geo-Set-diam-tw-backward-1} and earlier arguments, we now have that $|S|=k$.

\begin{claim}
\label{clm:3-Part-3-SAT-Geo-Set-diam-tw-backward-2}
For each $q\in [m]$, the vertex $c_q$ is covered either by a shortest path between $c^b_q$ and $t^{\delta}_{2i}$, where the variable $x^{\delta}_{i}$ appears positively in the clause $C_q$, or by a shortest path between $c^b_q$ and $f^{\delta}_{2i-1}$, where the variable $x^{\delta}_{i}$ appears negatively in the clause $C_q$. Moreover, $c_q$ is covered by no other type of shortest path between two vertices in $S$.
\end{claim}
\begin{claimproof}
By the construction of $G$, if the variable $x^{\delta}_{i}$ appears positively in the clause $C_q$, 
then there is a shortest path from $t^{\delta}_{2i}$ to $c^b_q$ of length four covering $g_2,c_q,c^a_q$, and $c^b_q$. 
If the variable $x^{\delta}_{i}$ appears negatively in the clause~$C_q$, 
then there is a shortest path from $f^{\delta}_{2i-1}$ to $c^b_q$ of length four covering $g_2,c_q,c^a_q$, and $c^b_q$.

Next, we show that $c_q$ is not covered by any shortest path between any other two vertices of $S$.
We can check that $c_q$ is not covered by any of the shortest paths between $z_1$ and $z_2$, 
between $z_j$ ($j\in \{1,2\}$) and $x^{\delta, \star}_i$ ($i\in [n], \delta\in \{\alpha, \beta, \gamma\}$), 
and between $z_j$ ($j\in \{1,2\}$) and $S\cap \{t^{\delta}_{2i},f^{\delta}_{2i-1}\}$ ($i\in [n], \delta\in \{\alpha, \beta, \gamma\}$).
Note that any shortest path from $z_j$ ($j\in \{1,2\}$) to $c^b_q$ ($q\in [m]$) is of length five, 
covering $y_j$, some vertex of $A^{\delta}$ ($\delta\in \{\alpha, \beta, \gamma\}$), some vertex of $V^{\delta}$, $c^a_q$, and~$c^b_q$.

We can check that 
$c_q$ is not covered by any of the shortest paths 
between $x^{\delta, \star}_{i}$ and $x^{\delta', \star}_{i'}$ ($i,i' \in [n], \delta,\delta' \in \{\alpha, \beta, \gamma\}$), 
and between $x^{\delta, \star}_{i}$ and $S\cap \{t^{\delta'}_{2i'},f^{\delta'}_{2i'-1}\}$ ($i,i' \in [n], \delta,\delta' \in \{\alpha, \beta, \gamma\}$).
Note that any shortest path from $x^{\delta, \star}_{i}$ ($i\in [n], \delta\in \{\alpha, \beta, \gamma\}$) to $c^b_q$ ($q\in [m]$) is of length five,
covering $x^{\delta, \circ}_{i}$, $g_1$, some vertex of $V^{\delta}$, $c^a_q$, and $c^b_q$.

Note that any shortest path between $c^b_q$ and $c^b_{q'}$ ($q,q'\in [m]$) is of length four, covering $c^a_q$, $g_3$, and $c^a_{q'}$. 

We can check that $c_q$ is not covered by any shortest paths 
between $S\cap \{t^{\delta}_{2i},f^{\delta}_{2i-1}\}$ and $S\cap \{t^{\delta'}_{2i'},f^{\delta'}_{2i'-1}\}$ ($i,i' \in [n], \delta,\delta' \in \{\alpha, \beta, \gamma\}$).

If the variable $x^{\delta}_{i}$ does not appear positively in the clause $C_q$,
then any shortest path between $c^b_q$ and $t^{\delta}_{2i}$ is of length three (because $c^a_q$ and $t^{\delta}_{2i}$ have a common neighbour in $V^{\delta}$), covering some vertex of $V^{\delta}$ and $c^a_q$, but not $c_q$. Similarly, if $x^{\delta}_{i}$ does not appear negatively in $C_q$, then any shortest path between $c^b_q$ and $f^{\delta}_{2i-1}$ is of length three and does not cover $c_q$.

By the case analysis above, the claim is true.
\end{claimproof}

By Claim~\ref{clm:3-Part-3-SAT-Geo-Set-diam-tw-backward-1}, 
exactly one vertex of $t^{\delta}_{2i}$ and $f^{\delta}_{2i-1}$ belongs to $S$ for each $i\in [n]$ and $\delta \in \{\alpha,\beta,\gamma\}$.
We define an assignment $\pi$ to the variables of $\psi$ as follows. 
For each $i\in [n]$ and $\delta \in \{\alpha,\beta,\gamma\}$,
if $t^{\delta}_{2i}\in S$, then $\pi(x^{\delta}_{i})=\true$.
Otherwise, $\pi(x^{\delta}_{i})=\false$.
Since $S$ is a geodetic set for $G$, 
every vertex $c_q$ ($q\in [m]$) is covered by a shortest path between two vertices of $S$.
By Claim~\ref{clm:3-Part-3-SAT-Geo-Set-diam-tw-backward-2}, 
every vertex $c_q$ ($q\in [m]$) is covered by a shortest path between $S\cap \{t^{\delta}_{2i},f^{\delta}_{2i-1}\}$ and $c^b_q$,
where the variable $x^{\delta}_{i}$ appears in the clause $C_q$.
It follows that every clause $C_q$ is satisfied by $\pi(x^{\delta}_{i})$.
As a result, $\psi$ is a satisfiable \textsc{$3$-Partitioned-$3$-SAT} formula.
\end{proof}

\begin{proof}[Proof of \Cref{thm:lower-bound-diam-tw-GS}.]
First, it is not hard to check that the diameter of $G$ is at most 5.
Then, let $X=V^{\alpha} \cup V^{\beta} \cup V^{\gamma} \cup \{g_1,g_2,g_3,y_1,y_2\}$.
We can check that every component of $G\setminus X$ has at most six vertices and $|X|=\OO(\log n)$.
Thus, the treewidth $\tw(G)$ --- in fact, even the treedepth $\td(G)$ --- of $G$ is bounded by $\OO(\log n)$.
By the description of the reduction, it takes polynomial time to compute the reduced
instance.
Hence, if there is an algorithm for \gsfull 
that runs in time $2^{f(\diam)^{o(\tw)}}$ (or $2^{f(\diam)^{o(td)}}$)
then, there is an algorithm running in time $2^{o(n)}$ for \textsc{$3$-Partitioned-$3$-SAT}, which by Proposition~\ref{prop:3-SAT-to-3-Partition-3-SAT}, contradicts the \ETH.
\end{proof}

\section{\smdfull: Lower Bound Regarding Vertex Cover}
\label{sec:strong-met-dim-lower-bound-vc}

The aim of this section is to prove
the following theorem.
\begin{theorem}
\label{thm:smd-lower-bound-vc}
Unless the \ETH\ fails, \smdfull does not admit:
\begin{itemize}
\item an algorithm 
running in time $2^{2^{o(\vc)}} \cdot n^{\OO(1)}$ for any computable function 
$f:\mathbb{N} \mapsto \mathbb{N}$, nor
\item a kernelization algorithm returning a kernel with $2^{o(\vc)}$ vertices.
\end{itemize}
\end{theorem}

To this end, we present a reduction from 
\textsc{Exact-$3$-Partitioned-$3$-SAT} (see \Cref{sec:preliminaries})
to \textsc{Strong Metric Dimension}.
We use the relation
between \textsc{Strong Metric Dimension}
and the \textsc{Vertex Cover} problem,
which was established in~\cite{OP07},
to prove the theorem.
We need the following definition in order
to state the relationship.
\begin{definition}
\label{def:mutually-max-distant}
Given a graph $G$, we say a vertex $u\in V(G)$ is maximally distant from $v\in V(G)$ if
there is no $x \in V(G) \setminus \{u\}$ such that a shortest path between $x$ and $v$ contains $u$.
Formally, for every 
$y \in N(u)$, we have
$d(y, v) \le d(u, v)$.
If $u$ is maximally distant from $v$, and
$v$ is maximally distant from $u$, 
then $u$ and $v$ are \emph{mutually maximally 
distant} in $G$, and we write $u \bowtie v$.
\end{definition}
For any two mutually maximally distant
vertices in $G$, there is no vertex in $G$ 
that strongly resolves them, except themselves. 
Hence, if $u \bowtie v$ in $G$, 
then, for any strong resolving set $S$ of 
$G$, at least one of $u$ or $v$ is in $S$,
i.e., $|\{u, v\}\cap S|\geq 1$.
Oellermann and Peters-Fransen~\cite{OP07}
showed that
this necessary condition
is also sufficient.
Consider an auxiliary graph $G_{SR}$
of $G$ defined as follows.
\begin{definition}
\label{def:strong-resolving-graph}
Given a connected graph $G$, the 
\emph{strong resolving graph} of $G$, 
denoted by $G_{SR}$, has vertex set 
$V(G)$ and two vertices $u, v$ are adjacent
if and only if
$u$ and $v$ are mutually maximally 
distant in $G$, i.e., $u \bowtie v$.
\end{definition}
\begin{proposition}[Theorem 2.1 in \cite{OP07}]
\label{prop:strong-met-dim-equiv-vc}
For a connected graph $G$,
$\smd(G) = \vc(G_{SR})$.
\end{proposition}
In light of the above proposition,
it is sufficient to prove the following lemma.
\begin{lemma}
\label{lemma:SAT-SMD-VC}
There is a polynomial-time reduction that,
given an instance $\psi$ of 
\textsc{Exact-$3$-Partitioned-$3$-SAT} 
on $3n$ variables,
returns an equivalent instance
$(G, k)$ of \textsc{Strong Metric Dimension}
such that 
$\vc(G) = \calO(\log(n))$
and $\vc(G_{SR}) = k$.
\end{lemma}

Recall the textbook reduction from
\textsc{$3$-SAT} to \textsc{Vertex Cover}
\cite{DBLP:books/daglib/0015106}.
In this, we add matching edges corresponding 
to variables, and vertex-disjoint triangles
corresponding to clauses.
Finally, we add edges between two vertices
corresponding to the same literal 
on 
the ``variable-side'' and the ``clause-side''.
We adopt the same reduction for
\textsc{Exact-$3$-Partitioned $3$-SAT} to 
\textsc{Vertex Cover}.

\begin{figure}[h]
\begin{center}
\includegraphics[scale=1]{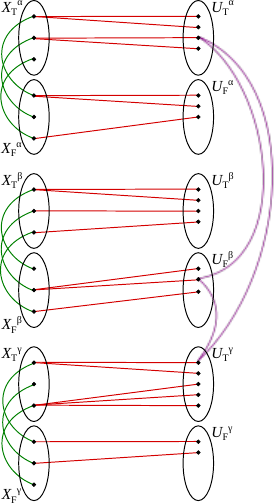}
\end{center}
\caption{
Overview of the reduction from
\textsc{Exact-$3$-Partitioned-$3$-SAT} to
\textsc{Vertex Cover}.
In Step~\ref{step:add-ind-set}, we add
all the independent sets mentioned.
In Step~\ref{step:add-variable-edges}, we
add matching green edges corresponding
to the assignment of the variables.
In Step~\ref{step:add-clause-edges},
we add purple triangles 
corresponding to clauses.
For example, the {purple} 
triangle corresponds to the 
clause $(x^{\alpha}_{3} \lor \neg x^{\beta}_2
\lor x^{\gamma}_{1})$.
In Step~\ref{step:add-variable-clause-edges},
we add red edges connecting the vertices
corresponding to the same literal
on the ``variable-side'' and ``clause-side''.
\label{fig:3-Par-3-SAT-to-VC}}
\end{figure}

Given an 
instance $\psi$ of 
\textsc{Exact-3-Partitioned-3-SAT},
with $m$ clauses and $3n$ variables
partitioned equally into $X^{\alpha}, X^{\beta}, X^{\gamma}$,
we construct the graph $H$ as follows.
\begin{enumerate}
\item \label{step:add-ind-set}
We rename the variables in $X^{\alpha}$
to $x^{\alpha}_{i}$ for $i \in [n]$. Analogously, we do this
for $X^{\beta}$ and $X^{\gamma}$.
For every $x^{\alpha}_i$,
we add two vertices $x^{\alpha}_{i, t}$
and $x^{\alpha}_{i, f}$
for the positive and negative literal, respectively.
Define $X^{\alpha}_{T}$ and $X^{\alpha}_{F}$
as the collection of vertices 
corresponding to the positive and negative 
literals of the variables
in $X^{\alpha}$, respectively.
We define the sets 
$X^{\beta}_{T}, X^{\beta}_F, X^{\gamma}_T,$
and $X^{\gamma}_F$, similarly.

Consider a clause $C_q = (x^{\alpha}_i \lor
\neg x^{\beta}_j \lor x^{\gamma}_{\ell})$.
We add three vertices: $x^{\alpha, q}_{i, t}$, 
$x^{\beta, q}_{j, f}$, and 
$x^{\gamma, q}_{\ell, t}$,
and the edges to make it a triangle.
Let $U^{\alpha}_T$ be the collection
of vertices corresponding to the 
positive literals of
variables in $X^{\alpha}$.
Formally, $U^{\alpha}_T = \{x^{\alpha, q}_{i, t} \mid$ 
there is a clause $C_q$
that contains the positive literal of the variable $x^{\alpha}_i\}$.
Note that, for a variable $x^{\alpha}_i$,
there may be multiple vertices corresponding 
to its positive literal in $U^{\alpha}_T$ depending on the number of clauses that contain $x^{\alpha}_i$.
We analogously define $U^{\alpha}_F$,
$U^{\beta}_T$, $U^{\beta}_F$, $U^{\gamma}_T$, and $U^{\gamma}_F$.
See \cref{fig:3-Par-3-SAT-to-VC}
for an illustration.

\item 
\label{step:add-variable-edges}
We add matching edges across 
$X^{\alpha}_T$ and $X^{\alpha}_F$
connecting $x^{\alpha}_{i, t}$
and $x^{\alpha}_{i, f}$ for every $i \in [n]$.
We add similar matching edges across
$X^{\beta}_T$ and $X^{\beta}_F$, and
$X^{\gamma}_T$ and $X^{\gamma}_F$. 
See the {green} edges in \cref{fig:3-Par-3-SAT-to-VC}.

\item 
\label{step:add-clause-edges}
As $\psi$ is an instance of 
\textsc{Exact-$3$-Partitioned-$3$-SAT}, 
for any clause there is a triangle that
contains {exactly} one vertex from each of
$U^{\alpha}_T \cup U^{\alpha}_F$,
$U^{\beta}_T \cup U^{\beta}_F$, and
$U^{\gamma}_T \cup U^{\gamma}_F$.
See the purple triangle in \cref{fig:3-Par-3-SAT-to-VC}.
For each clause, we add the edges to form its
corresponding triangle.

\item
\label{step:add-variable-clause-edges}
Finally, we add edges connecting 
a vertex corresponding to a literal on the
variable-side to vertices corresponding
to the same literal on the clause-side.
See the red edges in \cref{fig:3-Par-3-SAT-to-VC}.
\end{enumerate}

The reduction returns $(H, k)$ as the 
reduced instance of \textsc{Vertex Cover}, where $k = 3n + 2m$.
The correctness of the following 
lemma is spelled out in \cite{DBLP:books/daglib/0015106}.
\begin{lemma}
\label{lemma:Ex-3-Par-3-SAT-to-VC}
The formula $\psi$, with $3n$ variables and $m$ clauses, is a satisfiable 
\textsc{Exact-3-Partitioned-3-SAT} formula
if and only if $H$ admits a vertex cover
of size $k = 3n + 2m$.
\end{lemma}

In what remains of this section,
our objective is to construct $G$ such that
$G_{SR}$ is as ``close'' to $H$ as possible.
It will be helpful to think about $V(H)$ as a subset of $V(G) = V(G_{SR})$.
We want to construct $G$ such that
all the edges in $E(H)$ are present
in $G_{SR}$, while no undesirable edge 
appears in $G_{SR}$.
We use a \emph{set representation gadget}
and a \emph{bit representation gadget}
for the first and the second part, respectively.
However, for the second part, we need another trick since we may have some undesirable edges.
We ensure that all the edges in
$E(G_{SR}) \setminus E(H)$, i.e.,
undesirable edges, are 
incident to a clique, say $Z$, in $G_{SR}$.
Moreover, there is a vertex $z$ in $Z$ 
such that $N[z] = Z$, i.e., $z$ is not
adjacent to any vertex in $V(G_{SR}) \setminus Z$.
Then, without loss of generality,
we can assume that any vertex cover 
of $G_{SR}$ contains $Z \setminus \{z\}$.  
Hence, all the undesired edges are deleted 
by a pre-processing step while finding
the vertex cover of $G_{SR}$.
In other words, $E(H) = E(G_{SR} - (Z \setminus \{z\}))$. 
This ensures that  
the difficulty of finding a strong resolving
set in $G$ is encoded in finding a vertex cover in $G_{SR} - Z$, a graph with only 
desired edges.
We note the following easy observation
before presenting the primary tools used in 
the reduction.
\begin{observation}
\label{obs:pendant-in-G-to-clique-in-GSR}
Consider a connected graph $G$ that has
at least $3$ vertices.
Suppose $Z \subseteq V(G)$ is the collection
of all the pendent vertices in $G$.
Then, 
$Z$ is a clique in $G_{SR}$, and
every vertex in $N(Z)$ is an isolated
vertex in $G_{SR}$.
\end{observation}
\begin{proof}
Consider any two vertices $z_1, z_2$ in $Z$, and let $u_2$ be the unique neighbor of $z_2$.
It is easy to see that $d(z_1, u_2) < d(z_1, z_2)$.
Hence, $z_2$ is maximally distant from $z_1$.
Using the symmetric arguments, $z_1$ is 
maximally distant from $z_2$, and hence,
$z_1 \bowtie z_2$.
By the construction of $G_{SR}$, there is 
an edge with endpoints $z_1, z_2$.
As these were two arbitrary points in $Z$, we have that
$Z$ is a clique in $G_{SR}$.

Consider an arbitrary vertex 
$v \in V(G) \setminus \{u_2\}$.
If $v \neq z_2$, then $d(v, z_2) > d(v, u_2)$ 
and hence, $u_2$ is \emph{not} maximally distant 
from $v$.
This implies that $u_2$ is not adjacent with
any vertex in $V(G_{SR}) \setminus \{z_2\}$
in $G_{SR}$.
As $G$ is connected and has at least three 
vertices, $u_2$ is not maximally distant from 
$z_2$ either.
Hence, $u_2$ is an isolated vertex in 
$G_{SR}$.
Since $u_2$ is an arbitrary vertex in $N(Z)$, 
the second part of the claim follows.
\end{proof} 

\subsection{Preliminary Tools}
\label{subsec:prelim-3-Par-3-SAT-Strong-Met-Dim}
\begin{figure}[t]
    \centering
\includegraphics[scale=1.05]{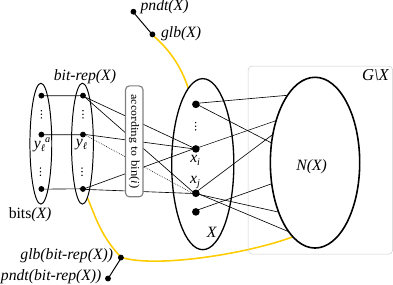}\hfill
\includegraphics[scale=1.05]{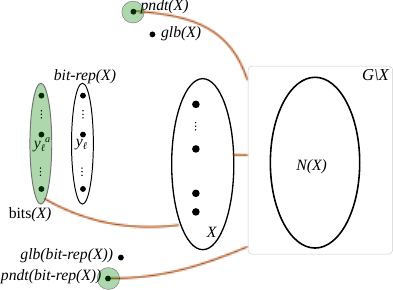}
    \caption{ \textbf{Set Identifying Gadget}.
   The graph $G$ is on the left side, and 
   the corresponding $G_{SR}$ is on the right
   side. The green-shaded region in $G_{SR}$
   denotes a clique. Note that the {brown} edges in $G_{SR}$ are not relevant at this time.
    \label{fig:strong-met-dim-set-identifying-gadget}}
\end{figure}

\subsubsection{Bit Representation Gadget to Add Independent Sets}
\label{subsubsec:strong-met-dim-ind-set}

In this subsection, we accomplish 
Step~\ref{step:add-ind-set} of the reduction
mentioned at the start of the section.
Formally, given a graph $G'$ and an independent set $X\subseteq V(G')$ of its vertices,
we want to add vertices and edges to $G'$ to obtain a graph $G$ such that 
$X$ remains an independent set
in $G$, and $X$ is also an independent set in $G_{SR}$.
We do this as follows:
\begin{itemize}
\item First, let $X=\{x_i\mid i\in [|X|]\}$, 
and set $q := \lceil \log(|X|+ 1) \rceil$. 
We select this value for $q$ to
uniquely represent each integer in $[|X|]$ by 
its bit-representation in binary (note that we 
start from $1$ and not $0$).
\item For every $\ell \in [q]$, add two 
vertices: $y^a_{\ell}$ and $y_{\ell}$, and the edge 
$(y^a_{\ell}, y_{\ell})$.
We denote $\bitrep(X)=\{y_{\ell}\mid \ell \in [q]\}$ and $\bits(X)=\{y^a_{\ell} \mid \ell\in [q]\}$ 
for convenience in a later case analysis.
Note that both $\bitrep(X)$ and $\bits(X)$ are independent sets
and $\bits(X)$ is a collection of pendent
vertices in $G$ whose neighborhoods are in 
$\bitrep(X)$.
\item For every integer $j \in [|X|]$, let $\bit(j)$ denote the binary
representation of $j$ using $q$ bits.
Connect $x_j$ with $y_{i}$
if the $i^{th}$ bit (going from left to right) in $\bit(j)$ is $1$.
\item Add a vertex, denoted by $\glb(X)$,
and make it adjacent to every vertex in $X$.
Add another vertex, denoted by $\pndt(X)$,
and make it adjacent only to $\glb(X)$.

\item Similarly, add a vertex 
$\glb(\bitrep(X))$ which is adjacent to every 
vertex in $\bitrep(X)$ ,
and a vertex $\pndt(\bitrep(X))$ which is 
adjacent only to $\glb(\bitrep(X))$.
\end{itemize}
This completes the construction of $G$. 
We use $\glb(\cdot)$ and $\pndt(\cdot)$
as a function to denote vertices
that are adjacent to every vertex in the set, i.e., global to the set, 
and the vertex pendent to the global
vertex of the set, respectively.
We mention the small caveat in this notation.
Pendent vertices in $ \bits(X) \cup \{\pndt(X)\}$ are \emph{not} adjacent to any vertex in $X$.
It is helpful to think of these pendent 
vertices together with $\bitrep(X)\cup \{\glb(X), \pndt(\bitrep(X)), \glb(\bitrep(X))\}$ as vertices added for $X$. 
Moreover, at a later stage in the reduction,
we may make $\glb(X)$ global to every
vertex in a new set $Y$.
However, we do \emph{not} rename it 
to $\glb(X \cup Y)$ for notational clarity.
Note that no vertex in 
$\bitrep(X) \cup \bits(X)$ is adjacent to any vertex in
$V(G)\setminus (\bitrep(X) \cup \bits(X) \cup X \cup \glb(\bitrep(X))$.
See \cref{fig:strong-met-dim-set-identifying-gadget} for an illustration.

We use this gadget as a building block
for our reduction and do not add 
any more edges whose both endpoints are
completely inside the gadget.
While adding other vertices and edges,
we ensure that we do not add a vertex
whose neighborhood intersects both $X$ 
and $\bitrep(X)$.
We now show that the following property 
(we aimed for) holds under this condition.

\begin{lemma}
\label{lemma:strong-met-dim-add-ind-set}
Consider the graph $G$, an independent set $X$, 
and the vertices and edges as defined 
above.
Suppose there is no vertex $v\in V(G)$ that is adjacent
to both a vertex in $X$ and
a vertex in $\bitrep(X)$.
Then,
$X$ is an independent set in $G_{SR}$.
\end{lemma}
\begin{proof}
Consider any two vertices, say $x_i, x_j$, in 
$X$.
We want to prove that these two 
vertices are not adjacent in $G_{SR}$.
By the construction, it is sufficient to
prove that either $x_i$ is not maximally disjoint
from $x_j$, or $x_j$ is not maximally disjoint 
from $x_i$.
As the bit-representation of $i$ is not the same
as $j$, there is a vertex, say $y_{\ell}$ in 
$\bitrep(X)$, that is adjacent to 
$x_j$, but not to $x_i$
(or vice-versa).
We consider the first case. 
Note that $d(x_i, x_j) = 2$ since
$\glb(X)$ is adjacent to both $x_i$ and $x_j$, and $X$ is an independent set.
Since, there is no vertex $v \in V(G)$ 
such that $N(v) \cap X \neq \emptyset$ and 
$N(v) \cap \bitrep(X) \neq \emptyset$, and 
$\bitrep(X)$ is an independent set, we have 
$d(x_i, y_{\ell}) > 2$. 
Thus, $d(x_i, y_{\ell}) = 3$.
Hence, there exists a vertex in $N(x_j)$ 
which is farther from $x_i$.
This implies that $x_j$ is not maximally 
distant from $x_i$, concluding the proof.
\end{proof}

\subsubsection{Set Representation Gadget to Add Edges}
\label{subsubsec:strong-met-dim-add-edges}

We use a set representation gadget to 
add edges across two independent sets in $G_{SR}$.
This will be useful to accomplish 
Steps~\ref{step:add-variable-edges},~\ref{step:add-clause-edges}, and~\ref{step:add-variable-clause-edges}.

Consider two independent sets $A$ and $B$ in $G'$,
and suppose there is a 
function $\phi: B \mapsto A$.
This function may not be one-to-one.
As we are defining this function,
it will be helpful to consider $A, B$ as 
an ordered pair.
Our objective is to add vertices and edges to $G'$ to obtain $G$
such that $G_{SR}$ contains an edge $(a_i, b_j)$ for some $a_i \in A$ and $b_j \in B$
if and only if $a_i = \phi(b_j)$.
Moreover, we want $A, B$ to remain as independent sets in $G$ and $G_{SR}$.

We change the function $\phi$
as per our requirements.
In Step~\ref{step:add-variable-edges},
we need to add the edge between 
$x^{\alpha}_{i, t}  \in X^{\alpha}_T$ 
and $x^{\alpha}_{i, f} \in X^{\alpha}_F$, and 
other such pairs.
In this case, we define 
$\phi(x^{\alpha}_{i, f}) = x^{\alpha}_{i, t}$.
Here, $\phi$ is a one-to-one and 
onto function from $B$ to $A$.
Consider a clause $C_{q} = (x^{\alpha}_{i}
\lor \neg x^{\beta}_j \lor x^{\gamma}_{\ell})$.
As mentioned before, in
the reduction, we add the vertices $x^{\alpha, q}_{i, t}$ to $U^{\alpha}_T$,
$x^{\beta, q}_{j, f}$ to $U^{\beta}_T$, and
$x^{\gamma, q}_{\ell, t}$ to $U^{\gamma}_T$.
We expect a triangle with these three vertices.
Hence, in Step~\ref{step:add-clause-edges},
while adding edges across $U^{\alpha}_T$
and $U^{\beta}_F$,
we define the function $\phi$ as
$\phi(x^{\beta, q'}_{j, f}) = x^{\alpha, q}_{i, t}$ if and only if $q = q'$,
i.e., if the literals corresponding 
to these two vertices appear in a clause.
For Step~\ref{step:add-variable-clause-edges},
while adding edges across $X^{\alpha}_T$
and $U^{\alpha}_T$, 
define $\phi(x^{\alpha, q'}_{i', t}) 
= x^{\alpha}_{i, t}$ if and only if
$i = i'$, i.e., if these two vertices 
correspond to the same literal.
In this case, multiple vertices in 
$U^{\alpha}_T$ may be assigned to 
a single vertex in $X^{\alpha}_T$
if they correspond to the same literal in $X^{\alpha}_T$.
For example, $\phi(x^{\alpha, q}_{i, t}) =
\phi(x^{\alpha, q'}_{i, t}) = x^{\alpha}_{i, t}$. 

\begin{figure}[t]
    \centering
        \includegraphics[scale=1]{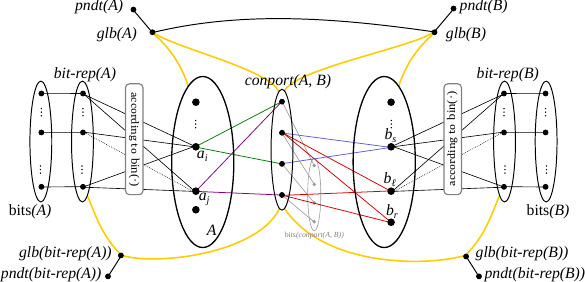}\hfill
        \includegraphics[scale=1]{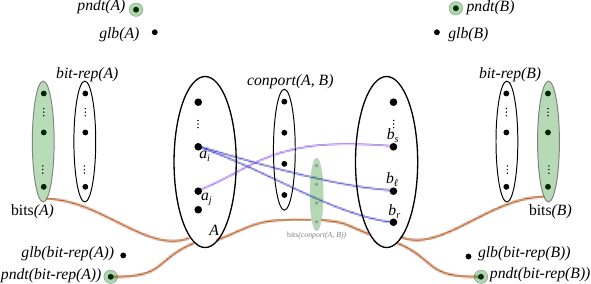}
        \caption{Set Representation Gadget
        to add edges across independent
        sets $A$ and $B$. The graph $G$ is above and the graph $G_{SR}$ is below. The yellow lines from a vertex to a set represent that the vertex is adjacent to every vertex in the set.
        Suppose that $\phi(b_{\ell}) = \phi(b_r) = a_i$ and $\phi(b_s)=a_j$. Note that in $G$, due to $\phi$, $a_i$ shares no common neighbor in $\conport(A, B)$ with $b_{\ell}$ nor $b_r$, and that $a_j$ shares no common neighbor in $\conport(A, B)$ with $b_s$. Furthermore, in $G$, we have $a_i \bowtie b_{\ell}$, $a_i \bowtie b_r$, and $a_j \bowtie b_s$, which create the edges between these pairs of vertices in $G_{SR}$. Lastly, the green-shaded region in $G_{SR}$ denotes a clique.
        \label{fig:strong-met-dim-add-edges}}
\end{figure}

We use \emph{set representations}
of integers to achieve this objective and recall some ideas from Section~\ref{subsubsec:set-representation-MD}.
For a positive integer $p$,  define $\calF_p$ 
as the collection of
subsets of $[2p]$ that contains exactly $p$ 
integers.
We critically use the fact that no set in $\calF_p$
is contained in any other set in $\calF_p$
(such a collection of sets are called a \emph{Sperner family}).
This implies for any two different sets $A, B \in \calF_p$,
$A$ intersects the complement of $B$, i.e.,
$A \cap ([2p] \setminus B) \neq \emptyset$.
Let $n$ be a positive integer such that $n \leq \binom{2p}{p}$.
We define
$\setrep: [n] \mapsto \calF_p$ as a one-to-one function
by arbitrarily assigning a set in $\calF_p$ to an integer in $[n]$.
By the asymptotic estimation of the central binomial coefficient, $\binom{2p}{p}\sim \frac{4^p}{\sqrt{\pi \cdot p}}$~\cite{Sperner}.
To get the upper bound of $p$, we scale down the asymptotic function and have $n \leq \frac{4^p}{2^p}=2^p$.
Thus, $p=\OO(\log n)$. 
Given independent sets $A = \{a_1, \dots, a_n\}$, $B = \{b_1, \dots, b_{n'}\}$ in $G'$, where $n' \le n$ and the function
$\phi: B \mapsto A$, we add 
vertices and edges to $G'$ to obtain $G$ as follows.
\begin{itemize}
\item We add the sets of vertices
$\bitrep(A)$, $\bits(A)$, 
the vertices $\glb(A)$, $\pndt(A)$, 
$\glb(\bitrep(A))$, $\pndt(\bitrep(A))$,
and the appropriate edges 
as mentioned in the previous subsection.
Similarly, we add the corresponding vertices
and edges with respect to $B$.
\item We add the edge $(\glb(A), \glb(B))$. 
\item We add a \emph{connection portal},
denoted by $\conport(A, B) =
\{v_1, v_2, \dots,  v_{2p}\}$.
For every vertex $v_{p'}$ in $\conport(A, B)$,
we add a new vertex and make it adjacent to
$v_{p'}$. 
The collection of these pendent vertices
are denoted by $\bits(\conport(A, B))$.
\item For every $i \in [n]$ and for every $p' \in \setrep(i)$, we
add the edge $(a_i,  v_{p'})$.
If $\phi(b_j) = a_i$ for some $j \in [n']$, then we add the edge $(b_j, v_{p'})$
for every $p'$ \emph{not} in $\setrep(i)$.
This ensures that, for every pair $a_i, b_j$,
if $a_i = \phi(b_j)$, then 
there is no vertex in $\conport(A, B)$
that is adjacent to both $a_i$ and $b_j$.
\item Finally, we make every vertex in 
$\conport(A, B)$ adjacent to $\glb(A)$, $\glb(\bitrep(A))$, $\glb(B)$, and
$\glb(\bitrep(B))$.
\end{itemize}

This completes the construction of $G$. See Figure~\ref{fig:strong-met-dim-add-edges} for an illustration.
As in the previous case,
we use this gadget as a building block
for our reduction and do not add 
any more edges whose both endpoints are
completely inside the gadget.
While adding other vertices and edges,
we ensure that we do not add another 
vertex (outside $\conport(A, B)$)
whose neighbourhood intersects both $A$ 
and $B$.
We now show that the following property 
(we aimed for) holds under this condition.

\begin{lemma}
\label{lemma:strong-met-dim-add-edges}
Consider the graph $G$, independent sets
$A, B$, connection portal $\conport(A, B)$ added with respect to the
function $\phi:B \mapsto A$,
and the vertices and edges as defined above.
For every vertex $v \in V(G) \setminus (\conport(A, B))$, suppose the following conditions are true.
\begin{itemize}
\item It is \emph{not} adjacent to both a vertex in $A$ and a vertex in $B$;
\item it is \emph{not} adjacent to both a vertex in $B$ and a vertex in $\bitrep(A)$;
\item it is \emph{not} adjacent to both a vertex in $A$ and a vertex in $\bitrep(B)$.
\end{itemize}
Then, the edge $(a_i, b_j)$ is present in $G_{SR}$ if and only if $a_i = \phi(b_j)$.
\end{lemma}
\begin{proof}
Consider the vertices $a_i, a_j \in A$ and 
$b_{\ell}, b_{r}, b_s \in B$ such that
$\phi(b_{\ell}) = \phi(b_r) = a_i$
and $\phi(b_s) = a_j$ (see \Cref{fig:strong-met-dim-add-edges}).
We focus on the vertex $a_i$.
Note that every vertex in $\conport(A, B)$ 
is at distance $1$ (if its in $N(a_i)$) or
at distance $2$ (as $\glb(A)$ is adjacent 
to every vertex in $A \cup \conport(A, B)$) from $a_i$.
By the construction,
the vertex $\glb(\bitrep(B))$ is at distance
$2$ from $a_i$, and hence, every vertex
in $\bitrep(B)$ is at distance at most $3$ from $a_i$.
Because of the edge $(\glb(A), \glb(B))$
and the global nature of the endpoints of
this edge, every vertex in $B$ is 
at distance at most $3$ from $a_i$.
Since $\phi(b_{\ell}) = a_i$,
as mentioned before, there is no vertex
in $\conport(A, B)$ that is adjacent to both
$a_i$ and $b_{\ell}$.
Moreover, there is no vertex $v$ in 
$V(G) \setminus \conport(A, B)$,
such that $N(v) \cap A \neq \emptyset$
and $N(v) \cap B \neq \emptyset$.
Hence, by the construction, $b_{\ell}$ is
at distance $3$ from $a_i$.
Hence, for every vertex $x \in N(b_{\ell})$,
we have $d(a_i, x) \le d(a_i, b_{\ell})$.
This implies that $b_{\ell}$ is maximally
distant from $a_i$.
As the gadget constructed is symmetric,
it is easy to see that $a_i$
is also maximally distant from $b_{\ell}$.
Hence, we have $a_i \bowtie b_{\ell}$
and $(a_i, b_{\ell})$ is an edge in 
$G_{SR}$.
Using similar arguments, the
edges $(a_i, b_r)$ and $(a_j, b_s)$ are
present in $G_{SR}$.

Now, consider $b_s \in B$ and $a_i \in A$
such that $\phi(b_s) \neq a_i$.
Then, by the properties of the set representation
gadget, there exists a vertex in 
$\conport(A,B)$ that is adjacent to 
both $a_i$ and $b_s$.
Hence, $d(a_i, b_s) = 2$.
The neighbors of $b_s$ in $\bitrep(B)$
are at distance at most $3$ from $a_i$.
However, as there is no vertex $v \in V(G)$
such that $N(v) \cap A \neq \emptyset$
and $N(v) \cap \bitrep(B) \neq \emptyset$,
every neighbor of $b_s$ in $\bitrep(B)$
is at distance exactly $3$ from $a_i$.
Hence, $b_s$ is not maximally distant 
from $a_i$.
This implies that there is no edge with endpoints $a_i, b_s$ in $G_{SR}$.
Hence, the edge $(a_{i'}, b_{j'})$ is 
present in $G_{SR}$ if and only if 
$\phi(b_{j'}) = a_{i'}$.
\end{proof}

Before presenting the reduction,
we note that there is a different schematic representation
of the gadgets mentioned above in 
\cref{fig:strong-met-dim-gadget-representation}.

\begin{figure}[t]
    \centering
        \includegraphics[scale=1]{./images1/setrep1.pdf}
        
        \vspace{0.3cm}
        \includegraphics[scale=0.7]{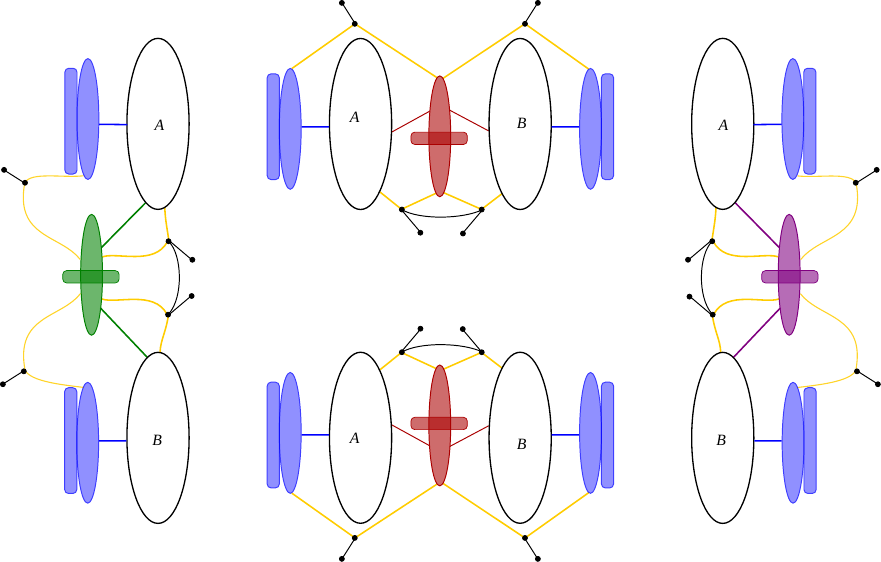}
        \caption{Set Representation Gadget
        to add edges across independent
        sets $A$ and $B$, and four different
        schematic representations of the same below. 
       In the schematic representation,
       yellow thick edges denote that 
       the vertex is adjacent with every
       vertex in the set. 
       The blue filled oval shape corresponds to 
       $\bitrep(\cdot)$, and the blue thick lines denote that the set, say $A$, is connected to $\bitrep(A)$ 
       according to its bit representation
       as mentioned in \cref{subsubsec:strong-met-dim-ind-set}.       
       The filled oval shape of green, red,
       or purple color denotes $\conport(\cdot, \cdot)$.
       The filled rectangle shape denotes
       $\bits(\cdot)$.
       The colors of connection ports correspond to the edges mentioned in Figure~\ref{fig:3-Par-3-SAT-to-VC}.
       \label{fig:strong-met-dim-gadget-representation}}
\end{figure}

\subsection{Reduction}

Consider an instance $\psi$ of \textsc{Exact-3-Partitioned-3-SAT}
with $m$ clauses and $3n$ vertices that are partitioned into
$X^{\alpha}, X^{\beta}, X^{\gamma}$.
Recall the reduction mentioned before Lemma~\ref{lemma:Ex-3-Par-3-SAT-to-VC}.
The reduction we present here, constructs a graph $G$ 
using the construction specified earlier, according to the steps mentioned in 
the reduction.
\begin{itemize}
\item Recall the sets defined in 
Step~\ref{step:add-ind-set}.
For every $A \in \{X^{\delta}_T, 
X^{\delta}_F, U^{\delta}_T, U^{\delta}_F \mid \delta \in \{\alpha, \beta,\gamma\}\}$,
we add the sets $A$, $\bitrep(A)$, $\bits(A)$, the vertices $\glb(A)$, $\pndt(A)$, 
$\glb(\bitrep(A))$, and 
$\pndt(\bitrep(A))$, and
the associated edges as mentioned 
in \cref{subsubsec:strong-met-dim-ind-set}.

\item 
For the edges mentioned in Step~\ref{step:add-variable-edges}, 
for every $\delta \in \{\alpha, \beta, \gamma\}$,
we assign 
$A = X^{\delta}_T$ and $B = X^{\delta}_F$,
and define $\phi:B \mapsto A$ as
$\phi(x^{\delta}_{i, f}) = x^{\delta}_{i, t}$.
Then, we add the connection portal 
$\conport(X^{\delta}_T, X^{\delta}_F)$
and the other vertices and edges mentioned
in \cref{subsubsec:strong-met-dim-add-edges}.

\item 
For the edges mentioned in Step~\ref{step:add-clause-edges},  for $\delta \in \{\alpha, \beta, \gamma\}$,
define $U^{\delta} = U^{\delta}_T \cup U^{\delta}_F$.
First, we assign 
$A = U^{\alpha}$ and $B = U^{\beta}$,
and define $\phi: B \mapsto A$ as follows:
for every clause $C_q = (x^{\alpha}_i \lor
\neg x^{\beta}_j \lor x^{\gamma}_{\ell})$,
$\phi(x^{\beta, q}_{j, f}) = x^{\alpha, q}_{i, t}$.
We add $\conport(U^{\alpha}, U^{\beta})$ and the corresponding vertices 
and edges as specified in 
\cref{subsubsec:strong-met-dim-add-edges}.
We let $\bitrep(U^{\alpha}) = \bitrep(U^{\alpha}_T) \cup \bitrep(U^{\alpha}_F)$
and 
$\bitrep(U^{\beta}) = \bitrep(U^{\beta}_T) \cup \bitrep(U^{\beta}_F)$.
We repeat the process for the pairs
$(U^{\beta}, U^{\gamma})$ and 
$(U^{\gamma}, U^{\alpha})$.
As $\psi$ is an instance of
\textsc{Exact-3-Partitioned-3-SAT},
every clause contains exactly three 
variables, and hence, in each case, $
\phi$ is well-defined.
We remark that, for every $U^{\delta}$, there are two global 
vertices now. 
We denote them by $\glb(U^{\delta})$
and $\glb^{\circ}(U^{\delta})$.
However, we do \emph{not} add another $\bitrep(U^{\delta})$.

Figure~\ref{fig:strong-met-dim-notfull} shows the vertices and edges added 
so far. 

\item For the edges mentioned in Step~\ref{step:add-variable-clause-edges},
we first consider $A = X^{\alpha}_T$ 
and $B = U^{\alpha}_T$, and 
define the function $\phi: B \mapsto A$
as follows:
for every $x^{\alpha, q}_{i, t}$ in 
$U^{\alpha}$ for some $q \in [m]$ and $i \in [n]$,
define $\phi(x^{\alpha, q}_{i, t}) = x^{\alpha}_{i, t}$. 
We add $\conport(X^{\alpha}_T,
U^{\alpha}_T)$ and the corresponding vertices 
and edges as specified in 
\cref{subsubsec:strong-met-dim-add-edges}.
However, the sets $\bitrep(X^{\alpha}_T)$ and 
$\bitrep(U^{\alpha}_T)$, and the vertices 
$\glb(X^{\alpha}_T)$ and
$\glb(U^{\alpha}_T)$ are already defined.
Hence, we reuse the sets
and introduce some new vertices.
\begin{itemize}
\item We add the vertex 
$\glb^{\star}(X^{\alpha}_T)$ and make it
adjacent to every vertex in $X^{\alpha}_T$.
We also add $\pndt^{\star}(X^{\alpha}_T)$
which is adjacent to only 
$\glb^{\star}(X^{\alpha}_T)$.
Similarly, we add
$\glb^{\star}(U^{\alpha}_T)$,  
$\pndt^{\star}(U^{\alpha}_T)$, and
the corresponding edges.
\item We add the vertex 
$\glb^{\star}(\bitrep(X^{\alpha}_T))$ and make it
adjacent to every vertex in $\bitrep(X^{\alpha}_T)$.
We also add 
$\pndt^{\star}(\bitrep(X^{\alpha}_T))$
which is adjacent to only 
$\glb^{\star}(\bitrep(X^{\alpha}_T))$.
Similarly, we add
$\glb^{\star}(\bitrep(U^{\alpha}_T))$,  
$\pndt^{\star}(\bitrep(U^{\alpha}_T))$, and
the corresponding edges.

\item We add 
$\conport(X^{\alpha}_T, U^{\alpha}_T)$
and $\bits(\conport(X^{\alpha}_T, U^{\alpha}_T))$
as mentioned in \cref{subsubsec:strong-met-dim-add-edges}.
\item Finally, 
we make every vertex in 
$\conport(X^{\alpha}_T, U^{\alpha}_T)$
adjacent to
$\glb^{\star}(X^{\alpha}_T)$,
$\glb^{\star}(\bitrep(X^{\alpha}_T))$,
$\glb^{\star}(U^{\alpha}_T)$, and
$\glb^{\star}(\bitrep(U^{\alpha}_T))$.
\end{itemize}
See Figure~\ref{fig:strong-met-dim-notfull}.
We repeat the process
for the remaining pairs 
$(X^{\delta}_T, U^{\delta}_T)$ and
$(X^{\delta}_F, U^{\delta}_F)$
for $\delta \in \{\alpha, \beta, \gamma\}$.

\item In the final step,
consider the following twelve vertices:
$\glb^{\star}(X^{\delta}_T)$, 
$\glb^{\star}(X^{\delta}_F)$,
$\glb^{\star}(U^{\delta}_T)$, and
$\glb^{\star}(U^{\delta}_F)$ 
for $\delta \in \{\alpha, \beta, \gamma\}$.
These vertices are highlighted using
filled orange circles around them in \cref{fig:strong-met-dim-overview}.
We add the edges to convert
these vertices into a clique (omitted in Figure~\ref{fig:strong-met-dim-overview}).
We denote it by $S_K$, for
\emph{short-cut clique}, as it provides all the necessary short-cuts.
We remark that such additional edges
are not necessary for the
other $\glb(\cdot)$ vertices.
\end{itemize}
This completes the construction of $G$.

Suppose $Z$ is the collection of all
the pendent vertices in $G$.
Formally,
\begin{align*}
Z =& \left[\bigcup_{\delta \in \{\alpha, \beta, \gamma\}} 
\bits(X^{\delta}_T) \cup 
\bits(X^{\delta}_F) \cup 
\bits(U^{\delta}_T) \cup 
\bits(U^{\delta}_F)
\right] \bigcup\\
&
\left[ 
\bigcup_{\delta \in \{\alpha, \beta, \gamma\}} 
\bits(\conport(X^{\delta}_T, X^{\delta}_F))
\right] 
\bigcup
\left[ 
\bigcup_{\delta \neq \epsilon \in \{\alpha, \beta, \gamma\}}
\bits(\conport(U^{\delta}, U^{\epsilon}))
\right] \bigcup
\\
&
\left[ 
\bigcup_{\delta \in \{\alpha, \beta, \gamma\}} 
\bits(\conport(X^{\delta}_T, U^{\delta}_T)) \cup
\bits(\conport(X^{\delta}_F, U^{\delta}_F))
\right] \bigcup\\
& \left\{
\pndt(X^{\delta}_T),
\pndt(X^{\delta}_F),
\pndt^{\star}(X^{\delta}_T),
\pndt^{\star}(X^{\delta}_F),
\mid \delta \in \{\alpha, \beta, \gamma\}
\right\} \cup\\
& \left\{
\pndt(\bitrep(X^{\delta}_T)),
\pndt(\bitrep(X^{\delta}_F)),
\pndt^{\star}(\bitrep(X^{\delta}_T)),
\pndt^{\star}(\bitrep(X^{\delta}_F))
\mid \delta \in \{\alpha, \beta, \gamma\}
\right\}\\
&
\cup \left\{
\pndt^{\star}(U^{\delta}_T),
\pndt^{\star}(U^{\delta}_F),
\pndt^{\star}(\bitrep(U^{\delta}_T)),
\pndt^{\star}(\bitrep(U^{\delta}_F)),
\mid \delta \in \{\alpha, \beta, \gamma\}
\right\} \bigcup\\
&
\cup \left\{
\pndt(U^{\delta}),
\pndt^{\circ}(U^{\delta}),
\pndt(\bitrep(U^{\delta})),
\pndt^{\circ}(\bitrep(U^{\delta}))
\mid \delta \in \{\alpha, \beta, \gamma\}
\right\} \bigcup
\end{align*}
The reduction returns
$(G, k)$ as an instance of
\textsc{Strong Metric Dimension}, where $k = 3n + 2m + (|Z| - 1)$.

\begin{figure}[h!]
	\centering
	\includegraphics[scale=1.7]{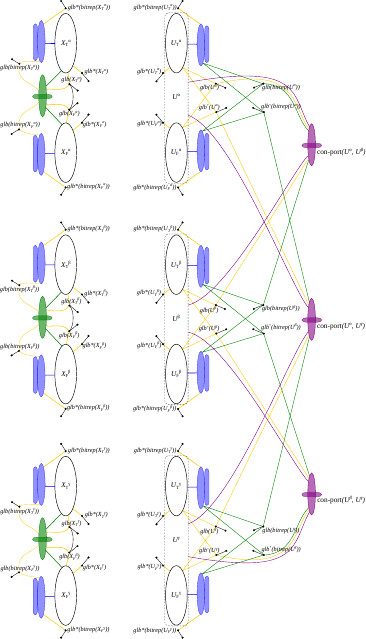}
	\caption{Overview of the vertices and edges added in the first step of the reduction along 
	with some new vertices like $\glb^{\star}(\cdot)$, which we define soon.
	Please refer to Figure~\ref{fig:strong-met-dim-clause-edges} for a more streamlined
	illustration of connections on clause-side vertices.
	We highlight that the construction so far is replicating the gadget mentioned in Subsection~\ref{subsubsec:strong-met-dim-ind-set} and \ref{subsubsec:strong-met-dim-add-edges}.
	Hence, it also satisfies the premises of Lemma~\ref{lemma:strong-met-dim-add-ind-set} 
	and \ref{lemma:strong-met-dim-add-edges}.
	This implies these vertices and edges across them are identical to corresponding 
	vertices and edges in $H$.
	\label{fig:strong-met-dim-notfull}}
\end{figure}

\begin{figure}[h!]
	\centering
	\includegraphics[scale=1.8]{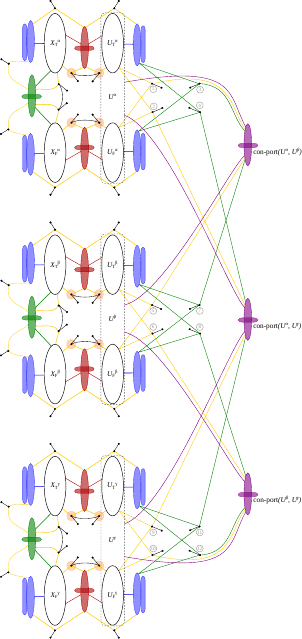}
	\caption{Overview of the reduction.
	The orange circled vertices denote the short-cut clique $S_K$.
		The edges $(1, 5)$, $(3, 9)$, and $(6, 10)$ are not shown to preserve clarity. 
		\label{fig:strong-met-dim-overview}}
\end{figure}

\begin{figure}[h!]
	\centering
	\includegraphics[scale=2]{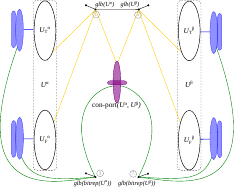}
	\caption{Highlighting the connections on the clause side.
		\label{fig:strong-met-dim-clause-edges}}
\end{figure}

\subsection{Correctness of the Reduction}
Suppose, given an instance $\psi$ of \textsc{Exact-$3$-Partitioned-$3$-SAT},
that the reduction above returns $(G,k)$ as an instance of \textsc{Strong Metric Dimension}.

\begin{lemma}
$\psi$ is a satisfiable \textsc{Exact-$3$-Partitioned-$3$-SAT} formula if and only if $G$ admits a strong resolving set of size $k$.
\end{lemma}
\begin{proof}
We note that the construction is very symmetric
with respect to the partition $X^{\alpha}$, $X^{\beta}$, $X^{\gamma}$.
Recall that, from the graph $G$, we construct $G_{SR}$
as mentioned in \cref{def:strong-resolving-graph}.
Hence, $V(G) = V(G_{SR})$.
Moreover, 
by \cref{prop:strong-met-dim-equiv-vc}, $\smd(G) = \vc(G_{SR})$.
Hence, it is sufficient to prove that
$\psi$ is a satisfiable \textsc{Exact-$3$-Partitioned-$3$-SAT} formula
if and only if $G_{SR}$ admits a vertex cover of size $k$.
We start by identifying
all the edges in $G_{SR}$.

Define the subsets $X$ an $U$ of $V(G)$ as
$X := \bigcup_{\delta \in \{\alpha, \beta, \gamma\}} (X^{\delta}_T \cup X^{\delta}_F)$
{and} 
$U := \bigcup_{\delta \in \{\alpha, \beta, \gamma\}} (U^{\delta}_T \cup U^{\delta}_F)$.
Observe that 
$\langle Z, N(Z), X, U \rangle$ is a 
partition of $V(G)$.
By \cref{obs:pendant-in-G-to-clique-in-GSR},
$N(Z)$ is a collection of pendent 
vertices in $G_{SR}$,
and hence, their presence in 
$G_{SR}$ is irrelevant while computing
its vertex cover.
Hence, we need to focus on $Z$, $X$, $U$,
and the edges across and within these sets in $G_{SR}$.
By \cref{obs:pendant-in-G-to-clique-in-GSR},
$Z$ is a clique in $G_{SR}$.
We first prove that there is 
a vertex $z \in Z$ such that 
$z$ is not adjacent to any 
vertex in $X \cup U$.
Hence, it is safe to assume that
any vertex cover of $G_{SR}$ contains $Z \setminus \{z\}$.
This also implies that we do not 
need to care about edges across
$Z$ and $X \cup U$, as these
edges will be covered by the 
vertices that are forced into
any solution.
In the end, we have that
the edges whose endpoints 
are in $X \cup U$ encode 
the instance $\psi$.

Consider the vertex 
$z = \pndt^{\star}(X^{\alpha}_T)$.
We argue that, in $G_{SR}$, 
$z$ is not adjacent to any vertex 
in $X \cup U$.
We prove that, for every vertex $u$ in $X \cup U$,
there is a vertex $v \in V(G)\setminus \{u\}$ such that
a shortest path between $z$ and $v$ contains $u$.
By \cref{def:mutually-max-distant} and the construction of
$G_{SR}$, this implies that $z$ is not adjacent to $u$ in $G_{SR}$.
First, consider the case where $u$ is in $X^{\alpha}_T$.
Every vertex in $X^{\alpha}_T$ is on some shortest path from 
$z$ to $\glb(X^{\alpha}_T)$.
Hence, $z$ is not adjacent to any vertex
in $X^{\alpha}_{T}$ in $G_{SR}$.
Now, consider any set $X' \in \{X^{\delta}_T, X^{\delta}_F \mid \delta \{\alpha, \beta, \gamma\}\}$
such that $X' \neq X^{\alpha}_T$.
Every vertex in $X'$ is on some shortest path from 
$z$ to 
$\glb(X')$.
For example, consider the shortest
path from $z$ to 
$\glb(X')$ that contains the vertices
$\glb^{\star}(X^{\alpha}_T)$,
$\glb^{\star}(X')$, 
(both of these vertices are part of $S_K$),
and a vertex in $X'$.
Hence, there is no edge incident 
to $z$ whose other endpoint is in 
$X'$.
This implies that $z$ is not adjacent to 
any vertex in $X$ in $G_{SR}$.
The arguments for any set 
$U' \in \{U^{\delta} \mid \delta \in \{\alpha, \beta, \gamma\}\}$
are identical.
For any $\delta \in \{\alpha, \beta, \gamma\}$, every vertex in $U^{\delta}$ is contained in a shortest path 
from $z$ to $\glb(U^{\delta})$.
The other vertices in the path are $\glb^{\star}(X^{\delta}_T)$ and
$\glb^{\star}(U^{\delta})$ (as both are in the clique $S_K$).
Hence, $z$ is not adjacent to any vertex in $U$ in $G_{SR}$.

Now, we consider the
edges across and within $X \cup U$ in $G_{SR}$.
We consider the following partition
of $X \cup U$:
$X^{\alpha}_T$, $X^{\alpha}_F$, $X^{\delta}_T$, $X^{\delta}_F$, and
$U^{\alpha}_T$, $U^{\alpha}_F$, $U^{\delta}_T$, $U^{\delta}_F$,
where 
$\delta \in \{\beta, \gamma\}$.
In \cref{table:SMD-edge-across-X-U}, we describe the edges across and within $X\cup U$ in $G_{SR}$.
In particular, there are two types of non-empty entries. For the first such type, we make reference
to a gadget that enforces certain edges to exist and others to not exist in $G_{SR}$,
and the existence and non-existence of these edges is proven by either
\cref{lemma:strong-met-dim-add-ind-set} or \cref{lemma:strong-met-dim-add-edges}.
In the second type of entry, we mention a vertex such that any vertex in the set in the
same row (on the far left) lies on a shortest path 
between the vertex mentioned in the entry and any vertex in the set in the same column (top).
In the latter case, this implies that there is no edge in $G_{SR}$ between these two sets.
For more details, see the caption of \cref{table:SMD-edge-across-X-U}.
For example, consider the entry in the first row and last column.
This indicates that there is no edge across $X^{\alpha}_T$ and $U^{\alpha}_T$.
Consider the shortest path from $u$ to $\glb(X^{\alpha}_T)$
that contains the vertices $\glb^{\star}(U^{\alpha}_F)$,
$\glb^{\star}(X^{\alpha}_T)$ (as both these vertices are in the short-cut clique),
and any arbitrary vertex $x$ in $X^{\alpha}_T$.
This implies that there no edge with endpoints $u, s$ in $G_{SR}$. 
As these two are arbitrary vertices in the respective sets, our claim holds. 

This concludes the proof since $G_{SR}$ is very close to the graph $H$ mentioned in \cref{lemma:Ex-3-Par-3-SAT-to-VC}.
Indeed, $G_{SR}$ has the properties mentioned in the paragraph following \cref{lemma:Ex-3-Par-3-SAT-to-VC}, that is, all
the edges in $E(H)$ are present in $G_{SR}$, all the edges in $E(G_{SR})\setminus E(H)$ are incident to $Z$, and $N[z]=Z$.
The proof then follows from the remaining arguments in that paragraph.
\begin{table}
\begin{center}
\begin{tabular}{|c|c|c|c|c|}
\hline
 & $X^{\alpha}_T$ & $X^{\alpha}_F$ &
   $U^{\alpha}_T$ & $U^{\alpha}_F$
 \\
   \hline
   $X^{\alpha}_T$ & 
   \textcolor{blue}{Ind-Set} & 
   \textcolor{green}{Set-Rep} &
   \textcolor{red}{Set-Rep} & 
   $\glb(X^{\alpha}_T)$ \\
   \hline
   $X^{\alpha}_F$ & & \textcolor{blue}{Ind-Set} &
   $\glb(X^{\alpha}_F)$
   & \textcolor{red}{Set-Rep} \\
   \hline
   $U^{\alpha}_T$ & & & \textcolor{blue}{Ind-Set} & $\bitrep(U^{\alpha}_T)$\\
   \hline
   $U^{\alpha}_F$ & & & & \textcolor{blue}{Ind-Set} \\
   \hline
   \hline
 & $X^{\delta}_T$ & $X^{\delta}_F$ &
   $U^{\delta}_T$ & $U^{\delta}_F$
 \\
   \hline
   $X^{\alpha}_T$ & 
   $\glb(X^{\alpha}_T)$ & 
$\glb(X^{\alpha}_T)$ &
$\glb(X^{\alpha}_T)$ & 
   $\glb(X^{\alpha}_T)$ \\
   \hline
   $X^{\alpha}_F$ & & 
   $\glb(X^{\alpha}_F)$ &
   $\glb(X^{\alpha}_F)$ & 
   $\glb(X^{\alpha}_F)$ \\
   \hline
   $U^{\alpha}_T$ & & &
   \textcolor{purple}{Set-Rep} & 
   \textcolor{purple}{Set-Rep}\\
   \hline
   $U^{\alpha}_F$ & & & &
   \textcolor{purple}{Set-Rep}\\
   \hline
\end{tabular}
\end{center}
\caption{Overview of the adjacencies across the partition
of $X \cup U$.
Here, $\delta \in \{\beta, \gamma\}$.
Diagonal entries marked with \textcolor{blue}{Ind-Set} denote
that the set remains as an independent set because of the gadget in
\cref{subsubsec:strong-met-dim-ind-set} and \cref{lemma:strong-met-dim-add-ind-set}.
\textcolor{green}{Set-Rep},  \textcolor{red}{Set-Rep}, and \textcolor{purple}{Set-Rep} 
denote that the corresponding set contains the edges with respect 
to the connection portal added based on the set representation gadget
in \cref{subsubsec:strong-met-dim-add-edges} and \cref{lemma:strong-met-dim-add-edges}.
The other non-empty entries denote that there 
is no edge across these sets. In particular, for such a non-empty entry, 
every vertex in the set in the same row (on the far left) lies on a shortest path 
between the vertex mentioned in the entry and any vertex in the set in the same column (top).
For example, every vertex in $X^{\alpha}_T$ lies on a shortest path between
$\glb(X^{\alpha}_T)$ and any vertex in $U^{\alpha}_T$. 
\label{table:SMD-edge-across-X-U}}
\end{table}
\end{proof}

\begin{proof}[Proof of Theorem~\ref{thm:smd-lower-bound-vc}.] 
Consider the set $Z$ defined above. 
As every set mentioned in its definition is of size $\calO(\log(n))$,
we have $|Z| = \calO(\log(n))$.
As every vertex in $Z$ is a pendent vertex, $|N(Z)| = \calO(\log(n))$.
By construction, it is easy to verify that $N(Z)$ is a vertex cover of $G$.
Hence, $\vc(G) = \calO(\log(n))$.
This implies that if there is an algorithm running in time $2^{2^{o(\vc)}} \cdot n^{\calO(1)}$,
then \textsc{Exact-3-Partitioned-3-SAT} has an algorithm running in time
$2^{o(n)}$, as the reduction takes polynomial time in the size of input.
This, however, contradicts Proposition~\ref{prop:3-SAT-to-3-Partition-3-SAT}.
Hence, the first part of the theorem is true.
The second part of the theorem follows from similar arguments coupled
with the facts that the problem admits a kernel with $2^{\calO(\vc)}$ vertices
and a brute-force algorithm running in time $2^{\calO(n)}$.
\end{proof}

\section{Algorithms}

\subsection{Dynamic Programming Algorithm for \mdfull}\label{sec:algo-diam-tw-MD}
\label{subsec:algo-tw-diam-MD}


\newcommand{\val}[1]{val_#1}
\renewcommand{\int}[1]{D_{int}{{(#1)}}}
\newcommand{\ext}[1]{D_{ext}{{(#1)}}}
\newcommand{\paire}[1]{D_{pair}{{(#1)}}}
\newcommand{\I}[1]{ I_{#1}=(X_{i_{#1}}, S_{I_{#1}}, \int {I_{#1}},\ext {I_{#1}},\paire {I_{#1})}}
\newcommand{\T}[1]{T(#1)}
\renewcommand{\vec}{\mathbf}
\newcommand{\distvec}[2]{\vec{\mathbf{{d_{#1}}(#2)}}}
\newcommand{\dimx}[1]{\dim(#1)}


The aim of this subsection is to prove the following theorem.
\begin{theorem}
\label{thm:algo-diam-tw-MD}
\mdfull admits an algorithm running in time
$2^{\diam^{\mathcal{O}(\tw)}} \cdot n^{\OO(1)}$.
\end{theorem}
To this end, we give a dynamic programming algorithm on a tree decomposition for this problem. The algorithm is inspired by the one from~\cite{BDP23} for chordal graphs, though there are some non-trivial differences. We will assume that a tree decomposition of the input graph $G$ of width $w$ is given to us. Note that one can compute a tree decomposition of width $w\leq 2\tw(G)+1$ in time $2^{\OO(\tw(G))}n$~\cite{K21}, and it can be transformed into a \emph{nice tree decomposition} of the same width with $\OO(wn)$ bags in time $\OO(w^2n)$~\cite{niceTW}.

\noindent\textbf{Overview.} As mentioned previously, our dynamic programming is similar to that of~\cite{BDP23}. However, in~\cite{BDP23}, as the diameter of the graph is unbounded, it was crucial to restrict the computations for each step of the dynamic programming to vertices ``not too far'' from the current bag. This was possible due to the metric properties of chordal graphs. In our case, as we consider the diameter of the graph as a parameter, we do not need such restrictions, which makes the proof a little bit simpler.

We now give an intuitive description of the dynamic programming scheme. At each step of the algorithm, we consider a bounded number of \emph{solution types}, depending on the properties of the solution vertices with respect to the current bag. At a given dynamic programming step, we will assume that the current solution resolves all vertex pairs in $G_i$. Such a vertex pair may be resolved by a vertex from $G-G_i$, or by a vertex in $G_i$ itself. 

Any bag $X_i$ of the tree decomposition whose node $i$ lies on a path between two join nodes in $T$, forms a separator of $G$: there are no edges between the vertices of $G_i-X_i$ and $G-G_ i$. For a vertex $v$ not in $X_i$, we consider its distance-vector to the vertices of $X_i$; the distance-vectors induce an equivalence relation on the vertices of $G-X_i$, whose classes we call \emph{$X_i$-classes}. Consider the two subgraphs $G_i$ and $G-G_i$. Any two solution vertices $x,y$ from $G-G_i$ that are in the same $X_i$-class, 
resolve the exact same pairs of vertices from $G_i$. Thus, for this purpose, it is irrelevant whether $x$ or $y$ will be in a resolving set, and it is sufficient to know that a vertex of their $X_i$-class will eventually be chosen. In this way, one can check whether a vertex pair from $G_i$ is resolved by a solution vertex of $G-G_i$. 

The same idea is used to ``remember'' the previously computed solution: it is sufficient to remember the $X_i$-classes of the vertices in the previously computed resolving set, rather than the vertices themselves.

It is slightly more delicate to make sure that vertex pairs in $G_i$ are resolved in the case where such a pair is resolved by a vertex in $G_i$. Indeed, this must be ensured, in particular when processing a join node $i$, for vertex pairs belonging to bags in the two sub-trees corresponding to the children $i_1,i_2$ of $i$. Such pairs may be resolved by four types of solution vertices: from $G-G_i$, $X_i$, $G_{i_1}-X_i$, or $G_{i_2}-X_i$. To ensure this, the dynamic programming scheme makes sure that, at each step, for any possible pair $C_1,C_2$ of $X_i$-classes, all vertex pairs $\langle u,v \rangle$ consisting of a vertex $u$ of $G_i$ with class $C_1$ and a vertex $v$ of $G-G_i$ with class $C_2$ are resolved. The crucial step here is that when a new vertex $v$ is introduced (i.e., added to a bag $X_i$ to form $X_{i'}$), depending on its $X_i$-class, it must be made sure that it is resolved from all other vertices depending on their $X_i$-classes, as described above. To ensure that $v$ is distinguished from all other vertices of $G_i$, we keep track of vertex pairs of $G_i \times (G-G_i)$ that are already resolved by the partial solution, and enforce that, when processing bag $X_{i'}$, for every vertex $x$ of $G_i$, the pair $\langle x,v \rangle$ is already resolved. As $v$ belongs to the new bag $X_{i'}$, we know its distances to all resolving vertices (indeed, $X_{i'}$-classes of solution vertices can be computed from their $X_i$-classes), and thus, the information can be updated accurately.

For a bag $X_i$ and a vertex $v$ not in $X_i$, the number of possible distance vectors to the vertices of $X_i$ is at most $\diam(G)^{|X_i|}$. Thus, a solution for bag $X_i$ will consist of: (i) the subset of vertices of $X_i$ selected in the solution; (ii) a subset of the $\diam(G)^{|X_i|}$ possible vectors to denote the $X_i$-classes from which the currently computed solution (for $G_i$) contains at least one vertex in the resolving set; (iii) a subset of the $\diam(G)^{|X_i|}$ possible vectors denoting the $X_i$-classes from which the future solution needs at least one vertex of $G-G_i$ in the resolving set; (iv) a subset of the $\diam(G)^{|X_i|}\times \diam(G)^{|X_i|}$ possible pairs of vectors representing the $X_i$-classes of the pairs of vertices in $G_i\times (G-G_i)$ that are already resolved by the partial solution.\\

\noindent\textbf{Formal description.} We mostly follow the notations used in~\cite{BDP23}. Before presenting the dynamic program, we first introduce some useful definitions and lemmas.

\begin{definition}
Given a vector $\vec{r}$, we refer to the $i$-th coordinate of $\vec r$ as $\vec{r}_i$.
\begin{itemize}
    \item Let $\vec{r}$ be a vector of size $k$ and let $m$ be an integer. The vector $\vec t = \vec{r|}m$ is the vector of size $k+1$ such that $\vec{t}_{k+1}=m$ and, for all $1 \leq i \leq k$, $\vec{t}_i=\vec{r}_i$ . 
    \item Let $\vec{r}$ be a vector of size $k$. The vector $\vec{r^-}$ is the vector of size $k-1$ such that, for all $1 \leq i \leq k-1$, $\vec{r}^-_i=\vec{r}_i$. 
\end{itemize}
\end{definition}

\begin{definition}
Let $G$ be a graph and let $X = \{v_1, \ldots,v_k\}$ be a subset of vertices of $G$. For a vertex $x$ of $G$, the \emph{distance vector} $\vec{d_X}(x)$ of $x$ to $X$ is the vector of size $k$ such that, for all $1 \leq j \leq k$, $\vec{d_X}(x)_j=d(x,v_j)$. For a set $S \subseteq V(G)$, we let $\vec{d_X}(S) = \{\vec{d_X}(s) \mid s \in S\}$.
\end{definition}

\begin{definition}
Let $\vec{r_1}, \vec{r_2}$, and $\vec{r_3}$ be three vectors of size $k$. We say that $\vec{r_3}$ \emph{resolves} the pair $\langle \vec{r_1},\vec{r_2} \rangle$ if \[\min_{1 \leq i \leq k} \vec{(r_1+r_3)}_i \neq \min_{1 \leq i \leq k} \vec{(r_2+r_3)}_i.\]
\end{definition}

\begin{lemma}
Let $X$ be a separator of a graph $G$, and let $G_1$ be a connected component of $G - X$. Let $\langle x,y \rangle$ be a pair of vertices of $G - G_1$, and let $\vec r$ be a vector of size $|X|$. If $\vec r$ resolves the pair $\langle \vec{d_X}(x),\vec{d_X}(y) \rangle$, then any vertex $s \in V(G_1)$ such that $\vec{d_X}(s) = \vec r$ resolves the pair $\langle x,y \rangle$.
\end{lemma}

\begin{proof}
To see this, it suffices to note that since $X$ separates $s$ from $x$ ($y$, resp.), $d(s,x) = \min_{1 \leq j \leq |X|} (\vec{d_X}(s) + \vec{d_X}(x))_j$ ($d(s,y) = \min_{1 \leq j \leq |X|} (\vec{d_X}(s) + \vec{d_X}(y))_j$, resp.); and since $\vec r$ resolves the pair $\langle \vec{d_X}(x),\vec{d_X}(y) \rangle$, $d(s,x) \neq d(s,y)$.
\end{proof}

\begin{definition}
Let $X$ be a separator of a graph $G$, and let $G_1,G_2$ be two (not necessarily distinct) connected components of $G - X$. Let $x \in V(G_1) \cup X$ and $y \in V(G_2) \cup X$. If a vector $\vec r$ resolves the pair $\langle \vec{d_X}(x), \vec{d_X}(y) \rangle$, then we say that $\vec r$ \emph{resolves} the pair $\langle x,y \rangle$. More generally, given a set $M$ of vectors, we say that the pair $\langle x,y \rangle$ \emph{is resolved} by $M$ if there exists $\vec r \in M$ that resolves the pair $\langle x,y \rangle$.
\end{definition}

We now define the generalized problem solved at each step of the dynamic programming algorithm, called \textsc{Extended Metric Dimension} ({\sc EMD} for short), whose instances are defined as follows.

\begin{definition}
Let $G$ be a graph and let $(T,\{X_i:i \in V(T)\})$ be a tree decomposition of $G$. For a node $i$ of $T$, an \emph{instance of {\sc EMD} for $i$} is a $5$-tuple $ I = (X_i, S_I, \int I, \ext I, \paire I)$ composed of the bag $X_i$ of $i$, a subset $S_I$ of $X_i$, and three sets of vectors satisfying the following.
\begin{itemize}
    \item $\int {I}\subseteq [\diam(G)]^{|X_i|}$ and $\ext{I} \subseteq [\diam(G)]^{|X_i|}$.
    \item $\ext{I} \neq \emptyset$ or $S_I \neq \emptyset$.
    \item $\paire{I} \subseteq [\diam(G)]^{|X_i|}\times [\diam(G)]^{|X_i|}$.
    \item For each pair of vectors $\langle \vec{r_1},\vec{r_2} \rangle \in \paire I$, there exist two vertices $x \in V(G_i)$ and $y \notin V(G_i)$ such that $\vec{d_{X_i}}(x)= \vec{r_1}$ and $\vec{d_{X_i}}(y)= \vec{r_2}$.
    \item For each vector $\vec r$ of $\ext I$, there exists $x \notin V(G_i)$ such that $\vec{d_{X_i}}(x)= \vec r$.
\end{itemize} 
\end{definition}

\begin{definition}
\label{def_instance}
A set $S \subseteq G_i$ is a solution for an instance $I$ of {\sc EMD} if the following hold.
\begin{itemize}
    \item \textbf{(S1)} Every pair of vertices of $G_i$ is  either resolved by a vertex in $S$ or resolved by a vector of $\ext I$. 
    \item \textbf{(S2)} For each vector $\vec{r} \in \int I$, there exists a vertex $s \in S$ such that $\vec{d_{X_i}}(s)= \vec{r}$. 
    \item \textbf{(S3)} For each pair of vectors $\langle \vec{r_1},\vec{r_2} \rangle \in \paire I$, any vertex $x \in V(G_i)$ such that $\vec{d_{X_i}}(x)= \vec{r_1}$, and any vertex $y \notin V(G_i)$ such that $\vec{d_{X_i}}(y)= \vec{r_2}$, the pair $\langle x,y \rangle$ is resolved by $S$. 
    \item \textbf{(S4)} $S \cap X_i =S_I$. 
\end{itemize}
\end{definition}

In the remainder of this section, for brevity, we will refer to an instance of the {\sc EMD} problem only by an instance.

\begin{definition}
\label{def_dim}
Let $I$ be an instance. We denote by $\dim (I)$ the minimum size of a set $S \subseteq V(G_i)$ which is a solution for $I$. If no such set exists, then we set $\dim (I)= + \infty$. We refer to this value as the \emph{extended metric dimension} of $I$.
\end{definition}

In the following, we fix a graph $G$ and a nice tree decomposition $(T,\{X_i:i \in V(T)\})$ of $G$. Given a node $i$ of $T$ and an instance $I$ for $i$, we show how to compute $\dim(I)$. The proof is divided according to the type of the node $i$.\\

\noindent
{\bf Leaf node.} Computing $\dim(I)$ when $I$ is an instance for a leaf node can be done with the following lemma.

\begin{lemma}
\label{lem:leaf_node}
Let $I$ be an instance for a leaf node $i$ and let $v$ be the only vertex in $X_i$. Then,
\begin{center}
$\dim (I) =
\begin{cases}
0 & \text{if } S_I = \emptyset, \int I = \emptyset, \text{ and } \paire I  = \emptyset \\
1 & \text{if } S_I = \{v\} \text{ and } \int I \subseteq \{(0)\} \\
+\infty & \text{otherwise}
\end{cases}
$
\end{center}
\end{lemma} 

\begin{proof}
Suppose first that $S_I = \emptyset$. Then, the empty set is the only possible solution for $I$; and the empty set is a solution for $I$ only if $\int I = \emptyset$ and $\paire I = \emptyset$. Suppose next that $S_I = \{v\}$. Then, the set $S = \{v\}$ is the only possible solution for $I$; and this set is a solution for $I$ only if $\int I = \emptyset$ or $\int I$ contains only the vector $\vec{d_{X_i}}(v)= (0)$.
\end{proof}

In the remainder of this section, we handle the three other types of nodes. For each type of node, we proceed as follows: we first define a notion of compatibility on the instances for the child/children of $i$ and show how to compute the extended metric dimension of $I$ from the extended metric dimension of compatible instances for the child/children of $i$.\\

\noindent
{\bf Join node.} Let $I$ be an instance for a join node $i$, and let $i_1$ and $i_2$ be the two children of $i$. Given a pair of instances $\langle I_1,I_2 \rangle$ for $\langle i_1,i_2 \rangle$, we say that a pair $\langle \vec r, \vec t \rangle \in [\diam(G)]^{|X_i|}\times [\diam(G)]^{|X_i|}$ is \emph{2-compatible} if there exist $x \in V(G_{i_1}), y \in V(G - G_i)$, and $\vec{u} \in \int{I_2}$ such that $\vec{d_{X_i}}(x) = \vec r$, $\vec{d_{X_i}}(y) = \vec t$, and $\vec u$ resolves the pair $\langle \vec{r},\vec{t} \rangle$. Symmetrically, we call a pair $\langle \vec r, \vec t \rangle \in [\diam(G)]^{|X_i|}\times [\diam(G)]^{|X_i|}$ \emph{1-compatible} if there exist $x \in V(G_{i_2}), y \in V(G - G_i)$, and $\vec{u} \in \int{I_1}$ such that $\vec{d_{X_i}}(x) = \vec r$, $\vec{d_{X_i}}(y) = \vec t$, and $\vec u$ resolves the pair $\langle \vec{r},\vec{t} \rangle$.

\begin{definition}
\label{def:compatible_pair}
A pair of instances $\langle I_1, I_2 \rangle$ for $\langle i_1,i_2 \rangle$ is \emph{compatible} with $I$ if the following hold.
\begin{itemize}
    \item \textbf{(J1)} $S_{I_1}=S_{I_2}=S_I$.
    \item \textbf{(J2)} $ \ext {I_1} \subseteq \ext{I} \cup \int{I_2} $ and $ \ext {I_2} \subseteq \ext{I} \cup \int{I_1}$.
    \item \textbf{(J3)} $\int{I} \subseteq \int{I_1} \cup \int {I_2}$.

    \item \textbf{(J4)} 
    Let $C_1 = \{\langle \vec r, \vec t \rangle \in [\diam(G)]^{|X_i|}\times [\diam(G)]^{|X_i|} \mid \not\exists x \in V(G_{i_1}) \text{ s.t. } \vec{d_{X_{i_1}}}(x) = \vec r\}$ and $C_2 = \{\langle \vec r, \vec t \rangle \in [\diam(G)]^{|X_i|}\times [\diam(G)]^{|X_i|} \mid \not\exists y \in V(G_{i_2}) \text{ s.t. } \vec{d_{X_{i_2}}}(y) = \vec r\}.$
    Further, let 
    $D_1=\{ \langle \vec r, \vec t \rangle \in [\diam(G)]^{|X_i|}\times [\diam(G)]^{|X_i|} \mid \langle \vec r, \vec t \rangle \text{ is 2-compatible}\}$
    and  
$D_2=\{ \langle \vec r, \vec t \rangle \in [\diam(G)]^{|X_i|}\times [\diam(G)]^{|X_i|} \mid \langle \vec r, \vec t \rangle \text{ is 1-compatible}\}.$
    Then, $ \paire{I} \subseteq (C_1 \cup D_1 \cup \paire{I_1}) \cap (C_2 \cup D_2 \cup \paire{I_2})$.

    \item \textbf{(J5)} For all $\vec{r_1}, \vec{r_2} \in [\diam(G)]^{|X_i|}$ for which there exist $x \in V(G_{i_1})$ and $y \in V(G_{i_2})$ such that $\vec{d_{X_{i_1}}}(x) = \vec{r_1}$ and $\vec{d_{X_{i_2}}}(y) = \vec{r_2}$, one of the following holds:
    \begin{itemize}
    \item $\langle \vec{r_1},\vec{r_2} \rangle \in \paire {I_1}$,
    \item $\langle \vec{r_2},\vec{r_1} \rangle \in \paire {I_2}$, or
    \item there exists $\vec{t} \in \ext{I}$ such that $\vec{t}$ resolves the pair $\langle \vec{r_1},\vec{r_2} \rangle$.
    \end{itemize}
    
\end{itemize}
\end{definition}

Let $\mathcal{F}_J(I)$ be the set of pairs of instances compatible with~$I$. We aim to prove the following.

\begin{lemma}
\label{lem:join_node}
Let $I$ be an instance for a join node $i$. Then, 
\[ 
\dim(I) = \min_{(I_1,I_2) \in \mathcal{F}_J(I)} (\dim(I_1) + \dim(I_2) - |S_I|). 
\]
\end{lemma}

To prove \Cref{lem:join_node}, we prove the following two lemmas.

\begin{lemma}
Let $\langle I_1,I_2 \rangle$ be a pair of instances for $\langle i_1,i_2 \rangle$ compatible with $I$ such that $\dim(I_1)$ and $\dim(I_2)$ have finite values. Let $S_1$ be a minimum-size solution for $I_1$ and $S_2$ a minimum-size solution for $I_2$. Then, $S = S_1 \cup S_2$ is a solution for $I$. In particular, 
\[ 
\dim(I) \leq \min_{(I_1,I_2) \in \mathcal{F}_J(I)} (\dim(I_1) + \dim(I_2) - |S_I|). 
\]
\end{lemma}

\begin{proof}
Let us show that every condition of \Cref{def_instance} is satisfied.\\

\noindent
{\bf (S1)} Let $\langle x,y \rangle$ be a pair of vertices of $G_i$. Assume first that $x,y \in V(G_{i_1})$. Then, since $S_1$ is a solution for $I_1$, either $S_1$ resolves the pair $\langle x,y \rangle$, in which case we are done; or the pair $\langle x,y \rangle$ is resolved by a vector $\vec{t} \in \ext{I_1}$. In the latter case, by compatibility, either $\vec{t} \in \ext{I}$, in which case the pair $\langle x,y \rangle$ is still resolved by $\vec{t} \in \ext{I}$, or $\vec{t} \in \int{I_2}$; but then, there exists $s \in S_2$ such that $\vec{d_{X_{i_2}}}(s) = \vec{t}$, and so, $s \in S$ resolves the pair $\langle x,y \rangle$. The case where $x,y \in V(G_{i_2})$ is handled symmetrically. 

Assume therefore that $x \in V(G_{i_1})$ and $y \in V(G_{i_2})$. Then, by compatibility, one of the following holds:
\begin{enumerate}
\item $\langle \vec{d_{X_{i_1}}}(x),\vec{d_{X_{i_2}}}(y) \rangle \in \paire{I_1}$,
\item $\langle \vec{d_{X_{i_2}}}(y),\vec{d_{X_{i_1}}}(x) \rangle \in \paire{I_2}$, or
\item there exists $\vec t \in \ext I$ such that $\vec t$ resolves the pair $\langle \vec{d_{X_{i_1}}}(x),\vec{d_{X_{i_2}}}(y) \rangle$.
\end{enumerate}  
Suppose that item 3. does not hold (we are done otherwise). If item 1. holds, then, since $S_1$ is a solution for $I_1$, the pair $\langle x,y \rangle$ is resolved by $S_1$; and we conclude symmetrically if item 2. holds.\\

\noindent
{\bf (S2)} Consider $\vec r \in \int I$. Then, by compatibility, $\vec r \in \int{I_1}$ or $\vec r \in \int{I_2}$. If the former holds, then since $S_1$ is a solution for $I_1$, there exists $s \in S_1$ such that $\vec{d_{X_{i_1}}}(s) = \vec r$; but then, $s \in S$ with $\vec{d_{X_i}}(s) = \vec r$. We conclude similarly if $\vec r \in \int{I_2}$.\\

\noindent
{\bf (S3)} Consider $\langle \vec r, \vec t \rangle \in \paire I$ and let $x \in V(G_i)$, $y \in V(G-G_i)$ be such that $\vec{d_{X_i}}(x) = \vec r$ and $\vec{d_{X_i}}(y) = \vec t$. Assume, without loss of generality, that $x \in V(G_{i_1})$. By compatibility, $\langle \vec r, \vec t \rangle \in (C_1 \cup D_1 \cup \paire{I_1}) \cap (C_2 \cup D_2 \cup \paire{I_2})$; in particular, $\langle \vec r, \vec t \rangle \in C_ 1 \cup D_1 \cup \paire{I_1}$. Note that $\langle \vec r, \vec t \rangle \notin C_1$ since $x \in V(G_{i_1})$ and $\vec{d_{X_i}}(x) = \vec r$. Now, suppose that $\langle \vec r, \vec t \rangle \in D_1$. Then, there exists $\vec u \in \int{I_2}$ such that $\vec u$ resolves the pair $\langle \vec r, \vec t \rangle$; and since $S_2$ is a solution for $I_2$, there exists $s \in S_2$ such that $\vec{d_{X_{i_2}}}(s) = \vec u$, and so, $s \in S$ resolves the pair $\langle x,y \rangle$. Finally, if  $\langle \vec r, \vec t \rangle \in \paire{I_1}$, then $S_1$ resolves $\langle x,y \rangle$ as it is a solution for $I_1$.\\

\noindent
{\bf (S4)} By compatibility, $S_{I_1} = S_{I_2} = S_I$, and thus, $S \cap X_i = S \cap X_{i_1} = S_{I_1} = S_I$.\\

It now follows from the above that $\dim(I) \leq \dim(I_1)+\dim(I_2)-|S_I|$, and since this holds true for any $\langle I_1,I_2 \rangle \in \mathcal{F}_J(I)$, the lemma follows.
\end{proof}

\begin{lemma}
Let $I$ be an instance for a join node $i$. Then, 
\[ 
\dim(I) \geq \min_{(I_1,I_2) \in \mathcal{F}_J(I)} (\dim(I_1) + \dim(I_2) - |S_I|). 
\]
\end{lemma}

\begin{proof}
If $\dim(I) = + \infty$, then the inequality readily holds. Thus, assume that $\dim(I) < + \infty$, and let $S$ be a minimum-size solution for $I$. For $j \in \{1,2\}$, let $S_j = S \cap V(G_{i_j})$. Now, let $I_1$ and $I_2$ be the two instances for $i_1$ and $i_2$, respectively, defined as follows.
\begin{itemize} 
\item $S_{I_1} = S_{I_2} = S_I$.
\item $\int{I_1} = \vec{d_{X_i}}(S_1)$ and $\int{I_2} = \vec{d_{X_i}}(S_2)$.
\item $\ext{I_1} = \ext{I} \cup \int{I_2}$ and $\ext{I_2} = \ext{I} \cup \int{I_1}$. 
\item We construct $\paire{I_1}$ as follows ($\paire{I_2}$ is constructed symmetrically). For every $\langle \vec r, \vec t \rangle \in ([\diam(G)]^{|X_{i_1}|})^2$, let $R_{\langle \vec r, \vec t \rangle} = \{\langle x,y \rangle \in V(G_{i_1}) \times V(G - G_{i_1}) \mid \vec{d_{X_{i_1}}}(x) = \vec r \text{ and } \vec{d_{X_{i_1}}}(y) = \vec t\}$. If, for every pair $\langle x,y \rangle \in R_{\langle \vec r, \vec t \rangle}$, $S_1$ resolves $\langle x,y \rangle$, then we add $\langle \vec r, \vec t \rangle$ to $\paire{I_1}$. 
\end{itemize}
Let us show that $\langle I_1,I_2 \rangle$ is compatible with $I$ and that, for $j \in \{1,2\}$, $S_j$ is a solution for $I_j$.

\begin{claim}
\label{clm:compatible}
The constructed pair of instances $\langle I_1,I_2 \rangle$ for $\langle i_1,i_2 \rangle$ is compatible with $I$.
\end{claim}

\begin{claimproof}
It is clear that conditions {\bf (J1)}, {\bf (J2)}, and {\bf (J3)} of \Cref{def:compatible_pair} hold; let us show that the remaining conditions hold as well.\\ 

\noindent
{\bf (J4)} Consider a pair $\langle \vec r, \vec t \rangle \in \paire{I}$. Let us show that $\langle \vec r, \vec t \rangle \in C_1 \cup D_1 \cup \paire{I_1}$ (showing that $\langle \vec r, \vec t \rangle \in C_2 \cup D_2 \cup \paire{I_2}$ can be done symmetrically). 

If there exists no vertex $x \in V(G_{i_1})$ such that $\vec{d_{X_{i_1}}}(x) = \vec r$, then $\langle \vec r, \vec t \rangle \in C_1$. Otherwise, let $x \in V(G_{i_1})$ be a vertex such that $\vec{d_{X_{i_1}}}(x) = \vec r$, and let $y \in V(G - G_i)$ be a vertex such that $\vec{d_{X_i}}(y) = \vec t$ (note that such a vertex $y$ exists since $\langle \vec r, \vec t \rangle \in \paire{I}$). Then, since $S$ is a solution for $I$, $\langle x,y \rangle$ is resolved by $S$. Now, if there exists a vertex $s \in S \cap V(G_{i_2})$ such that $s$ resolves $\langle x,y \rangle$, then $s$ resolves every pair $\langle u,v \rangle \in V(G_{i_1}) \times V(G-G_i)$ such that $\vec{d_{X_{i_1}}}(u) = \vec r$ and $\vec{d_{X_i}}(v) = \vec t$; but then, $\langle \vec r, \vec t \rangle \in D_1$. Thus, assume that no vertex in $S \cap V(G_{i_2})$ resolves $\langle x,y \rangle$. Then, there exists a vertex $s \in S \cap V(G_{i_1})$ that resolves the pair $\langle x,y \rangle$; and since this holds for every pair $\langle u,v \rangle \in V(G_{i_1}) \times V(G-G_i)$ (and, a fortiori, for every pair $\langle u,v \rangle \in V(G_{i_1}) \times V(G-G_{i_1})$) such that $\vec{d_{X_{i_1}}}(u) = \vec r \text{ and } \vec{d_{X_{i_1}}}(v) = \vec t$, we conclude that $\langle \vec r, \vec t \rangle \in \paire{I_1}$.\\

\noindent
{\bf (J5)} Let $\vec{r_1},\vec{r_2} \in [\diam(G)]^{|X_i|}$ be two vectors for which there exist $x \in V(G_{i_1})$ and $y \in V(G_{i_2})$ such that $\vec{d_{X_{i_1}}}(x) = \vec{r_1}$ and $\vec{d_{X_{i_2}}}(y) = \vec{r_2}$. Then, since $S$ is a solution for $I$, either the pair $\langle x,y \rangle$ is resolved by a vector in $\ext I$, in which case condition (J5) holds; or there exists a vertex in $S$ resolving $\langle x,y \rangle$. Let us show that, in the latter case, $\langle \vec{r_1},\vec{r_2} \rangle \in \paire{I_1}$ or $\langle \vec{r_2},\vec{r_1} \rangle \in \paire{I_2}$. Suppose toward a contradiction that this does not hold. Then, there exist $\langle x_1,y_1 \rangle \in V(G_{i_1}) \times V(G_{i_2})$ and $\langle x_2,y_2 \rangle \in V(G_{i_2}) \times V(G_{i_1})$ such that $\vec{d_{X_{i_1}}}(x_1) = \vec{d_{X_{i_2}}}(y_2) = \vec{r_1}$, $\vec{d_{X_{i_1}}}(y_1) = \vec{d_{X_{i_2}}}(x_2) = \vec{r_2}$, $S_1$ does not resolve $\langle x_1,y_1 \rangle$, and $S_2$ does not resolve $\langle x_2,y_2 \rangle$. Now, since $S$ is a solution for $I$, there exists $s \in S$ such that $s$ resolves the pair $\langle x_1,x_2 \rangle$; but then, either $s \in S_1$, in which case $s$ resolves the pair $\langle x_1,y_1 \rangle$, or $s \in S_2$, in which case $s$ resolves the pair $\langle x_2,y_2 \rangle$, a contradiction in both cases.
\end{claimproof}

\begin{claim}
\label{clm:compatible_sol}
For every $j \in \{1,2\}$, $S_j$ is a solution for $I_j$.
\end{claim}

\begin{claimproof}
We only prove that $S_1$ is a solution for $I_1$ as the other case is symmetric. To this end, let us show that every condition of \Cref{def_instance} is satisfied.\\

\noindent
{\bf (S1)} Consider two vertices $x,y \in V(G_{i_1})$. Since $S$ is a solution for $I$, the pair $\langle x,y \rangle$ is either resolved by a vector in $\ext I$, in which case we are done as $\ext I \subseteq \ext{I_1}$ by construction; or resolved by a vertex $s \in S$. Now, if $s \in V(G_{i_1})$, then $s$ is a vertex of $S_1$ resolving $\langle x,y \rangle$. Otherwise, $s \in V(G_{i_2})$ and by construction of $I_1$, there exists a vector $\vec r \in \ext{I_1}$ such that $\vec{d_{X_i}}(s) = \vec r$, and so, $\vec r$ resolves the pair $\langle x,y \rangle$.\\

\noindent
{\bf (S2)} Readily follows from the fact that $\int{I_1} = \vec{d_{X_i}} (S_1)$.\\

\noindent
{\bf (S3)} By construction, for every $\langle \vec r,\vec t \rangle \in \paire{I_1}$, any $x \in V(G_{i_1})$ such that $\vec{d_{X_{i_1}}}(x)= \vec r$, and any $y \notin V(G_{i_1})$ such that $\vec{d_{X_{i_1}}}(y)= \vec t$, there exists $s \in S_1$ such that $s$ resolves the pair $\langle x,y \rangle$.\\

\noindent
{\bf (S4)} By construction, $S_{I_1} = S_I$, and thus, $S \cap X_{i_1} = S \cap X_i = S_I = S_{I_1}$.
\end{claimproof}

To conclude, since the sets $S_1$ and $S_2$ are solutions for $I_1$ and $I_2$, respectively, we have that $\dim(I_1) \leq |S_1|$ and $\dim(I_2) \leq |S_2|$. Now, $|S| = |S_1| + |S_2| - |S_I|$, and so, $|S| = \dim(I) \geq \dim(I_1) + \dim(I_2) - |S_I| \geq \min_{\langle J_1,J_2 \rangle \in \mathcal{F}_J(I)} (\dim(J_1) + \dim(J_2) - |S_I|)$.
\end{proof}

\noindent
{\bf Introduce node.} Let $I$ be an instance for an introduce node $i$ with child $i_1$, and let $v \in V(G)$ be such that $X_i = \{v\} \cup X_{i_1}$. Further, let $X_i = \{v_1,\ldots,v_k\}$, where $v = v_k$.

\begin{definition}
\label{def:compatible_introduce}
An instance $I_1$ for $i_1$ is compatible with $I$ of \emph{type 1} if the following hold.
\begin{itemize}
\item {\bf (I1)} $S_I = S_{I_1}$.
\item {\bf (I2)} For all $\vec r \in \ext{I_1}$, there exists $\vec t \in \ext I$ such that $\vec{t^-} = \vec r$.
\item {\bf (I3)} For all $\vec{r} \in \int{I}$, $\vec{r}_k = \min_{1 \leq \ell \leq k-1} (\vec r + \vec{d_{X_{i_1}}}(v))_\ell$ and $\vec{r^-}\in \int{I_1}$. 
\item {\bf (I4)} Let 
$P_1 = \{\langle \vec r, \vec t \rangle \in ([\diam(G)]^{|X_i|})^2 \mid \vec{r}_k \geq 1 \text{ and } \langle \vec{r^-},\vec{t^-} \rangle \in \paire{I_1}\}$
and 
$C_1 = \{\langle \vec{d_{X_i}}(v),\vec t \rangle \in ([\diam(G)]^{|X_i|})^2 \mid \exists \vec u \in \int{I_1}, \vec u \text{ resolves } \langle \vec{d_{X_i}}(v)^-,\vec{t^-} \rangle\}.$
Then, $\paire I \subseteq P_1 \cup C_1$.
\item {\bf (I5)} For all $\vec r \in [\diam(G)]^{|X_i|}$ such that 
\begin{itemize}
\item there exists $x \in V(G_{i_1})$ with $\vec{d_{X_i}}(x) = \vec r$, and
\item no vector in $\ext I$ resolves the pair $\langle x,v \rangle$,
\end{itemize}
$\langle \vec{r^-},\vec{d_{X_{i_1}}}(v) \rangle \in \paire{I_1}$.
\end{itemize}

An instance $I_1$ for $i_1$ is compatible with $I$ of \emph{type 2} if the following hold.
\begin{itemize}
\item {\bf (I'1)} $S_I = S_{I_1} \cup \{v\}$.
\item {\bf (I'2)} $\vec{d_{X_{i_1}}}(v) \in \ext{I_1}$ and, for all $\vec r \in \ext{I_1} \setminus \{\vec{d_{X_{i_1}}}(v)\}$, there exists $\vec t \in \ext I$ such that $\vec{t^-} = \vec r$.
\item {\bf (I'3)} For all $\vec{r} \in \int{I} \setminus \{\vec{d_{X_i}}(v)\}$, $\vec{r}_k = \min_{1 \leq \ell \leq k-1} (\vec r + \vec{d_{X_{i_1}}}(v))_\ell$ and $\vec{r^-}\in \int{I_1}$. 
\item {\bf (I'4)} Let 
$P_2 = \{\langle \vec r, \vec t \rangle \in ([\diam(G)]^{|X_i|})^2 \mid \vec{r}_k \geq 1 \text{ and } \langle \vec{r^-},\vec{t^-} \rangle \in \paire{I_1}\}$
and 
$C_2 = \{\langle \vec r,\vec t \rangle \in ([\diam(G)]^{|X_i|})^2 \mid \vec{r}_k \neq \vec{t}_k)\}.$
Then, $\paire I \subseteq P_2 \cup C_2$. 
\end{itemize}
\end{definition}

We denote by $\mathcal{F}_{1}(I)$ the set of instances for $i_1$ compatible with $I$ of type 1, and by $\mathcal{F}_{2}(I)$ the set of instances for $i_1$ compatible with $I$ of type 2. We aim to prove the following.

\begin{lemma}
\label{lem:introduce_node}
Let $I$ be an instance for an introduce node $i$. Then, 
\[
\dim(I)=\min\;\{ \min_{I_1 \in \mathcal{F}_{1}(I)}\;\{\dim(I_1)\},\min_{I_2 \in \mathcal{F}_{2}(I)}\;\{\dim(I_2)+1\}\}.
\]
\end{lemma}

To prove \Cref{lem:introduce_node}, we prove the following lemmas.

\begin{lemma}
\label{lem:upper_1}
Let $I_1$ be an instance for $i_1$ compatible with $I$ of type 1, and let $S$ be a minimum-size solution for $I_1$. Then, $S$ is a solution for $I$.
\end{lemma}

\begin{proof}
Let us prove that every condition of \Cref{def_instance} is satisfied.\\

\noindent
{\bf (S1)} Let $\langle x,y \rangle$ be a pair of vertices of $G_i$. Assume first that $x \neq v$ and $y \neq v$, and suppose that $S$ does not resolve $\langle x,y \rangle$ (we are done otherwise). Then, since $S$ is a solution for $I_1$ and $I_1$ is compatible with $I$, there exists $\vec r \in \ext I$ such that $\vec{r^-} \in \ext{I_1}$ and $\vec{r^-}$ resolves the pair $\langle x,y \rangle$; but then, $\vec r$ resolves $\langle x,y \rangle$ as well. Assume next that $x = v$ and suppose that no vector in $\ext I$ resolves the pair $\langle x,y \rangle$ (we are done otherwise). Then, since $I_1$ is compatible with $I$ of type 1, $\langle \vec{d_{X_{i_1}}}(y),\vec{d_{X_{i_1}}}(x) \rangle \in \paire{I_1}$, and so, $S$ resolves the pair $\langle x,y \rangle$ as it is a solution for $I_1$.\\

\noindent
{\bf (S2)} Consider $\vec r \in \int I$. Since $I_1$ is compatible with $I$, there exists $\vec t \in \int{I_1}$ such that $\vec r = \vec{t |} \min_{1 \leq \ell \leq k-1} (\vec t + \vec{d_{X_{i_1}}}(v))_\ell$. Now, since $S$ is a solution for $I_1$, there exists $s \in S$ such that $\vec{d_{X_{i_1}}}(s) = \vec t$; but then, $\vec{d_{X_i}}(s) = \vec r$ as $d(s,v) =  \min_{1 \leq \ell \leq k-1} (\vec t + \vec{d_{X_{i_1}}}(v))_\ell$ (indeed, $X_{i_1}$ separates $v$ from $s$).\\

\noindent
{\bf (S3)} Consider $\langle \vec r, \vec t \rangle \in \paire I$. Let $x \in V(G_i)$ be such that $\vec{d_{X_i}}(x) = \vec r$, and let $y \notin V(G_i)$ be such that $\vec{d_{X_i}}(y) = \vec t$. Then, since $I_1$ is compatible with $I$ of type 1, $\langle \vec r, \vec t \rangle \in P_1$ or $\langle \vec r, \vec t \rangle \in C_1$. Now, if the former holds, then, $x \neq v$ and $\langle \vec{r^-},\vec{t^-} \rangle \in \paire{I_1}$, and so, $S$ resolves the pair $\langle x,y \rangle$ as it is a solution for $I_1$. Suppose therefore that the latter holds. Then, $x = v$ and there exists $\vec u \in \int{I_1}$ such that $\vec u$ resolves the pair $\langle x,y \rangle$; but $S$ is a solution for $I_1$, and thus, there exists $s \in S$ such that $\vec{d_{X_{i_1}}}(s) = \vec u$.\\

\noindent
{\bf (S4)} By compatibility of type 1, $S_I = S_{I_1}$, and thus, $S \cap X_i = S \cap (X_{i_1} \cup \{v\}) = S_{I_1} = S_I$.
\end{proof}

\begin{lemma}
\label{lem:upper_2}
Let $I_1$ be an instance for $i_1$ compatible with $I$ of type 2, and let $S$ be a minimum-size solution for $I_1$. Then, $S \cup \{v\}$ is a solution for $I$.
\end{lemma}

\begin{proof}
Let us prove that the conditions of \Cref{def_instance} are satisfied. In the following, we let $S' = S \cup \{v\}$.\\

\noindent
{\bf (S1)} Let $\langle x,y \rangle$ be a pair of vertices of $G_i$ such that $x \neq v$ and $y \neq v$ (it is otherwise clear that the pair is resolved by $v \in S'$). Suppose that $S$ does not resolve the pair $\langle x,y \rangle$ (we are done otherwise). Then, since $S$ is a solution for $I_1$, $\langle x,y \rangle$ is resolved by a vector $\vec r \in \ext{I_1}$. Now, $I_1$ is compatible with $I$, and thus, there exists $\vec t \in \ext I$ such that $\vec{t^-} = \vec r$; but then, $\vec t$ resolves the pair $\langle x,y \rangle$ as well.\\

\noindent
{\bf (S2)} Consider $\vec r \in \int I$ and suppose that $\vec r \neq \vec{d_{X_i}}(v)$ (otherwise $v \in S'$ has $\vec r$ as its distance vector to $X_i$). Then, since $I_1$ is compatible with $I$, there exists $\vec t \in \int{I_1}$ such that $\vec r = \vec{t |} \min_{1 \leq \ell \leq k-1} (\vec t + \vec{d_{X_{i_1}}}(v))_\ell$. Now, $S$ is a solution for $I_1$, and thus, there exists $s \in S$ such that $\vec{d_{X_{i_1}}}(s) = \vec t$; but then, $\vec{d_{X_i}}(s) = \vec r$ as $d(s,v) =  \min_{1 \leq \ell \leq k-1} (\vec t + \vec{d_{X_{i_1}}}(v))_\ell$ (indeed, $X_{i_1}$ separates $v$ from $s$).\\

\noindent
{\bf (S3)} Consider $\langle \vec{r},\vec{t} \rangle \in \paire I$. Let $x \in V(G_i)$ be such that $\vec{d_{X_i}}(x) = \vec r$, and let $y \notin V(G_i)$ be such that $\vec{d_{X_i}}(y) = \vec t$. Suppose that $v$ does not resolve the pair $\langle x,y \rangle$ (we are done otherwise). Then, $\vec{r}_k = d(x,v) = d(y,v) = \vec{t}_k \geq 1$, which implies that $\langle \vec{r^-},\vec{t^-} \rangle \in \paire{I_1}$ as $I_1$ is compatible with $I$ of type 2. But $S$ is a solution for $I_1$, and so, $S$ (and, a fortiori, $S'$) resolves the pair $\langle x,y \rangle$.\\ 

\noindent
{\bf (S4)} By compatibility of type 2, $S_I = S_{I_1} \cup \{v\}$, and thus, $S \cap X_i = S \cap (X_{i_1} \cup \{v\}) = S_{I_1} \cup \{v\} = S_I$.
\end{proof} 

As a consequence of Lemmas \ref{lem:upper_1} and \ref{lem:upper_2}, the following holds.

\begin{lemma}
Let $I$ be an instance for an introduce node $i$. Then, 
\[
\dim(I) \leq \min\;\{ \min_{I_1 \in \mathcal{F}_{1}(I)}\;\{\dim(I_1)\},\min_{I_2 \in \mathcal{F}_{2}(I)}\;\{\dim(I_2)+1\}\}.
\]
\end{lemma}

\begin{lemma}
\label{lem:lower_1}
Let $S$ be a minimum-size solution for $I$ such that $v \notin S$. Then, there exists $I_1 \in \mathcal{F}_{1}(I)$ such that $S$ is a solution for $I_1$.
\end{lemma}

\begin{proof}
Let $I_1$ be the instance for $i_1$ defined as follows. 
\begin{itemize}
\item $S_{I_1} = S_I$, $\ext{I_1} = \{\vec{r^-} \mid \vec r \in \ext I\}$, and $\int{I_1} = \vec{d_{X_{i_1}}}(S)$.
\item For any $\langle \vec r, \vec t \rangle \in ([\diam(G)]^{|X_{i_1}|})^2$, let $R_{\langle \vec r,\vec t \rangle} = \{\langle x,y \rangle \in V(G_{i_1}) \times V(G-G_{i_1}) \mid \vec{d_{X_{i_1}}}(x) = \vec r \text{ and } \vec{d_{X_{i_1}}}(y) = \vec t\}$. If $S$ resolves every pair in $R_{\langle \vec r,\vec t \rangle}$, then we add $\langle \vec r,\vec t \rangle$ to $\paire{I_1}$.
\end{itemize}
Let us prove that $I_1 \in \mathcal{F}_{1}(I)$ and that $S$ is a solution for $I_1$.

\begin{claim}
The constructed instance $I_1$ is compatible with $I$ of type 1.
\end{claim} 

\begin{claimproof}
It is clear that conditions {\bf (I1)} and {\bf (I2)} of \Cref{def:compatible_introduce} hold; let us show that the remaining conditions hold as well.\\ 

\noindent
{\bf (I3)} Since $S$ is a solution for $I$, for every $\vec r \in \int I$, there exists $s \in S$ such that $\vec{d_{X_i}}(s) = \vec r$; but then, $\vec{r^-} = \vec{d_{X_{i_1}}}(s)$ and $\vec{r}_k = d(s,v) = \min_{1 \leq \ell \leq k-1} (\vec r + \vec{d_{X_{i_1}}}(v))_\ell$ as $v \notin S$ and $X_{i_1}$ separates $s$ from $v$. \\

\noindent
{\bf (I4)} Consider $\langle \vec r,\vec t \rangle \in \paire I$ and assume first that $\vec r \neq \vec{d_{X_i}}(v)$. Then, for any $x \in V(G_i)$ such that $\vec{d_{X_i}}(x) = \vec r$ and any $y \notin V(G_i)$ such that $\vec{d_{X_i}}(y) = \vec t$, in fact $x \in V(G_{i_1})$ and $S$ resolves the pair $\langle x,y \rangle$ as it is a solution for $I$; and since $\vec{d_{X_{i_1}}}(x) = \vec{r^-}$ and $\vec{d_{X_{i_1}}}(y) = \vec{t^-}$, it follows by construction that $\langle \vec{r^-},\vec{t^-} \rangle \in \paire{I_1}$, and thus, $\langle \vec r,\vec t \rangle \in P_1$. Second, assume that $\vec r = \vec{d_{X_i}}(v)$ (note that $v$ is then the only vertex in $G_i$ with distance vector $\vec r$ to $X_i$). Let $y \notin V(G_i)$ be such that $\vec{d_{X_i}}(y) = \vec t$. Then, since $S$ is a solution for $I$, there exists $s \in S$ such that $s$ resolves the pair $\langle x,y \rangle$, which implies that $\vec{d_{X_{i_1}}}(s)$ resolves $\langle \vec{d_{X_i}}(v)^-,\vec{t^-} \rangle$, and thus, $\langle \vec r,\vec t \rangle \in C_1$.\\ 

\noindent
{\bf (I5)} Consider $\vec r \in [\diam(G)]^{|X_i|}$ for which there exists $x \in V(G_{i_1})$ such that $\vec{d_{X_i}}(x) = \vec r$, and assume that no vector in $\ext I$ resolves the pair $\langle x,v \rangle$. Then, $S$ must resolve the pair $\langle x,v \rangle$ as it is a solution for $I$; and since this holds for any vertex with distance vector $\vec r$ with respect to $X_i$, it follows by construction that $\langle \vec{r^-},\vec{d_{X_{i_1}}}(v) \rangle \in \paire{I_1}$. 
\end{claimproof}

\begin{claim}
$S$ is a solution for $I_1$.
\end{claim}

\begin{claimproof}
Let us prove that the conditions of \Cref{def_instance} are satisfied.\\

\noindent
{\bf (S1)} Let $\langle x,y \rangle$ be a pair of vertices of $G_{i_1}$ and suppose that $S$ does not resolve the pair $\langle x,y \rangle$ (we are done otherwise). Then, since $S$ is a solution for $I$, there exists $\vec r \in \ext I$ such that $\vec r$ resolves $\langle x,y \rangle$; but then, $\vec{r^-} \in \ext{I_1}$ resolves $\langle x,y \rangle$.\\

\noindent
{\bf (S2)} Readily follows from the fact that $\int{I_1}= \vec{d_{X_{i_1}}}(S)$.\\

\noindent
{\bf (S3)} By construction, for any $\langle \vec r, \vec t \rangle \in \paire{I_1}$, any $x \in V(G_{i_1})$ such that $\vec{d_{X_{i_1}}}(x) = \vec r$, and any $y \notin V(G_{i_1})$ such that $\vec{d_{X_{i_1}}}(y) = \vec t$, $S$ resolves the pair $\langle x,y \rangle$.\\

\noindent
{\bf (S4)} By construction, $S_{I_1}= S_I$, and thus, $S \cap X_{i_1} = S \cap (X_i \setminus \{v\}) = S_I = S_{I_1}$ as $v \notin S_I$ by assumption.
\end{claimproof}

The lemma now follows from the above two claims.
\end{proof}

\begin{lemma}
\label{lem:lower_2}
Let $S$ be a minimum-size solution for $I$ such that $v \in S$. Then, there exists $I_1\in \mathcal{F}_2(I)$ such that $S \setminus \{v\}$ is a solution of $I_1$.
\end{lemma}

\begin{proof}
Let $I_1$ be the instance for $i_1$ defined as follows.
\begin{itemize} 
\item $S_{I_1} = S_I \setminus \{v\}$, $\ext{I_1} = \{\vec{d_{X_{i_1}}}(v)\} \cup \{\vec{r^-} \mid \vec r \in \ext I\}$, and $\int{I_1} = \vec{d_{X_{i_1}}}(S\setminus \{v\})$.
\item For every $\langle \vec r, \vec t \rangle \in ([\diam(G)]^{|X_{i_1}|})^2$, let $R_{\langle \vec r, \vec t \rangle} = \{\langle x,y \rangle \in V(G_{i_1}) \times V(G- G_{i_1}) \mid \vec{d_{X_{i_1}}}(x) = \vec r \text{ and } \vec{d_{X_{i_1}}}(y) = \vec t\}$. If $S \setminus \{v\}$ resolves every pair in $R_{\langle \vec r, \vec t \rangle}$, then we add $\langle \vec r, \vec t \rangle$ to $\paire{I_1}$.
\end{itemize}
Let us prove that $I_1\in \mathcal{F}_2(I)$ and that $S\setminus \{v\}$ is a solution of $I_1$. In the following, we let $S' = S \setminus \{v\}$.

\begin{claim}
The constructed instance $I_1$ is compatible with $I$ of type 2.
\end{claim}

\begin{claimproof}
It is clear that conditions {\bf (I'1)} and {\bf (I'2)} of \Cref{def:compatible_introduce} hold; let us show that the remaining conditions hold as well. \\

\noindent
{\bf (I'3)} Since $S$ is a solution for $I$, for every $\vec r \in \int I \setminus \{\vec{d_{X_i}}(v)\}$, there exists $s \in S$ such that $\vec{d_{X_i}}(s) = \vec r$ (in particular, $s \neq v$); but then, $\vec{r^-} = \vec{d_{X_{i_1}}}(s)$ and $\vec{r^-} \in \int{I_1}$ by construction.\\

\noindent
{\bf (I'4)} Consider $\langle \vec r, \vec t \rangle \in \paire I$. Let $x \in V(G_i)$ be such that $\vec{d_{X_i}}(x) = \vec r$ and let $y \notin V(G_i)$ be such that $\vec{d_{X_i}}(y) = \vec t$. If $v$ resolves the pair $\langle x,y \rangle$, then $\vec{r}_k = d(x,v) \neq d(y,v) = \vec{t}_k$, and so, $\langle \vec r, \vec t \rangle \in C_2$. Suppose therefore that $v$ does not resolve $\langle x,y \rangle$. Then, since $S$ is a solution for $I$, it must be that $S \setminus \{v\}$ resolves the pair $\langle x,y \rangle$; and since this holds for any pair with distance vectors $\langle \vec r,\vec t \rangle$ to $X_i$, $\langle \vec{r^-},\vec{t^-} \rangle \in \paire{I_1}$ by construction, and so, $\langle \vec r, \vec t \rangle \in P_2$.
\end{claimproof}

\begin{claim} 
$S'$ is a solution for $I_1$.
\end{claim}

\begin{claimproof}
Let us prove that the conditions of \Cref{def_instance} are satisfied.\\

\noindent
{\bf (S1)} Let $\langle x,y \rangle$ be a pair of vertices of $G_{i_1}$ and suppose that $S'$ does not resolve $\langle x,y \rangle$ (we are done otherwise). Then, since $S$ is a solution for $I$, either $S \setminus S' = \{v\}$ resolves $\langle x,y \rangle$, in which case $\vec{d_{X_{i_1}}}(v) \in \ext{I_1}$ resolves $\langle x,y \rangle$; or there exists a vector $\vec r \in \ext I$ resolving $\langle x,y \rangle$, in which case $\vec{r^-} \in \ext{I_1}$ resolves the pair as well.\\

\noindent
{\bf (S2)} Readily follows from the fact that $\int{I_1} = \vec{d_{X_{i_1}}}(S')$.\\

\noindent
{\bf (S3)} By construction, for every $\langle \vec r,\vec t \rangle \in \paire{I_1}$, any $x \in V(G_{i_1})$ such that $\vec{d_{X_{i_1}}}(x) = \vec r$, and any $y \notin V(G_{i_1})$ such that $\vec{d_{X_{i_1}}}(y) = \vec t$, $S'$ resolves the pair $\langle x,y \rangle$.\\ 

\noindent
{\bf (S4)} By construction, $S_{I_1} = S_I \setminus \{v\}$, and so, $S \cap X_{i_1} = S \cap (X_i \setminus \{v\}) = S_I \setminus \{v\}$.
\end{claimproof}

The lemma now follows from the above two claims.
\end{proof}

As a consequence of Lemmas \ref{lem:lower_1} and \ref{lem:lower_2}, the following holds.

\begin{lemma}
Let $I$ be an instance for an introduce node $i$. Then, 
\[
\dim(I) \geq \min\;\{ \min_{I_1 \in \mathcal{F}_{1}(I)}\;\{\dim(I_1)\},\min_{I_2 \in \mathcal{F}_{2}(I)}\;\{\dim(I_2)+1\}\}.
\]
\end{lemma}

\noindent
{\bf Forget node.} Let $I$ be an instance for a forget node $i$ with child $i_1$, and let $v \in V(G)$ be such that $X_i = X_{i_1} \setminus \{v\}$. Further, let $X_{i_1} = \{v_1, \ldots, v_k\}$, where $v = v_k$.

\begin{definition}
\label{def:compatible_forget}
An instance $I_1$ for $i_1$ is compatible with $I$ if the following hold.
\begin{itemize}
\item {\bf (F1)} $S_I = S_{I_1} \setminus \{v\}$.
\item {\bf (F2)} For all $\vec r \in \ext{I_1}$, there exists $\vec t \in \ext I$ such that $\vec{r^-} = \vec t$.
\item {\bf (F3)} For all $\vec r \in \int I$, there exists $\vec t \in \int{I_1}$ such that $\vec{t^-} = \vec r$.
\item {\bf (F4)} For all $\vec r, \vec t \in [\diam(G)]^{|X_i|}$, let 
$R_{\langle \vec r, \vec t \rangle} = \{\langle x,y \rangle \in V(G_i) \times V(G-G_i) \mid \vec{d_{X_i}}(x) = \vec r \text{ and } \vec{d_{X_i}}(y) = \vec t\}.$
Then, 
$\paire I \subseteq \{\langle \vec r, \vec t \rangle \in ([\diam(G)]^{|X_i|})^2 \mid \forall (x,y) \in R_{\langle \vec r, \vec t \rangle}, \langle \vec{r |} d(x,v),\allowbreak \vec{t |} d(y,v) \rangle \in \paire{I_1}\}.$  
\end{itemize}
\end{definition}

We denote by $\mathcal{F}_{F}(I)$ the set of instances for $i_1$ compatible with $I$. We aim to prove the following.

\begin{lemma}
\label{lem:forget_node}
Let $I$ be an instance for a forget node $i$. Then,
\[
\dim(I)= \min_{I_1 \in \mathcal{F}_{F}(I)}\;\{\dim(I_1)\}.
\]
\end{lemma}

To prove \Cref{lem:forget_node}, we prove the following lemmas.

\begin{lemma}
Let $I_1$ be an instance for $i_1$ compatible with $I$, and let $S$ be a minimum-size solution for $I_1$. Then, $S$ is a solution for $I$. In particular,
\[
\dim(I) \leq \min_{I_1 \in \mathcal{F}_{F}(I)}\;\{\dim(I_1)\}.
\]
\end{lemma}

\begin{proof}
Let us prove that  the conditions of \Cref{def_instance} are satisfied.

\noindent
{\bf (S1)} Let $\langle x,y \rangle$ be a pair of vertices of $G_i$. Since $V(G_i) = V(G_{i_1})$ and $S$ is a solution for $I_1$, either $S$ resolves the pair $\langle x,y \rangle$, in which case we are done; or there exists $\vec t \in \ext{I_1}$ such that $\vec t$ resolves $\langle x,y \rangle$. In the latter case, since $I_1$ is compatible with $I$, there then exists $\vec r \in \ext I$ such that $\vec{t^-} = \vec r$; but then, $\vec r$ resolves the pair $\langle x,y \rangle$.\\

\noindent
{\bf (S2)} Consider $\vec r \in \int I$. Since $I_1$ is compatible with $I$, there exists $\vec t \in \int{I_1}$ such that $\vec{t^-} = \vec r$; and since $S$ is a solution for $I$, there exists $s \in S$ such that $\vec{d_{X_{i_1}}}(s) = \vec t$. Now, note that $\vec{d_{X_i}}(s) = \vec{t^-} = \vec r$.\\

\noindent
{\bf (S3)} Consider $\langle \vec r,\vec t \rangle \in \paire I$. Let $x \in V(G_i)$ be such that $\vec{d_{X_i}}(x) = \vec r$ and let $y \notin V(G_i)$ be such that $\vec{d_{X_i}}(y) = \vec t$. Then, since $I_1$ is compatible with $I$, there exists $\langle \vec u, \vec v \rangle \in \paire{I_1}$ such that $\vec{d_{X_{i_1}}}(x) = \vec u$ and $\vec{d_{X_{i_1}}}(y) = \vec v$; and since $S$ is a solution for $I_1$, $S$ resolves the pair $\langle x,y \rangle$.\\

\noindent
{\bf (S4)} By compatibility, $S_I = S_{I_1} \setminus \{v\}$, and so, $S \cap X_i = S \cap (X_{i_1} \setminus \{v\}) = S_{I_1} \setminus \{v\}$.
\end{proof}

\begin{lemma}
Let $S$ be a minimum-size solution for $I$. Then, there exists $I_1 \in \mathcal{F}_{F}(I)$ such that $S$ is a solution for $I_1$. In particular,
\[
\dim(I) \geq \min_{I_1 \in \mathcal{F}_{F}(I)}\;\{\dim(I_1)\}.
\]
\end{lemma}

\begin{proof}
Let $I_1$ be the instance for $i_1$ defined as follows.
\begin{itemize}
\item $S_{I_1} = S \cap X_{i_1}$, $\ext{I_1} = \{\vec r \in [\diam(G)^{|X_{i_1}|}] \mid \vec{r^-} \in \ext I \text{ and } \vec{r}_k = \min_{1 \leq \ell \leq k-1} (\vec{r^-} + \vec{d_{X_{i_1}}}(v))_\ell\}$, and $\int{I_1} = \vec{d_{X_{i_1}}}(S)$.
\item For every $\langle \vec r, \vec t \rangle \in ([\diam(G)]^{|X_{i_1}|})^2$, let $R_{\langle \vec r, \vec t \rangle} = \{\langle x,y \rangle \in V(G_{i_1}) \times V(G-G_{i_1}) \mid \vec{d_{X_{i_1}}}(x) = \vec r \text{ and } \vec{d_{X_{i_1}}}(y) = \vec t\}$. If $S$ resolves every pair in $R_{\langle \vec r, \vec t \rangle}$, then we add $\langle \vec r, \vec t \rangle$ to $\paire{I_1}$.
\end{itemize}
Let us prove that $I_1 \in \mathcal{F}_{F}(I)$ and that $S$ is a solution for $I_1$.

\begin{claim}
The constructed instance $I_1$ is compatible with $I$.
\end{claim}

\begin{claimproof}
It is clear that conditions {\bf (F1)} and {\bf (F2)} of \Cref{def:compatible_forget} hold; let us show that the remaining conditions hold as well.\\

\noindent
{\bf (F3)} Since $S$ is a solution for $I$, for every $\vec r \in \int I$, there exists $s \in S$ such that $\vec{d_{X_i}}(s) = \vec r$; but then, $\vec{d_{X_{i_1}}}(s)^- = \vec{d_{X_i}}(s)$, where $\vec{d_{X_{i_1}}}(s) \in \int{I_1}$ by construction.\\

\noindent
{\bf (F4)} Consider $\langle \vec r, \vec t \rangle \in \paire I$. Let $x \in V(G_i)$ be such that $\vec{d_{X_i}}(x) = \vec r$ and let $y \notin V(G_i)$ be such that $\vec{d_{X_i}}(y) = \vec t$. Then, $S$ resolves the pair $\langle x,y \rangle$ as it is a solution for $I$; and since this holds for every pair $\langle a,b \rangle \in V(G_i) \times V(G-G_i)$ such that $(\vec{d_{X_{i_1}}}(a),\vec{d_{X_{i_1}}}(b)) = (\vec{r |} d(x,v),\vec{t |} d(y,v))$, by construction $\langle \vec{r |} d(x,v),\vec{t |} d(y,v) \rangle \in \paire{I_1}$.
\end{claimproof}

\begin{claim}
$S$ is a solution for $I_1$.
\end{claim}

\begin{claimproof}
Let us prove that the conditions of \Cref{def_instance} hold.\\ 

\noindent
{\bf (S1)} Let $\langle x,y \rangle$ be a pair of vertices of $G_{i_1}$. Since $V(G_i) = V(G_{i_1})$ and $S$ is a solution for $I$, either $S$ resolves the pair $\langle x,y \rangle$, in which case we are done: or there exists $\vec t \in \ext I$ such that $\vec t$ resolves $\langle x,y \rangle$. In the latter case, by construction $\vec{t |} a \in \ext{I_1}$, where $a = \min_{1 \leq \ell \leq k-1} (\vec t + \vec{d_{X_{i_1}}}(v))_\ell$; but then, $\vec{t |} a$ resolves $\langle x,y \rangle$.\\

\noindent
{\bf (S2)} Readily follows from the fact that $\int{I_1} = \vec{d_{X_{i_1}}}(S)$.\\ 

\noindent
{\bf (S3)} By construction, for every $\langle \vec r, \vec t \rangle \in \paire{I_1}$, any $x \in V(G_{i_1})$ such that $\vec{d_{X_{i_1}}}(x) = \vec r$, and any $y \notin V(G_{i_1})$ such that $\vec{d_{X_{i_1}}}(y) = \vec t$, $S$ resolves the pair $\langle x,y \rangle$.\\

\noindent
{\bf (S4)} By construction, $S_{I_1} = S \cap X_{i_1}$.
\end{claimproof}

The lemma now follows from the above two claims.
\end{proof}


To complete the proof of \Cref{{thm:algo-diam-tw-MD}}, let us now explain how the algorithm proceeds. Given a nice tree decomposition $(T,\mathcal{X})$ of a graph $G$ rooted at node $r \in V(T)$, the algorithm computes the extended metric dimension for all possible instances in a bottom-up traversal of $T$. It computes the values for leaf nodes using \Cref{lem:leaf_node}, for join nodes using \Cref{lem:join_node}, for introduce nodes using \Cref{lem:introduce_node}, and for forget nodes using \Cref{lem:forget_node}. The correctness of this algorithm follows from these lemmas and the following.

\begin{lemma}
Let $G$ be a graph and let $(T,\{X_i:i\in V(T)\})$ be a nice tree decomposition of $G$ rooted at node $r \in V(T)$. Then,
\[
\md(G) = \min_{S_r \subseteq X_r} \dim(X_r,S_r,\emptyset,\emptyset,\emptyset).
\]
\end{lemma}

\begin{proof}
Let $S$ be a minimum-size resolving set of $G$. Then, by \Cref{def_instance}, $S$ is a solution for the {\sc EMD} instance $(X_r,S\cap X_r,\emptyset,\emptyset,\emptyset)$, and so,
\[
\min_{S_r \subseteq X_r} \dim(X_r,S_r,\emptyset,\emptyset,\emptyset) \leq \dim(X_r,S\cap X_r,\emptyset,\emptyset,\emptyset) \leq \md(G).
\]
Conversely, let $S' \subseteq X_r$ be a set attaining the minimum above, and let $S$ be a minimum-size solution for the {\sc EMD} instance $(X_r,S',\emptyset,\emptyset,\emptyset)$. Then, by \Cref{def_instance}, every vertex of $G_r = G$ is resolved by $S$, and so,
\[
\md(G) \leq \dim(X_r,S',\emptyset,\emptyset,\emptyset) = \min_{S_r \subseteq X_r} \dim(X_r,S_r,\emptyset,\emptyset,\emptyset),
\]
which concludes the proof.
\end{proof}

To get the announced complexity, observe first that, at each node $i \in V(T)$, there are at most $2^{|X_i|} \cdot 2^{\diam(G)^{|X_i|}} \cdot 2^{\diam(G)^{|X_i|}} \cdot 2^{\diam(G)^{2|X_i|}}$ possible instances to consider, where $|X_i| = \OO(\tw(G))$. Since $T$ has $\OO(\tw(G)\cdot n)$ nodes, there are in total $\OO(\alpha(\tw(G)) \cdot \tw(G) \cdot n)$ possible instances, where $\alpha(k) = 2^{k} \cdot 2^{\diam(G)^{k}} \cdot 2^{\diam(G)^{k}} \cdot 2^{\diam(G)^{2k}}$. The running time of the algorithm then follows from these facts and the next lemma (note that to avoid repeated computations, we can first compute the distance between every pair of vertices of $G$ in $n^{\OO(1)}$ time, as well as all possible distance vectors to a bag from the possible distance vectors to its child/children).

\begin{lemma}
Let $I$ be an {\sc EMD} instance for a node $i \in V(T)$, and assume that, for every child $i_1$ of $i$ and every {\sc EMD} instance $I_1$ for $i_1$ compatible with $I$, $\dim(I_1)$ is known. Then, $\dim(I)$ can be computed in time $\OO(\alpha(|X_i|)) \cdot n^{\OO(1)}$. 
\end{lemma} 

\begin{proof}
If $i$ is a leaf node, then $\dim(I)$ can be computed in constant time by \Cref{lem:leaf_node}. Otherwise, let us prove that one can compute all compatible instances in the child nodes in the announced time (note that $i$ has at most two child nodes). First, given a 5-tuple $(X_{i_1},S_{I_1},\int{I_1},\ext{I_1},\allowbreak \paire{I_1})$, checking whether it is an {\sc EMD} instance can be done in $\OO(|I_1|) \cdot n^{\OO(1)}$ time; and the number of such 5-tuples is bounded by $\alpha(|X_{i_1}|)$. It is also not difficult to see that checking for compatibility can, in each case, be done in $\OO(|I|) \cdot n^{\OO(1)}$ time. Now, note that, by \Cref{def_instance}, $|I| = \OO(\diam(G)^{\OO(|X_i|)}$ and thus, computing all compatible instances can indeed be done in $\OO(\alpha(|X_i|)) \cdot n^{\OO(1)}$. Then, since computing the minimum using the formulas of Lemmas \ref{lem:join_node}, \ref{lem:introduce_node}, and \ref{lem:forget_node} can be done in $\OO(\alpha(|X_i|)$ time, the lemma follows. 
\end{proof}
\renewcommand{\int}[1]{D_{int}{{(#1)}}}
\newcommand{\extext}[1]{D_{ext/ext}{{(#1)}}}
\newcommand{\intint}[1]{D_{int/int}{{(#1)}}}
\renewcommand{\ext}[1]{D_{ext}{{(#1)}}}
\newcommand{\add}[1]{D_{add}{{(#1)}}}
\renewcommand{\I}[1]{ I_{#1}=(X_{i_{#1}}, S_{I_{#1}}, \int {I_{#1}},\ext {I_{#1}},\add {I_{#1})}}
\renewcommand{\T}[1]{T(#1)}
\renewcommand{\vec}{\mathbf}
\renewcommand{\dimx}[1]{\dim(#1)}

\subsection{Dynamic Programming Algorithm for \gsfull}
\label{subsec:algo-tw-diam-GD}
%
%

In this subsection, we prove the following theorem.
\begin{theorem}
\label{thm:algo-diam-tw-GD}
\gsfull\ admits an algorithm running in time
$2^{\diam^{\mathcal{O}(\tw)}} \cdot n^{\OO(1)}$.
\end{theorem}
The proof follows along the same lines as that of the 
proof of Theorem~\ref{thm:algo-diam-tw-MD}.

\noindent\textbf{Overview.} We first give an intuitive description of the dynamic programming scheme. At each step of the algorithm, we consider a bounded number of \emph{solution types}, depending on the properties of the solution vertices with respect to the current bag. At a given dynamic programming step, we will assume that the current solution covers all vertices in $G_i$. Such a vertex may be covered by (1) two vertices in $G_i$, (2) a vertex in $G_i$ and a vertex in $G-G_i$, or (3) two vertices in $G-G_i$. 

Any bag $X_i$ of the tree decomposition whose node $i$ lies on a path between two join nodes in $T$, forms a separator of $G$: there are no edges between the vertices of $G_i-X_i$ and $G-G_ i$. For a vertex $v$ not in $X_i$, we consider its distance-vector to the vertices of $X_i$; the distance-vectors induce an equivalence relation on the vertices of $G-X_i$, whose classes we call \emph{$X_i$-classes}. Consider the two subgraphs $G_i$ and $G-G_i$. Given a vertex $z$ in $G_i$, any two solution vertices $x,y$ from $G-G_i$ that are in the same $X_i$-class,
will cover together with $z$ the exact same vertices from $G_i$, that is, a vertex $u$ of $G_i$ is covered by $z$ and $x$ if and only if it is covered by $z$ and $y$. Thus, for case (2), it is irrelevant whether $x$ or $y$ will be in a geodetic set, and it is sufficient to know that a vertex of their $X_i$-class will eventually be chosen. Similarly, for any four vertices $x_1,x_2,y_1,y_2$ of $G-G_i$ such that $x_1,x_2$ ($y_1,y_2$, respectively) are in the same $X_i$-class \emph{and} $d(x_1,y_1)= d(x_2,y_2)$, we have that $x_1$ and $y_1$ cover exactly the same vertices in $G_i$ as $x_2$ and $y_2$. Thus, for case (3), it is irrelevant whether $x_1,y_1$ or $x_2,y_2$ will be in the geodetic set, and it is sufficient to know that a vertex from each of their $X_i$-classes whose distance between them is $d(x_1,y_1)$ will eventually be chosen.   

The same idea is used to ``remember'' the previously computed solution: it is sufficient to remember the $X_i$-classes of the vertices in the previously computed geodetic set, as well as pairs $C_1,C_2$ of $X_i$-classes together with an integer $d$ corresponding to the distance between any vertex in $C_1$ and any vertex in $C_2$, rather than the vertices themselves.

Keeping track, in the aforementioned way, of the ``past'' and ``future'' solution, is sufficient when processing a join node $i$. Indeed, for a join node $i$ with children $i_1$,$i_2$, a vertex of $G_i = G_{i_1} \cup G_{i_2}$ may be covered by: two vertices from $G-G_i$; a vertex from $G_{i_1}$ and a vertex from $G - G_i$; two vertices from $G_{i_1}$; a vertex from $G_{i_1}$ and a vertex from $G_{i_2}$; two vertices from $G_{i_2}$; a vertex from $G_{i_2}$ and a vertex from $G-G_i$. This is also sufficient when processing an introduce node $i$ where a new vertex $v$ is introduced (i.e., added to the child bag $X_{i'}$ to form $X_i$). Indeed, we may check that $v$ is either covered by two vertices from $G_{i_1}$; a vertex of $G_{i_1}$ and a vertex of $G-G_i$; or two vertices of $G-G_i$. If this does not hold, then we add $v$ into the solution.

For a bag $X_i$ and a vertex $v$ not in $X_i$, the number of possible distance vectors to the vertices of $X_i$ is at most $\diam(G)^{|X_i|}$. Thus, a solution for bag $X_i$ will consist of: (i) the subset of vertices of $X_i$ selected in the solution; (ii) a subset of the $\diam(G)^{|X_i|}$ possible vectors to denote the $X_i$-classes from which the currently computed solution (for $G_i$) contains at least one vertex in the geodetic set; (iii) a subset of the $\diam(G)^{|X_i|}$ possible vectors denoting the $X_i$-classes from which the future solution needs at least one vertex of $G-G_i$ in the geodetic set; (iv) a subset of the $\diam(G)^{|X_i|} \times \diam(G)^{|X_i|} \times \diam(G)$ possible elements representing the pairs of $X_i$-classes and their distance to each other from which the currently computed solution (for $G_i$) contains at least two vertices in the geodetic set; (v) a subset of the $\diam(G)^{|X_i|} \times \diam(G)^{|X_i|} \times \diam(G)$ possible elements representing the pairs of $X_i$-classes and their distance from which the future solution needs at least two vertices of $G-G_i$ in the geodetic set.\\

\noindent\textbf{Formal description.} Before presenting the dynamic program, we first introduce some useful definitions and lemmas (see also \Cref{sec:algo-diam-tw-MD} for missing definitions). 

For a set $S$, we denote by $\mathcal{P}_2(S)$ the set of subsets of $S$ of size 2. Given a graph $G$ and a set $S \subseteq V(G)$, we say that a vertex $x \in V(G)$ is \emph{covered by $S$} if either $x \in S$ or there exist $u,v \in S$ such that $x$ lies on a shortest path from $u$ to $v$. The smallest size of a geodetic set for $G$ is denoted by $\gs(G)$.


\begin{definition}
Let $\vec{r_1}, \vec{r_2}$, and $\vec{r_3}$ be three vectors of size $k$, and let $d$ be an integer. We say that $\vec{r_3}$ is \emph{covered} by $(\{\vec{r_1},\vec{r_2}\},d)$ if \[\min_{1 \leq i \leq k} \vec{(r_1+r_3)}_i  + \min_{1 \leq i \leq k} \vec{(r_2+r_3)}_i = d.\]
\end{definition}

\begin{definition}
Let $G$ be a graph and let $X = \{v_1, \ldots,v_k\}$ be a subset of vertices of $G$. Given a vertex $x$ of $G$, the \emph{distance vector} $\vec{d_X}(x)$ of $x$ to $X$ is the vector of size $k$ such that for all $1 \leq j \leq k$, $\vec{d_X}(x)_j=d(x,v_j)$. For a set $S \subseteq V(G)$, we let $\vec{d_X}(S) = \{\vec{d_X}(s) \mid s \in S\}$.
\end{definition}

\begin{definition}
Let $G$ be a graph and let $X = \{v_1,\ldots,v_k\}$ be a subset of vertices of $G$. 
\begin{itemize}
\item Let $\vec{r_1},\vec{r_2}$ be two vectors of size $k$ and let $d$ be an integer. Then, for any $x \in V(G)$, we say that $x$ is \emph{covered} by $(\{\vec{r_1},\vec{r_2}\},d)$ if $\vec{d_X}(x)$ is covered by $(\{\vec{r_1},\vec{r_2}\},d)$.
\item Let $x,y$ be two vertices of $G$ and let $\vec r$ be a vector of size $k$. We say that $\vec r$ is \emph{covered} by $x$ and $y$ if $\vec r$ is covered by $(\{\vec{d_X}(x), \vec{d_X}(y)\},d(x,y))$. More generally, given a set $S$ of vertices of $G$, we say that $\vec r$ is \emph{covered} by $S$ if there exist $x,y \in S$ such that $\vec r$ is covered by $x$ and $y$.
\item Let $s$ be a vertex of $G$ and let $\vec r$ be a vector of size $k$. Then, for any $x \in V(G)$, we say that $x$ is \emph{covered} by $s$ and $\vec r$ if $d(s,x) + \min_{1 \leq j \leq k} (\vec{d_X}(x) + \vec r)_j = \min_{1 \leq j \leq k} (\vec{d_X}(s) + \vec r)_j$. 
\end{itemize}
\end{definition}

\begin{lemma}
\label{lem:covering}
Let $X = \{v_1, \ldots, v_k\}$ be a separator of a graph $G$, and let $G_1$ be a connected component of $G - X$. Further, let $x \in V(G_1) \cup X$.
\begin{enumerate} 
\item Let $\vec{r_1},\vec{r_2}$ be two vectors of size $k$, and let $d$ be an integer. If $x$ is covered by $(\{\vec{r_1},\vec{r_2}\},d)$, then, for any $u,v \in V(G-G_1)$ such that $\vec{d_X}(u) = \vec{r_1}$, $\vec{d_X}(v) = \vec{r_2}$, and $d(u,v) = d$, $x$ is covered by $u$ and $v$.
\item Let $s$ be a vertex of $V(G_1) \cup X$ and let $\vec r$ be a vector of size $k$. If $x$ is covered by $s$ and $\vec r$, then, for any $u \in V(G-G_1)$ such that $\vec{d_X}(u) = \vec r$, $x$ is covered by $s$ and $u$.
\item Let $\vec{r_1},\vec{r_2}$ be two vectors of size $k$. If $x$ is covered by $(\{\vec{r_1},\vec{r_2}\},\min_{1 \leq j \leq k}(\vec{r_1} + \vec{r_2})_j)$, then, for any $u \in V(G-G_1)$ such that $\vec{d_X}(u) = \vec{r_1}$, $x$ is covered by $u$ and $\vec{r_2}$.
\item Let $\vec r$ be a vector of size $k$ and let $u,v$ be two vertices of $G-G_1$. If $\vec r$ is covered by $u$ and $v$, then, for any $w \in V(G_1) \cup X$ such that $\vec{d_X}(w) = \vec r$, $w$ is covered by $u$ and $v$.
\end{enumerate}
\end{lemma}

\begin{proof}
To prove item (1), it suffices to note that since $X$ separates $x$ from $u$, $d(x,u) = \min_{1 \leq j \leq k} (\vec{d_X}(x) + \vec{d_X}(u))_j$ (note that if $x$ or $u$ belongs to $X$, then surely this equality holds as well); and for the same reason, $d(x,v) = \min_{1 \leq j \leq k} (\vec{d_X}(x) + \vec{d_X}(v))_j$. Now, $x$ is covered by $(\{\vec{d_X}(u),\vec{d_X}(v)\},d(u,v))$, and so, $d(u,x) + d(x,v) = d(u,v)$ by definition. Items (2), (3), and (4) follow from similar arguments. 
\end{proof}

We now define the problem solved at each step of the dynamic programming algorithm, called \textsc{Extended Geodetic Set} ({\sc EGS} for short), whose instances are defined as follows.

\begin{definition}
Let $G$ be a graph and let $(T,\{X_i:i \in V(T)\})$ be a tree decomposition of $G$. For a node $i$ of $T$, an \emph{instance} of {\sc EGS} is a $6$-tuple $I = (X_i, S_I, \int I, \ext I, \intint I,\allowbreak \extext I)$ composed of the bag $X_i$ of $i$, a subset $S_I$ of $X_i$, and four sets satisfying the following.
\begin{itemize}
    \item $\int I, \ext I \subseteq [\diam(G)]^{|X_i|}$.
    \item $\intint I,\extext I \subseteq \mathcal{P}_2([\diam(G)]^{|X_i|}) \times [\diam(G)]$.
    \item For each $\vec r \in \ext I$, there exists $x \notin V(G_i)$ such that $\vec{d_{X_i}}(x) = \vec r$.
    \item For each $(\{\vec r, \vec t\},d) \in \extext I$, there exist $x,y \notin V(G_i)$ such that $\vec{d_{X_i}}(x) = \vec r$, $\vec{d_{X_i}}(y) = \vec t$ and $d(x,y) = d$.
    \item $S_I \neq \emptyset$ or $\ext I \neq \emptyset$ or $\extext I \neq \emptyset$.
\end{itemize} 
\end{definition}

\begin{definition}
\label{def_instance-GS}
A set $S \subseteq V(G_i)$ is a solution for an instance $I$ of {\sc EGS} if the following hold.
\begin{itemize}
    \item \textbf{(S1)} Every vertex of $G_i$ is either covered by $S$, covered by a vertex in S and a vector in $\ext I$, or covered by an element of $\extext I$.
    \item \textbf{(S2)} For each $\vec{r} \in \int I$, there exists $s \in S$ such that $\vec{d_{X_i}}(s)= \vec{r}$. 
    \item \textbf{(S3)} For each $(\{\vec{r_1},\vec{r_2}\},d) \in \intint I$, there exist two distinct vertices $s_1,s_2 \in S$ such that $\vec{d_{X_i}}(s_1)= \vec{r_1}$, $\vec{d_{X_i}}(s_2)= \vec{r_2}$, and $d(s_1,s_2) = d$.
    \item \textbf{(S4)} $S \cap X_i =S_I$. 
\end{itemize}
\end{definition}

In the remainder of this section, for brevity, we will refer to an instance of the {\sc EGS} problem only as an instance.

\begin{definition}
\label{def_dim-GS}
Let $I$ be an instance. We denote by $\dim (I)$ the minimum size of a set $S \subseteq V(G_i)$ which is a solution for $I$. If no such set exists, then we set $\dim (I)= + \infty$. We refer to this value as the \emph{extended geodetic set number} of $I$.
\end{definition}

In the following, we fix a graph $G$ and a nice tree decomposition $(T,\{X_i:i \in V(T)\})$ of $G$. Given a node $i$ of $T$ and an instance $I$ for $i$, we show how to compute $\dim(I)$. The proof is divided according to the type of the node $i$.\\

\noindent
{\bf Leaf node.} Computing $\dim(I)$ when $I$ is an instance for a leaf node can be done with the following lemma.

\begin{lemma}
\label{lem:leaf_node-GS}
Let $I$ be an instance for a leaf node $i$ and let $v$ be the only vertex in $X_i$. Then,
\begin{center}
$\dim (I) =
\begin{cases}
0 & \text{if } S_I = \emptyset, \int I = \emptyset, \text{ and } \intint I = \emptyset \\
1 & \text{if } S_I = \{v\}, \int I \subseteq \{(0)\}, \text{ and } \intint I = \emptyset\\
+\infty & \text{otherwise}
\end{cases}
$
\end{center}
\end{lemma} 

\begin{proof}
Suppose first that $S_I = \emptyset$. Then, the empty set is the only possible solution for $I$; and the empty set is a solution for $I$ only if $\int I = \emptyset$ and $\intint I = \emptyset$. Suppose next that $S_I = \{v\}$. Then, the set $S = \{v\}$ is the only possible solution for $I$; and this set is a solution for $I$ only if $\int I = \emptyset$ or $\int I$ contains only the vector $\vec{d_{X_i}}(v)= (0)$, and $\intint I = \emptyset$.
\end{proof}

In the remainder of this section, we handle the three other types of nodes. For each type of node, we proceed as follows: we first define a notion of compatibility on the instances for the child/children of a node $i$ and show how to compute the extended geodetic set number of an instance $I$ for $i$ from the extended geodetic set number of instances for the child/children of $i$ compatible with $I$.\\

\noindent
{\bf Join node.} Let $I$ be an instance for a join node $i$, and let $i_1$ and $i_2$ be the two children of $i$. In the following, we let $X_i = \{v_1, \ldots, v_{k}\} = X_{i_1} = X_{i_2}$. 

\begin{definition}
\label{def:compatible_pair-GS}
A pair of instances $(I_1, I_2)$ for $(i_1,i_2)$ is \emph{compatible} with $I$ if the following hold.
\begin{itemize}
    \item \textbf{(J1)} $S_{I_1}=S_{I_2}=S_I$.
    \item \textbf{(J2)} $\int{I} \subseteq \int{I_1} \cup \int {I_2}$.
    \item \textbf{(J3)} $\ext {I_1} \subseteq \ext{I} \cup \int{I_2} $ and $\ext {I_2} \subseteq \ext{I} \cup \int{I_1}$.
    \item \textbf{(J4)} Let 
    $D_1 = \{(\{\vec{r_1},\vec{r_2}\},\min_{1 \leq j \leq k} (\vec{r_1} + \vec{r_2})_j) \mid\vec{r_1} \in \ext I, \vec{r_2} \in \int{I_2}\}.$
    Then, $\extext{I_1} \subseteq \extext I \cup \intint{I_2} \cup D_1$. 
    
    Symmetrically, let 
        $D_2 = \{(\{\vec{r_1},\vec{r_2}\},\min_{1 \leq j \leq k} (\vec{r_1} + \vec{r_2})_j) \mid\vec{r_1} \in \ext I, \vec{r_2} \in \int{I_1}\}.$
        Then, $\extext{I_2} \subseteq \extext I \cup \intint{I_1} \cup D_2$. 
     \item \textbf{(J5)} Let 
     $F = \{(\{\vec{r_1},\vec{r_2}\},\min_{1 \leq j \leq k} (\vec{r_1} + \vec{r_2})) \mid \vec{r_1} \in \int{I_1}, \vec{r_2} \in \int{I_2}\}.$
     Then, $\intint I \subseteq \intint{I_1} \cup \intint{I_2} \cup F$.
\end{itemize}
\end{definition}

Let $\mathcal{F}_J(I)$ be the set of pairs of instances compatible with~$I$. We aim to prove the following.

\begin{lemma}
\label{lem:join_node-GS}
Let $I$ be an instance for a join node $i$. Then, 
\[ 
\dim(I) = \min_{(I_1,I_2) \in \mathcal{F}_J(I)} (\dim(I_1) + \dim(I_2) - |S_I|). 
\]
\end{lemma}

To prove \Cref{lem:join_node-GS}, we prove the following two lemmas.

\begin{lemma}
Let $(I_1,I_2)$ be a pair of instances for $(i_1,i_2)$ compatible with $I$ such that $\dim(I_1)$ and $\dim(I_2)$ have finite values. Let $S_1$ be a minimum-size solution for $I_1$ and $S_2$ a minimum-size solution for $I_2$. Then, $S = S_1 \cup S_2$ is a solution for $I$. In particular, 
\[ 
\dim(I) \leq \min_{(I_1,I_2) \in \mathcal{F}_J(I)} (\dim(I_1) + \dim(I_2) - |S_I|). 
\]
\end{lemma}

\begin{proof}
Let us show that every condition of \Cref{def_instance-GS} is satisfied.\\

\noindent
{\bf (S1)} Let $x$ be a vertex of $G_i$ and assume, without loss of generality, that $x \in V(G_{i_1})$ (the case where $x \in V(G_{i_2})$ is symmetric). Then, since $S_1$ is a solution for $I_1$, either $x$ is covered by $S_1$, in which case we are done; or (1) $x$ is covered by a vertex $s \in S_1$ and a vector $\vec r \in \ext{I_1}$; or (2) $x$ is covered by an element $(\{\vec{r_1},\vec{r_2}\},d) \in \extext{I_1}$. 

Suppose first that (1) holds. Then, by compatibility, either $\vec r \in \ext I$, in which case we are done, or $\vec r \in \int{I_2}$. In the latter case, since $S_2$ is a solution for $I_2$, there exists $s_2 \in S_2$ such that $\vec r = \vec{d_{X_{i_2}}}(s_2) = \vec{d_{X_i}}(s_2)$; but then, by \Cref{lem:covering}(2), $x$ is covered by $s_1,s_2 \in S$. 

Suppose next that (2) holds. Then, by compatibility, either $(\{\vec{r_1},\vec{r_2}\},d) \in \extext I$, in which case we are done; or (i) $(\{\vec{r_1},\vec{r_2}\},d) \in \intint{I_2}$; or (ii) $(\{\vec{r_1},\vec{r_2}\},d) \in D_1$. Now, if (i) holds, then, since $S_2$ is a solution for $I_2$, there exist $s_1,s_2 \in S_2$ such that $\vec{r_1} = \vec{d_{X_{i_2}}}(s_1) = \vec{d_{X_i}}(s_1)$, $\vec{r_2} = \vec{d_{X_{i_2}}}(s_2) = \vec{d_{X_i}}(s_2)$, and $d = d(s_1,s_2)$; but then, by \Cref{lem:covering}(1), $x$ is covered by $s_1,s_2 \in S$. Thus, suppose that (ii) holds. Then, since $S_2$ is a solution for $I_2$, there exists $s_2 \in S_2$ such that, say, $\vec{r_2} = \vec{d_{X_{i_2}}}(s_2) = \vec{d_{X_i}}(s_2)$; but then, by \Cref{lem:covering}(3), $x$ is covered by $s_2 \in S$ and $\vec{r_1} \in \ext I$.\\

\noindent
{\bf (S2)} Consider a vector $\vec r \in \int I$. Then, by compatibility, $\vec r \in \int{I_1} \cup \int{I_2}$, say $\vec r \in \int{I_1}$ without loss of generality. Now, $S_1$ is a solution for $I_1$, and so, there exists $s_1 \in S_1 \subseteq S$ such that $\vec r = \vec{d_{X_{i_1}}}(s_1) = \vec{d_{X_i}}(s_1)$.\\

\noindent
{\bf (S3)} Consider an element $(\{\vec{r_1},\vec{r_2}\},d) \in \intint I$. Then, by compatibility, $(\{\vec{r_1},\vec{r_2}\},d) \in \intint{I_1} \cup \intint{I_2} \cup F$. Now, if $(\{\vec{r_1},\vec{r_2}\},d) \in \intint{I_1}$, then, since $S_1$ is a solution for $I_1$, there exist $s_1,s_2 \in S_1 \subseteq S$ such that $\vec{r_1} = \vec{d_{X_{i_1}}}(s_1) = \vec{d_{X_i}}(s_1)$, $\vec{r_2} = \vec{d_{X_{i_1}}}(s_2) = \vec{d_{X_i}}(s_2)$, and $d(s_1,s_2) = d$; and we conclude symmetrically if $(\{\vec{r_1},\vec{r_2}\},d) \in \intint{I_2}$. Thus, suppose that $(\{\vec{r_1},\vec{r_2}\},d) \in F$. Then, since $S_1$ and $S_2$ are solutions for $I_1$  and $I_2$, respectively, there exist $s_1 \in S_1$ and $s_2 \in S_2$ such that $\vec{r_1} = \vec{d_{X_{i_1}}}(s_1) = \vec{d_{X_i}}(s_1)$ and $\vec{r_2} = \vec{d_{X_{i_1}}}(s_2) = \vec{d_{X_i}}(s_2)$; but then, since $X_i$ separates $s_1$ and $s_2$, $\min_{1 \leq j \leq k} (\vec{r_1} + \vec{r_2})_j = d(s_1,s_2)$.\\

\noindent
{\bf (S4)} By compatibility, $S_{I_1} = S_{I_2} = S_I$, and thus, $S \cap X_i = S \cap X_{i_1} = S_{I_1} = S_I$.\\  

It now follows from the above that $\dim(I) \leq |S| = |S_1| + |S_2| - |S_I| = \dim(I_1)+\dim(I_2)-|S_I|$; and since this holds true for any $(I_1,I_2) \in \mathcal{F}_J(I)$, the lemma follows.
\end{proof}

\begin{lemma}
Let $I$ be an instance for a join node $i$. Then, 
\[ 
\dim(I) \geq \min_{(I_1,I_2) \in \mathcal{F}_J(I)} (\dim(I_1) + \dim(I_2) - |S_I|). 
\]
\end{lemma}

\begin{proof}
If $\dim(I) = + \infty$, then the inequality readily holds. Thus, assume that $\dim(I) < + \infty$ and let $S$ be a minimum-size solution for $I$. For $j \in \{1,2\}$, let $S_j = S \cap V(G_{i_j})$. Now, let $I_1$ and $I_2$ be the two instances for $i_1$ and $i_2$, respectively, defined as follows. 
\begin{itemize}
\item$S_{I_1} = S_{I_2} = S_I$.
\item $\int{I_1} = \vec{d_{X_i}}(S_1)$ and $\int{I_2} = \vec{d_{X_i}}(S_2)$.
\item $\ext{I_1} = \ext{I} \cup \int{I_2}$ and $\ext{I_2} = \ext{I} \cup \int{I_1}$. 
\item $\intint{I_1} = \{(\{\vec{d_{X_{i_1}}}(s_1),\vec{d_{X_{i_1}}}(s_2)\},d(s_1,s_2)) \mid s_1,s_2 \in S_1\}$ and $\intint{I_2} = \{ \allowbreak (\{\vec{d_{X_{i_1}}}(s_1),\allowbreak \vec{d_{X_{i_1}}}(s_2)\},d(s_1,s_2)) \mid s_1,s_2 \in S_2\}$.
\item $\extext{I_1} = \extext I \cup \intint{I_2} \cup D_1$ and $\extext{I_2} = \extext I \cup \intint{I_1}\allowbreak \cup D_2$ (see \Cref{def:compatible_pair-GS} for the definitions of $D_1$ and $D_2$).
\end{itemize}
Let us show that the pair of instances $(I_1,I_2)$ is compatible with $I$ and that, for $j \in \{1,2\}$, $S_j$ is a solution for $I_j$.

\begin{claim}
\label{clm:compatible-GS}
The constructed pair of instances $(I_1,I_2)$ for $(i_1,i_2)$ is compatible with $I$.
\end{claim}

\begin{claimproof}
It is clear that conditions {\bf (J1)} through {\bf (J4)} of \Cref{def:compatible_pair-GS} hold; let us show that condition {\bf (J5)} holds as well.\\ 

\noindent
{\bf (J5)} Consider an element $(\{\vec r, \vec t\}, d) \in \intint I$. Since $S$ is a solution for $I$, there exist $x,y \in S$ such that $\vec{d_{X_i}}(x) = \vec r$, $\vec{d_{X_i}}(y) = \vec t$, and $d(x,y) = d$. Now, if $x,y \in S_1$, then, by construction, $(\{\vec{d_{X_{i_1}}}(x),\vec{d_{X_{i_1}}}(y)\},d(x,y)) \in \intint{I_1}$; and we conclude symmetrically if $x,y \in S_2$. Thus, assume, without loss of generality, that $x \in S_1$ and $y \in S_2$. Then, by construction, $\vec{d_{X_{i_1}}}(x) \in \int{I_1}$ and $\vec{d_{X_{i_2}}}(y) \in \int{I_2}$; and since $X_i$ separates $x$ and $y$, $d = d(x,y) = \min_{1 \leq j \leq k} (\vec{d_{X_i}}(x) + \vec{d_{X_i}}(y))_j$, that is, $(\{\vec r, \vec t\}, d) \in F$.
\end{claimproof}

\begin{claim}
\label{clm:compatible_sol-GS}
For every $j \in \{1,2\}$, $S_j$ is a solution for $I_j$.
\end{claim}

\begin{claimproof}
We only prove that $S_1$ is a solution for $I_1$ as the other case is symmetric. To this end, let us show that every condition of \Cref{def_instance-GS} is satisfied.\\

\noindent
{\bf (S1)} Consider a vertex $x$ of $G_{i_1}$. Then, since $V(G_{i_1}) \subseteq V(G_i)$ and $S$ is a solution for $I$, either (1) $x$ is covered by two vertices $s_1,s_2 \in S$; or (2) $x$ is covered by a vertex $s \in S$ and a vector $ \vec r \in \ext I$; or (3) $x$ is covered by an element of $\extext I$. Since, by construction, $\extext I \subseteq \extext{I_1}$, let us assume that (3) does not hold (we are done otherwise).

Suppose first that (1) holds and assume that at least one of $s_1$ and $s_2$ does not belong to $S_1$ (we are done otherwise), say $s_2 \notin S_1$ without loss of generality. Then, $\vec{d_{X_{i_2}}}(s_2) \in \int{I_2} \subseteq \ext{I_1}$ by construction, and thus, if $s_1 \in S_1$, then $x$ is covered by $s_1 \in S_1$ and $\vec{d_{X_{i_1}}}(s_2) \in \ext{I_1}$. Suppose therefore that $s_1,s_2 \in S_2$. Then, by construction, $(\{\vec{d_{X_i}}(s_1),\vec{d_{X_i}}(s_2)\},d(s_1,s_2)) \in \intint{I_2}$; but $\intint{I_2} \subseteq \extext{I_1}$ by construction, and thus, $x$ is covered by an element of $\extext{I_1}$. 

Second, suppose that (2) holds. Then, by construction, $\vec r \in \ext{I_1}$, and thus, if $s \in S_1$, then we are done. Now, if $s \in S_2$, then $\vec{d_{X_{i_2}}}(s) \in \int{I_2}$ which implies that $(\{\vec r, \vec{d_{X_i}}(s)\}, \min_{1 \leq j \leq k} (\vec r+\vec{d_{X_i}}(s))_j) \in D_1$; but then, $x$ is covered by an element of $\extext{I_1}$.\\

\noindent
{\bf (S2)} and {\bf (S3)} readily follow from the fact that, by construction, $\int{I_1} = \vec{d_{X_i}}(S_1)$ and $\intint{I_1} = \{(\{\vec{d_{X_{i_1}}}(s_1),\vec{d_{X_{i_1}}}(s_2)\},d(s_1,s_2)) \mid s_1,s_2 \in S_1\}$, respectively.\\

\noindent
{\bf (S4)} By construction, $S_{I_1} = S_I$, and thus, $S \cap X_{i_1} = S \cap X_i = S_I = S_{I_1}$.
\end{claimproof}

To conclude, since the sets $S_1$ and $S_2$ are solutions for $I_1$ and $I_2$, respectively, $\dim(I_1) \leq |S_1|$ and $\dim(I_2) \leq |S_2|$. Now, $|S| = |S_1| + |S_2| - |S_I|$, and so, $\dim(I) = |S| \geq \dim(I_1) + \dim(I_2) - |S_I| \geq \min_{(J_1,J_2) \in \mathcal{F}_J(I)} (\dim(J_1) + \dim(J_2) - |S_I|)$.
\end{proof}

\noindent
{\bf Introduce node.} Let $I$ be an instance for an introduce node $i$ with child $i_1$, and let $v \in V(G)$ be such that $X_i = X_{i_1} \cup \{v\}$. In the following, we let $X_i = \{v_1,\ldots,v_k\}$ where $v = v_k$. 

\begin{definition}
\label{def:compatible_introduce-GS}
An instance $I_1$ for $i_1$ is compatible with $I$ of \emph{type 1} if the following hold.
\begin{itemize}
\item {\bf (I1)} $S_I = S_{I_1}$.
\item {\bf (I2)} For each $\vec{r} \in \int{I}$, $\vec{r}_k = \min_{1 \leq j \leq k-1} (\vec r + \vec{d_{X_{i_1}}}(v))_j$ and $\vec{r^-}\in \int{I_1}$. 
\item {\bf (I3)} For each $\vec r \in \ext{I_1}$, there exists $\vec t \in \ext I$ such that $\vec{t^-} = \vec r$.
\item {\bf (I4)} For each $(\{\vec r, \vec t\},d) \in \intint I$, $\vec{r}_k = \min_{1 \leq j \leq k-1} (\vec r + \vec{d_{X_{i_1}}}(v))_j$, $\vec{t}_k = \min_{1 \leq j \leq k-1} (\vec t \allowbreak+ \vec{d_{X_{i_1}}}(v))_j$, and $(\{\vec{r^-},\vec{t^-}\},d) \in \intint{I_1}$.
\item {\bf (I5)} For each $(\{\vec{r_1},\vec{r_2}\},d) \in \extext{I_1}$, there exists $(\{\vec{t_1},\vec{t_2}\},d) \in \extext{I}$ such that $\vec{t_1^-} = \vec{r_1}$ and $\vec{t_2^-} = \vec{r_2}$.
\item {\bf (I6)} One of the following holds.
\begin{itemize}
\item $v$ is covered by an element of $\intint{I_1}$.
\item There exist $\vec r \in \int{I_1}$ and $\vec t \in \ext I$ such that $v$ is covered by $(\{\vec{r|}d, \vec t\},\min_{1 \leq j \leq k} (\vec{r|}d + \vec t)_j)$ where $d = \min_{1 \leq j \leq k-1}(\vec r + \vec{d_{X_{i_1}}}(v))_j$.
\item $v$ is covered by an element of $\extext I$.
\end{itemize} 
\end{itemize}

\noindent
An instance $I_1$ for $i_1$ is compatible with $I$ of \emph{type 2} if the following hold.
\begin{itemize}
\item {\bf (I'1)} $S_I = S_{I_1} \cup \{v\}$.
\item {\bf (I'2)} For each $\vec{r} \in \int{I} \setminus \{\vec{d_{X_i}}(v)\}$, $\vec{r}_k = \min_{1 \leq j \leq k-1} (\vec r + \vec{d_{X_{i_1}}}(v))_j$ and $\vec{r^-}\in \int{I_1}$. 
\item {\bf (I'3)} For each $\vec r \in \ext{I_1} \setminus \{\vec{d_{X_{i_1}}}(v)\}$, there exists $\vec t \in \ext I$ such that $\vec{t^-} = \vec r$.
\item {\bf (I'4)} For each $(\{\vec r, \vec t\},d) \in \intint I$, $\vec{r}_k = \min_{1 \leq j \leq k-1} (\vec r + \vec{d_{X_{i_1}}}(v))_j$ and $\vec{t}_k = \min_{1 \leq j \leq k-1} (\vec t + \vec{d_{X_{i_1}}}(v))_j$. Furthermore, one of the following holds:
\begin{itemize}
\item $(\{\vec{r^-}, \vec{t^-}\},d) \in \intint{I_1}$,
\item $\vec r = \vec{d_{X_i}}(v)$, $d = \vec{t}_k$, and $\vec{t^-} \in \int{I_1}$, or
\item $\vec t = \vec{d_{X_i}}(v)$, $d = \vec{r}_k$, and $\vec{r^-} \in \int{I_1}$.
\end{itemize}
\item {\bf (I'5)} For each $(\{\vec{r_1},\vec{r_2}\},d) \in \extext{I_1}$, one of the following holds:
\begin{itemize}
\item there exist $x,y \notin V(G_{i_1}) \cup \{v\}$ such that $\vec{d_{X_{i_1}}}(x) = \vec{r_1}$, $\vec{d_{X_{i_1}}}(y) = \vec{r_2}$, $d(x,y) = d$, and $(\{\vec{r_1|}d(x,v),\allowbreak \vec{r_2|}d(y,v)\},d) \in \extext{I}$
\item $\vec{r_1} = \vec{d_{X_{i_1}}}(v)$ and there exists $x \notin V(G_{i_1}) \cup \{v\}$ such that $\vec{d_{X_i}}(x) = \vec{r_2|}d$ and $\vec{r_2|}d \in \ext I$, or
\item $\vec{r_2} = \vec{d_{X_{i_1}}}(v)$ and there exists $x \notin V(G_{i_1}) \cup \{v\}$ such that $\vec{d_{X_i}}(x) = \vec{r_1|}d$ and $\vec{r_1|}d \in \ext I$.
\end{itemize}
\end{itemize}
\end{definition}

We denote by $\mathcal{F}_{1}(I)$ the set of instances for $i_1$ compatible with $I$ of type 1, and by $\mathcal{F}_{2}(I)$ the set of instances for $i_1$ compatible with $I$ of type 2. We aim to prove the following.

\begin{lemma}
\label{lem:introduce_node-GS}
Let $I$ be an instance for an introduce node $i$. Then, 
\[
\dim(I)=\min\;\{ \min_{I_1 \in \mathcal{F}_{1}(I)}\;\{\dim(I_1)\},\min_{I_2 \in \mathcal{F}_{2}(I)}\;\{\dim(I_2)+1\}\}.
\]
\end{lemma}

Before turning to the proof of \Cref{lem:introduce_node-GS}, we first show the following technical lemma.

\begin{lemma}
\label{lem:covering2-GS}
Let $x,s_1,s_2$ be three vertices of $G_{i_1}$ and let $\vec r, \vec{t_1}, \vec{t_2}$ be three vectors of size $k$ for which there exist $y,z_1,z_2 \notin V(G_i)$ such that $\vec{d_{X_i}}(y) = \vec r$, $\vec{d_{X_i}}(z_1) = \vec{t_1}$, and $\vec{d_{X_i}}(z_2) = \vec{t_2}$. Then, the following hold.
\begin{enumerate}
\item $x$ is covered by $s_1$ and $\vec{t_1}$ if and only if $x$ is covered by $s_1$ and $\vec{t_1^-}$.
\item $x$ is covered by $(\{\vec{t_1}, \vec{t_2}\},d(z_1,z_2))$ if and only if $x$ is covered by $(\{\vec{t_1^-},\vec{t_2^-}\},d(z_1,z_2))$.
\item $\vec r$ is covered by $s_1$ and $s_2$ if and only if $\vec{r^-}$ is covered by $s_1$ and $s_2$.
\item $x$ is covered by $(\{\vec{d_{X_{i_1}}}(v),\vec{r_1^-}\},\vec{r_1}_k)$ if and only if $x$ is covered by $v$ and $\vec{r_1}$.
\end{enumerate}
\end{lemma}

\begin{proof}
To prove item (1), it suffices to note that since $X_i$ separates $x$ from $z_1$ (or $x \in X_i$), $d(z_1,x) = \min_{1 \leq j \leq k} (\vec{t_1} + \vec{d_{X_i}}(x))_j$; and for the same reason, $d(z_1,s_1) = \min_{1 \leq j \leq k} (\vec{t_1} + \vec{d_{X_i}}(s_1))_j$. But, this is also true of $X_{i_1}$, and thus, $d(z_1,x) = \min_{1 \leq j \leq k-1} (\vec{t_1} + \vec{d_{X_{i_1}}}(x))_j$ and $d(z_1,s_1) = \min_{1 \leq j \leq k-1} (\vec{t_1} + \vec{d_{X_{i_1}}}(s_1))_j$. Items (2), (3), and (4) follow from similar arguments.
\end{proof}

To prove \Cref{lem:introduce_node-GS}, we prove the following four lemmas.

\begin{lemma}
\label{lem:upper_1-GS}
Let $I_1$ be an instance for $i_1$ compatible with $I$ of type 1 such that $\dim(I_1) < \infty$, and let $S$ be a minimum-size solution for $I_1$. Then, $S$ is a solution for $I$.
\end{lemma}

\begin{proof}
Let us prove that every condition of \Cref{def_instance-GS} is satisfied.\\

\noindent
{\bf (S1)} Let $x$ be a vertex of $G_i$ and assume first that $x \neq v$. Then, since $S$ is a solution for $I_1$, either $x$ is covered by $S$, in which case we are done; or (1) $x$ is covered by a vertex $s \in S$ and a vector $\vec r \in \ext{I_1}$; or (2) $x$ is covered by an element $(\{\vec{r_1}, \vec{r_2}\},d) \in \extext{I_1}$. Now, if (1) holds, then, by compatibility, there exists $\vec t \in \ext I$ such that $\vec{t^-} = \vec r$; but then, by \Cref{lem:covering2-GS}(1), $x$ is covered by $s$ and $\vec t$ (recall indeed that by the definition of $\ext I$, there exists $y \notin V(G_i)$ such that $\vec{d_{X_i}}(y) = \vec t$). Suppose next that (2) holds. Then, by compatibility, there exists $(\{\vec{t_1},\vec{t_2}\},d) \in \extext I$ such that $\vec{t_1^-} = \vec{r_1}$ and $\vec{t_2^-} = \vec{r_2}$; but then, by \Cref{lem:covering2-GS}(2), $x$ is covered by $(\{\vec{t_1},\vec{t_2}\},d)$ (recall indeed that by the definition of $\extext I$, there exist $y,z \notin V(G_i)$ such that $\vec{d_{X_i}}(y) = \vec{t_1}$, $\vec{d_{X_i}}(z) = \vec{t_2}$, and $d=d(y,z)$). 

Assume now that $x = v$. Then, by compatibility, either $v$ is covered by an element of $\extext I$, in which case we are done; or (1) $v$ is covered by an element $(\{\vec{r_1}, \vec{r_2}\},d) \in \intint{I_1}$; or (2) there exist $\vec r \in \int{I_1}$ and $\vec t \in \ext I$ such that $v$ is covered by $(\{\vec{r|}d, \vec t\},\min_{1 \leq j \leq k} (\vec{r|}d + \vec t)_j)$ where $d = \min_{1 \leq j \leq k-1}(\vec r + \vec{d_{X_{i_1}}}(v))_j$. Now, if (1) holds, then, since $S$ is a solution for $I_1$, there exist $s_1,s_2 \in S$ such that $\vec{d_{X_{i_1}}}(s_1) = \vec{r_1}$, $\vec{d_{X_{i_1}}}(s_2) = \vec{r_2}$, and $d(s_1,s_2) = d$; but then, by \Cref{lem:covering}(1), $v$ is covered by $s_1$ and $s_2$. Now, if (2) holds, then, since $S$ is a solution for $I_1$, there exists $s \in S$ such that $\vec{d_{X_{i_1}}}(s) = \vec r $; but then, since $X_{i_1}$ separates $s$ from $v$ (or $s \in X_{i_1}$), $d = d(s,v)$, and thus, by \Cref{lem:covering}(3), $v$ is covered by $s$ and $\vec t$. \\

\noindent
{\bf (S2)} Consider a vector $\vec r \in \int I$. Then, by compatibility, $\vec{r}_k = \min_{1 \leq j \leq k-1} (\vec r + \vec{d_{X_{i_1}}}(v))_j$ and $\vec{r^-}\in \int{I_1}$. Now, $S$ is a solution for $I_1$, and so, there exists $s \in S$ such that $\vec{d_{X_{i_1}}}(s) = \vec{r^-}$; but then, $\vec{d_{X_i}}(s) = \vec r$ as $X_{i_1}$ separates $s$ from $v$ (or $s \in X_{i_1}$).\\

\noindent
{\bf (S3)} Consider an element $(\{\vec r,\vec t\},d) \in \intint I$. Then, by compatibility, $\vec{r}_k = \min_{1 \leq j \leq k-1} \allowbreak(\vec r \allowbreak+ \vec{d_{X_{i_1}}}(v))_j$ and $\vec{t}_k = \min_{1 \leq j \leq k-1} (\vec t + \vec{d_{X_{i_1}}}(v))_j$. Furthermore, either (1) $(\{\vec{r^-}, \vec{t^-}\},d) \in \intint{I_1}$; or (2) $\vec r = \vec{d_{X_i}}(v)$, $d = \vec{t}_k$, and $\vec{t^-} \in \int{I_1}$; or (3) $\vec t = \vec{d_{X_i}}(v)$, $d = \vec{r}_k$, and $\vec{r^-} \in \int{I_1}$. Now, if (1) holds, then, since $S$ is a solution for $I_1$, there exist $s_1,s_2 \in S$ such that $\vec{d_{X_{i_1}}}(s_1) = \vec{r^-}$, $\vec{d_{X_{i_1}}}(s_2) = \vec{t^-}$, and $d = d(s_1,s_2)$; but then, $\vec{d_{X_i}}(s_1) = \vec{r}$ and $\vec{d_{X_i}}(s_2) = \vec{t}$ as $X_{i_1}$ separates $v$ from both $s_1$ and $s_2$ (or $s_1$ or $s_2$ belongs to $X_{i_1}$). Suppose next that (2) holds. Then, since $S$ is a solution for $I_1$, there exists $s \in S$ such that $\vec{d_{X_{i_1}}}(s) = \vec{t^-}$; but then, $\vec{d_{X_i}}(s) = \vec{t}$ and $d = d(s,v)$ as $X_{i_1}$ separates $s$ from $v$. Case (3) is handled symmetrically.\\

\noindent
{\bf (S4)} By compatibility, $S_{I_1} = S_I$, and thus, $S \cap X_{i_1} = S \cap (X_i \setminus \{v\}) = S_I = S_{I_1}$.
\end{proof}

\begin{lemma}
\label{lem:upper_2-GS}
Let $I_1$ be an instance for $i_1$ compatible with $I$ of type 2 such that $\dim(I_1) < \infty$, and let $S$ be a minimum-size solution for $I_1$. Then, $S \cup \{v\}$ is a solution for $I$.
\end{lemma}

\begin{proof}
Let us prove that the conditions of \Cref{def_instance-GS} are satisfied. In the following, we let $S' = S \cup \{v\}$.\\

\noindent
{\bf (S1)} Let $x$ be a vertex of $G_i$. Since $v \in S'$, we may safely assume that $x \neq v$, that is, $x \in V(G_{i_1})$. Now, $S$ is a solution for $I_1$, and thus, either $x$ is covered by $S$, in which case we are done; or (1) $x$ is covered by a vertex $s \in S$ and a vector of $\vec r \in \ext{I_1}$; or (2) $x$ is covered by an element $(\{\vec{r_1},\vec{r_2}\},d) \in \extext{I_1}$. 

Suppose first that (1) holds. If $\vec r = \vec{d_{X_{i_1}}}(v)$, then, by \Cref{lem:covering}(2), $x$ is covered by $s,v \in S'$. Otherwise, by compatibility, there exists $\vec t \in \ext I$ such that $\vec{t^-} = \vec r$; but then, by \Cref{lem:covering2-GS}(1), $x$ is covered by $s$ and $\vec t$. 

Suppose next that (2) holds. If $\vec{r_1} = \vec{d_{X_{i_1}}}(v)$, then, by compatibility, there exists $x \notin V(G_i)$ such that $\vec{d_{X_i}}(x) = \vec{r_2|}d$ and $\vec{r_2|}d \in \ext I$; but then, by \Cref{lem:covering2-GS}(4), $x$ is covered by $v$ and $\vec{r_2|}d$. We conclude symmetrically if $\vec{r_2} = \vec{d_{X_{i_1}}}(v)$. Thus, by compatibility, we may assume that there exists $(\{\vec{t_1},\vec{t_2}\},d) \in \extext I$ such that $\vec{t_1^-} =\vec{r_1}$ and $\vec{t_2^-} = \vec{r_2}$; but then, by \Cref{lem:covering2-GS}(2), $(\{\vec{t_1},\vec{t_2}\},d)$ covers $x$.\\

\noindent
{\bf (S2)} Consider a vector $\vec r \in \int I$ and assume that $\vec r \neq \vec{d_{X_i}}(v)$ (as $v \in S'$, we are done otherwise). Then, by compatibility, $\vec{r}_k = \min_{1 \leq j \leq k-1} (\vec r + \vec{d_{X_{i_1}}}(v))_j$ and $\vec{r^-}\in \int{I_1}$. Now, $S$ is a solution for $I_1$, and thus, there exists $s \in S$ such that $\vec{d_{X_{i_1}}}(s) = \vec{r^-}$; but then, $\vec{d_{X_i}}(s) = \vec r$ as $X_{i_1}$ separates $s$ from $v$ (or $s \in X_{i_1}$).\\

\noindent
{\bf (S3)} Consider an element $(\{\vec r,\vec t\},d) \in \intint I$. Then, by compatibility, $\vec{r}_k = \min_{1 \leq j \leq k-1} \allowbreak (\vec r + \vec{d_{X_{i_1}}}(v))_j$ and $\vec{t}_k = \min_{1 \leq j \leq k-1} (\vec t + \vec{d_{X_{i_1}}}(v))_j$. Furthermore, either (1) $(\{\vec{r^-}, \vec{t^-}\},d) \in \intint{I_1}$; or (2) $\vec r = \vec{d_{X_i}}(v)$, $d = \vec{t}_k$, and $\vec{t^-} \in \int{I_1}$; or (3) $\vec t = \vec{d_{X_i}}(v)$, $d = \vec{r}_k$, and $\vec{r^-} \in \int{I_1}$. Now, if (1) holds, then, since $S$ is a solution for $I_1$, there exist $s_1,s_2 \in S$ such that $\vec{d_{X_{i_1}}}(s_1) = \vec{r^-}$, $\vec{d_{X_{i_1}}}(s_2) = \vec{t^-}$, and $d(s_1,s_2) = d$; but then, $\vec{d_{X_i}}(s_1) = \vec r$ and $\vec{d_{X_i}}(s_2) = \vec t$ as $X_{i_1}$ separates $v$ from both $s_1$ and $s_2$ (or $s_1$ or $s_2$ belongs to $X_{i_1}$). Similarly, if (2) holds, then, since $S$ is a solution for $I_1$, there exists $s \in S$ such that $\vec{d_{X_{i_1}}}(s) = \vec{t^-}$; but then, $\vec{d_{X_i}}(s) = \vec t$ as $X_{i_1}$ separates $s$ from $v$ (or $s \in X_{i_1}$). Case 3 is handled symmetrically.\\

\noindent
{\bf (S4)} By compatibility, $S_I = S_{I_1} \cup \{v\}$, and thus, $S \cap X_i = S \cap (X_{i_1} \cup \{v\}) = S_{I_1} \cup \{v\} = S_I$.
\end{proof}

As a consequence of Lemmas \ref{lem:upper_1-GS} and \ref{lem:upper_2-GS}, the following holds.

\begin{lemma}
Let $I$ be an instance for an introduce node $i$. Then, 
\[
\dim(I) \leq \min\;\{ \min_{I_1 \in \mathcal{F}_{1}(I)}\;\{\dim(I_1)\},\min_{I_2 \in \mathcal{F}_{2}(I)}\;\{\dim(I_2)+1\}\}.
\]
\end{lemma}

\begin{lemma}
\label{lem:lower_1-GS}
Assume that $\dim(I) < \infty$ and let $S$ be a minimum-size solution for $I$ such that $v \notin S$. Then, there exists $I_1 \in \mathcal{F}_{1}(I)$ such that $S$ is a solution for $I_1$.
\end{lemma}

\begin{proof}
Let $I_1$ be the instance for $i_1$ defined as follows.
\begin{itemize}
\item $S_{I_1} = S_I$.
\item $\int{I_1} = \vec{d_{X_{i_1}}}(S)$.
\item $\ext{I_1} = \{\vec{r^-} \mid \vec r \in \ext I\}$. 
\item $\intint{I_1} = \{(\{\vec{d_{X_{i_1}}}(s_1),\vec{d_{X_{i_1}}}(s_2)\},d(s_1,s_2)) \mid s_1,s_2 \in S\}$.
\item $\extext{I_1} = \{(\{\vec{r^-},\vec{t^-}\},d) \mid (\{\vec r, \vec t\},d) \in \extext I\}$.
\end{itemize}
Let us show that $I_1$ is compatible with $I$ of type 1 and that $S$ is a solution for $I_1$.

\begin{claim}
The constructed instance $I_1$ is compatible with $I$ of type 1.
\end{claim}

\begin{claimproof}
It is clear that condition {\bf (I1)} of \Cref{def:compatible_introduce-GS} holds; let us show that the remaining conditions hold as well. \\

\noindent
{\bf (I2)} Consider a vector $\vec r \in \int I$. Then, since $S$ is a solution for $I$, there exists $s \in S$ such that $\vec{d_{X_i}}(s) = \vec r$; but then, $\vec{r}_k = \min_{1 \leq j \leq k-1} (\vec r + \vec{d_{X_{i_1}}}(v))_j$ as $X_{i_1}$ separates $s$ from $v$ (or $s \in X_{i_1}$), and $\vec{r^-} \in \int{I_1}$ by construction.\\

\noindent
{\bf (I3)} Readily follows from the fact that, by construction, $\ext{I_1} = \{\vec{r^-} \mid \vec r \in \ext I\}$.\\

\noindent
{\bf (I4)} Consider an element $(\{\vec r, \vec t\},d) \in \intint I$. Then, since $S$ is a solution for $I$, there exist $s_1,s_2\in S$ such that $\vec{d_{X_i}}(s_1) = \vec r$, $\vec{d_{X_i}}(s_2) = \vec t$, and $d(s_1,s_2) = d$; but then, $\vec{r}_k = \min_{1 \leq j \leq k-1} (\vec r + \vec{d_{X_{i_1}}}(v))_j$ and $\vec{t}_k = \min_{1 \leq j \leq k-1} (\vec t + \vec{d_{X_{i_1}}}(v))_j$ as $X_{i_1}$ separates $v$ from both $s_1$ and $s_2$ (or $s_1$ or $s_2$ belongs to $X_{i_1}$), and $(\{\vec{r^-}, \vec{t^-}\},d) \in \intint{I_1}$ by construction.\\

\noindent
{\bf (I5)} Readily follows from the fact that, by construction, $\extext{I_1} = \{(\{\vec{r^-},\vec{t^-}\},d) \mid (\{\vec r, \vec t\},d) \in \extext I\}$.\\

\noindent
{\bf (I6)} Since $S$ is a solution for $I$ and $v \notin S$, either $v$ is covered by an element of $\extext I$, in which case we are done; or $v$ is covered by $s_1,s_2 \in S$, in which case $(\{\vec{d_{X_{i_1}}}(s_1),\vec{d_{X_{i_1}}}(s_2)\},\allowbreak d(s_1,s_2)) \in \intint{I_1}$ covers $v$; or $v$ is covered by a vertex $s \in S$ and a vector $\vec r \in \ext I$, in which case $v$ is covered by $(\{\vec{d_{X_i}}(s),\vec r\}, \min_{1 \leq j \leq k} (\vec{d_{X_i}}(s) + \vec r)_j)$, where $\vec{d_{X_i}}(s)^- \in \int{I_1}$ by construction.
\end{claimproof}

\begin{claim}
$S$ is a solution for $I_1$.
\end{claim}

\begin{claimproof}
Let us prove that the conditions of \Cref{def_instance-GS} are satisfied.\\

\noindent
{\bf (S1)} Consider a vertex $x$ of $G_{i_1}$. Then, since $V(G_{i_1}) \subseteq V(G_i)$ and $S$ is a solution for $I$, either $x$ is covered by $S$, in which case we are done; or (1) $x$ is covered by a vertex $s \in S$ and a vector $\vec r \in \ext I$; or (2) $x$ is covered by $(\{\vec r, \vec t\},d) \in \extext I$. Now, if (1) holds, then, by \Cref{lem:covering2-GS}(1), $x$ remains covered by $s$ and $\vec{r^-}$, where $\vec{r^-} \in \ext{I_1}$ by construction; and if (2) holds, then, by \Cref{lem:covering2-GS}(2), $x$ remains covered by $(\{\vec{r^-}, \vec{t^-}\},d)$, which is an element of $\extext{I_1}$ by construction.\\

\noindent
{\bf (S2)} and {\bf (S3)} readily follow from the fact that, by construction, $\int{I_1} = \vec{d_{X_{i_1}}}(S)$ and $\intint{I_1} = \{(\{\vec{d_{X_{i_1}}}(s_1),\vec{d_{X_{i_1}}}(s_2)\},d(s_1,s_2)) \mid s_1,s_2 \in S\}$, respectively.\\

\noindent
{\bf (S4)} By construction, $S_{I_1} = S_I$, and thus, $S \cap X_{i_1} = S \cap (X_i \setminus \{v\}) = S_I = S_{I_1}$.
\end{claimproof}

The lemma now follows from the above two claims.
\end{proof}

\begin{lemma}
\label{lem:lower_2-GS}
Assume that $\dim(I) < \infty$ and let $S$ be a minimum-size solution for $I$ such that $v \in S$. Then, there exists $I_1\in \mathcal{F}_2(I)$ such that $S \setminus \{v\}$ is a solution of $I_1$.
\end{lemma}

\begin{proof}
In the following, we let $S' = S \setminus \{v\}$. Now, let $I_1$ be the instance for $i_1$ defined as follows.
\begin{itemize}
\item $S_{I_1} = S_I \setminus \{v\}$.
\item  $\int{I_1} = \vec{d_{X_{i_1}}}(S')$.
\item $\ext{I_1} = \{\vec{r^-} \mid \vec r \in \ext I\} \cup \{\vec{d_{X_{i_1}}}(v)\}$.
\item $\intint{I_1} = \{(\{\vec{d_{X_{i_1}}}(s_1),\vec{d_{X_{i_1}}}(s_2)\},d(s_1,s_2)) \mid s_1,s_2 \in S' \}$.
\item $\extext{I_1} = \{(\{\vec r,\vec t\},d) \mid \exists x,y \notin V(G_i) \text{ s.t. } \vec{d_{X_{i_1}}}(x) = \vec r, \vec{d_{X_{i_1}}}(y) = \vec t, d(x,y) = d, \text{ and } (\{\vec{r|}d(x,v),\vec{t|}d(y,v)\},d) \in \extext I\} \cup \{(\{\vec{d_{X_{i_1}}}(v),\vec r\},d) \mid \exists x \notin V(G_{i_1}) \text{ s.t. } \allowbreak \vec{d_{X_{i_1}}}(x) = \vec r, d(x,v) = d, \text{ and } \vec{r|}d \in \ext I\}$.
\end{itemize}
Let us show that $I_1$ is compatible with $I$ of type 2, and that $S'$ is a solution for $I_1$.

\begin{claim}
The constructed instance $I_1$ is compatible with $I$ of type 2.
\end{claim} 

\begin{claimproof}
It is not difficult to see that condition {\bf (I'1)} of \Cref{def:compatible_introduce-GS} holds; let us show that the remaining conditions hold as well.\\

\noindent
{\bf (I'2)} Consider an element $\vec r \in \int I \setminus \{\vec{d_{X_i}}(v)\}$. Then, since $S$ is a solution for $I$, there exists $s \in S$ such that $\vec{d_{X_i}}(s) = \vec r$; but then, $\vec{r}_k = \min_{1 \leq j \leq k-1} (\vec r + \vec{d_{X_{i_1}}}(v))_j$ as $X_{i_1}$ separates $s$ from $v$ (or $s \in X_{i_1}$), and $\vec{r^-} \in \int{I_1}$ by construction.\\

\noindent
{\bf (I'3)} Readily follows from the fact that, by construction, $\ext{I_1} = \{\vec{r^-} \mid \vec r \in \ext I\} \cup \{\vec{d_{X_{i_1}}}(v)\}$.\\ 

\noindent 
{\bf (I'4)} Consider an element $(\{\vec r, \vec t\},d) \in \intint I$. Then, since $S$ is a solution for $I$, there exist $s_1,s_2 \in S$ such that $\vec{d_{X_i}}(s_1) = \vec r$, $\vec{d_{X_i}}(s_2) = \vec t$, and $d(s_1,s_2) = d$; in particular, $d(s_1,v) = \vec{r}_k = \min_{1 \leq j \leq k-1} (\vec r + \vec{d_{X_{i_1}}}(v))_j$ and $d(s_2,v) = \vec{t}_k = \min_{1 \leq j \leq k-1} (\vec t + \vec{d_{X_{i_1}}}(v))_j$, as either $v \in \{s_1,s_2\}$, or $X_{i_1}$ separates $v$ from both $s_1$ and $s_2$ (or $s_1$ or $s_2$ belongs to $X_{i_1}$). Now, if, say, $s_1 = v$ (the case where $s_2 = v$ is symmetric), then, by construction, $\vec{t^-} \in \int{I_1}$, and thus, the second or third item of {\bf (I'4)} is satisfied. Otherwise, $s_1,s_2 \in S'$, and so, $(\{\vec{r^-},\vec{t^-}\},d) \in \intint{I_1}$ by construction.\\

\noindent
{\bf (I'5)} Readily follows from the fact that, by construction, $\extext{I_1} = \{(\{\vec r,\vec t\},d) \mid \exists x,y \notin V(G_i) \text{ s.t. } \vec{d_{X_{i_1}}}(x) = \vec r, \vec{d_{X_{i_1}}}(y) = \vec t, d(x,y) = d,$ and $(\{\vec{r|}d(x,v),\vec{t|}d(y,v)\},d) \in \extext I\} \allowbreak \cup \{(\{\vec{d_{X_{i_1}}}(v),\vec r\},d) \mid \exists x \notin V(G_{i_1}) \text{ s.t. } \vec{d_{X_{i_1}}}(x) = \vec r, d(x,v) = d, \text{ and } \vec{r|}d \in \ext I\}$.
\end{claimproof}

\begin{claim}
$S'$ is a solution for $I_1$.
\end{claim}

\begin{claimproof}
Let us prove that the conditions of \Cref{def_instance-GS} are satisfied.\\

\noindent
{\bf (S1)} Consider a vertex $x$ of $G_{i_1}$. Then, since $V(G_{i_1}) \subseteq V(G_i)$ and $S$ is a solution for $I$, either (1) $x$ is covered by two vertices $s_1,s_2 \in S$; (2) $x$ is covered by a vertex $s \in S$ and a vector $\vec r \in \ext I$; or (3) $x$ is covered by an element $(\{\vec r, \vec t\},d) \in \extext I$. 

Now, if (1) holds, then we may assume that one of $s_1$ and $s_2$ is $v$ (otherwise, $S'$ covers $x$), say $s_1 = v$ without loss of generality; but then, $\vec{d_{X_{i_1}}}(v) \in \ext{I_1}$ by construction, and thus, $x$ is covered by $s \in S'$ and $\vec{d_{X_{i_1}}}(v) \in \ext{I_1}$ (note indeed that $d(s,v) = \min_{1 \leq j \leq k-1} (\vec{d_{X_{i_1}}}(s) + \vec{d_{X_{i_1}}}(v))_j$ and $d(x,v) = \min_{1 \leq j \leq k-1} (\vec{d_{X_{i_1}}}(x) + \vec{d_{X_{i_1}}}(v))_j$ as $X_{i_1}$ separates $v$ from both $s$ and $x$). 

Suppose next that (2) holds. If $s \neq v$, then, since $\vec{r^-} \in \ext{I_1}$ by construction, $x$ remains covered by $s \in S'$ and $\vec{r^-} \in \ext{I_1}$ by \Cref{lem:covering2-GS}(1); and if $s = v$, then $x$ is covered by $(\{\vec{d_{X_{i_1}}}(v),\vec{r^-}\},\vec{r}_k)$ which is an element of $\extext{I_1}$ by construction. 

Finally, if (3) holds, then by the definition of $\extext I$, there exist $x,y \notin V(G_i)$ such that  $\vec{d_{X_i}}(x) = \vec r$, $\vec{d_{X_i}}(y) = \vec t$, and $d(x,y) = d$; but then, by \Cref{lem:covering}(2), $x$ is covered by $(\{\vec{r^-},\vec{t^-}\},d)$, which is an element of $\extext{I_1}$ by construction.\\

\noindent
{\bf (S2)} and {\bf (S3)} readily follow from the fact that, by construction, $\int{I_1} = \vec{d_{X_{i_1}}}(S')$ and $\intint{I_1} = \{(\{\vec{d_{X_{i_1}}}(s_1),\vec{d_{X_{i_1}}}(s_2)\},d(s_1,s_2)) \mid s_1,s_2 \in S' \}$, respectively.\\

\noindent
{\bf (S4)} By construction, $S_{I_1} = S_I \setminus \{v\}$, and so, $S \cap X_{i_1} = S \cap (X_i \setminus \{v\}) = S_I \setminus \{v\} = S_{I_1}$. 
\end{claimproof}

The lemma now follows from the above two claims.
\end{proof}

As a consequence of Lemmas \ref{lem:lower_1-GS} and \ref{lem:lower_2-GS}, the following holds.

\begin{lemma}
Let $I$ be an instance for an introduce node $i$. Then, 
\[
\dim(I) \geq \min\;\{ \min_{I_1 \in \mathcal{F}_{1}(I)}\;\{\dim(I_1)\},\min_{I_2 \in \mathcal{F}_{2}(I)}\;\{\dim(I_2)+1\}\}.
\]
\end{lemma}

\noindent
{\bf Forget node.} Let $I$ be an instance for a forget node $i$ with child $i_1$, and let $v \in V(G)$ be such that $X_i = X_{i_1} \setminus \{v\}$. Further, let $X_{i_1} = \{v_1, \ldots, v_k\}$, where $v = v_k$.

\begin{definition}
\label{def:compatible_forget-GS}
An instance $I_1$ for $i_1$ is compatible with $I$ if the following hold.
\begin{itemize}
\item {\bf (F1)} $S_I = S_{I_1} \setminus \{v\}$.
\item {\bf (F2)} For each $\vec r \in \int I$, there exists $\vec t \in \int{I_1}$ such that $\vec{t^-} = \vec r$.
\item {\bf (F3)} For each $\vec r \in \ext{I_1}$, $\vec{r^-} \in \ext I$.
\item {\bf (F4)} For each $(\{\vec{r_1},\vec{r_2}\},d) \in \intint I$, there exists $(\{\vec{t_1},\vec{t_2}\},d) \in \intint{I_1}$ such that $\vec{t_1^-} = \vec{r_1}$ and $\vec{t_2^-} = \vec{r_2}$.
\item {\bf (F5)} For each $(\{\vec{r_1},\vec{r_2}\},d) \in \extext{I_1}$, $(\{\vec{r_1^-},\vec{r_2^-}\},d) \in \extext I$.
\end{itemize}
\end{definition}

We denote by $\mathcal{F}_{F}(I)$ the set of instances for $i_1$ compatible with $I$. We aim to prove the following.

\begin{lemma}
\label{lem:forget_node-GS}
Let $I$ be an instance for a forget node $i$. Then,
\[
\dim(I)= \min_{I_1 \in \mathcal{F}_{F}(I)}\;\{\dim(I_1)\}.
\]
\end{lemma}

Before turning to the proof of \Cref{lem:forget_node-GS}, we first prove the following technical lemma, which is the analog of \Cref{lem:covering2-GS}.

\begin{lemma}
\label{lem:covering3}
Let $x,s_1,s_2$ be three vertices of $G_i$, and let $\vec r, \vec{t_1}, \vec{t_2}$ be three vectors of size $k$ for which there exist $y,z_1,z_2 \notin V(G_i)$ such that $\vec{d_{X_{i_1}}}(y) = \vec r$, $\vec{d_{X_{i_1}}}(z_1) = \vec{t_1}$, and $\vec{d_{X_{i_1}}}(z_2) = \vec{t_2}$. Then, the following hold.
\begin{enumerate}
\item $x$ is covered by $s_1$ and $\vec{t_1}$ if and only if $x$ is covered by $s_1$ and $\vec{t_1^-}$.
\item $x$ is covered by $(\{\vec{t_1}, \vec{t_2}\},d(y,z))$ if and only if $x$ is covered by $(\{\vec{t_1^-},\vec{t_2^-}\},d(y,z))$.
\item $\vec r$ is covered by $s_1$ and $s_2$ if and only if $\vec{r^-}$ is covered by $s_1$ and $s_2$.
\end{enumerate}
\end{lemma}

\begin{proof}
To prove item (1), it suffices to note that, since $X_{i_1}$ separates $x$ from $z_1$ (or $x \in X_{i_1}$), $d(z_1,x) = \min_{1 \leq j \leq k} (\vec{t_1} + \vec{d_{X_{i_1}}}(x))_j$; and for the same reason, $d(z_1,s_1) = \min_{1 \leq j \leq k} (\vec{t_1} + \vec{d_{X_{i_1}}}(s_1))_j$. But, this is also true of $X_i$, and thus, $d(z_1,x) = \min_{1 \leq j \leq k-1} (\vec{t_1} + \vec{d_{X_i}}(x))_j$ and $d(z_1,s_1) = \min_{1 \leq j \leq k-1} (\vec{t_1} + \vec{d_{X_i}}(s_1))_j$. Items (2) and (3) follow from similar arguments.
\end{proof}

To prove \Cref{lem:forget_node-GS}, we prove the following two lemmas.

\begin{lemma}
Let $I_1$ be an instance for $i_1$ compatible with $I$ such that $\dim(I_1) < \infty$, and let $S$ be a minimum-size solution for $I_1$. Then, $S$ is a solution for $I$. In particular,
\[
\dim(I) \leq \min_{I_1 \in \mathcal{F}_{F}(I)}\;\{\dim(I_1)\}.
\]
\end{lemma}

\begin{proof}
Let us prove that  the conditions of \Cref{def_instance-GS} are satisfied.\\

\noindent
{\bf (S1)} Consider a vertex $x$ of $G_i$. Then, since $V(G_i) = V(G_{i_1})$ and $S$ is a solution for $I_1$, either $x$ is covered by $S$, in which case we are done; or (1) $x$ is covered by a vertex $s \in S$ and a vector $\vec r \in \ext{I_1}$; or (2) $x$ is covered by  an element $(\{\vec{r_1},\vec{r_2}\},d) \in \extext{I_1}$. Now, if (1) holds, then, by compatibility, $\vec{r^-} \in \ext I$; but then, by \Cref{lem:covering3}(1), $x$ is covered by $s$ and $\vec{r^-}$. Otherwise, (2) holds, in which case, by compatibility, $(\{\vec{r_1^-},\vec{r_2^-}\},d) \in \extext I$; but then, by \Cref{lem:covering3}(2), $x$ is covered by $(\{\vec{r_1^-},\vec{r_2^-}\},d)$. \\

\noindent
{\bf (S2)} Consider a vector $\vec r \in \int I$. Then, by compatibility, there exists $\vec t \in \int{I_1}$ such that $\vec{t^-} = \vec r$; and since $S$ is a solution for $I_1$, there then exists $s \in S$ such that $\vec{d_{X_{i_1}}}(s) = \vec t$.\\

\noindent
{\bf (S3)} Consider an element $(\{\vec{r_1},\vec{r_2}\},d) \in \intint I$. Then, by compatibility, there exists $(\{\vec{t_1},\vec{t_2}\},d) \in \intint{I_1}$ such that $\vec{t_1^-} = \vec{r_1}$ and $\vec{t_2^-} = \vec{r_2}$; and since $S$ is a solution for $I_1$, there then exist $s_1,s_2 \in S$ such that $\vec{d_{X_{i_1}}}(s_1) = \vec{t_1}$, $\vec{d_{X_{i_1}}}(s_2) = \vec{t_2}$, and $d(s_1,s_2) = d$.\\

\noindent
{\bf (S4)} By compatibility, $S_I = S_{I_1} \setminus \{v\}$, and thus, $S \cap X_i = S \cap (X_{i_1} \setminus \{v\}) = S_{I_1} \setminus \{v\} = S_I$.
\end{proof}

\begin{lemma}
Assume that $\dim(I) < \infty$ and let $S$ be a minimum-size solution for $I$. Then, there exists $I_1 \in \mathcal{F}_{F}(I)$ such that $S$ is a solution for $I_1$. In particular,
\[
\dim(I) \geq \min_{I_1 \in \mathcal{F}_{F}(I)}\;\{\dim(I_1)\}.
\]
\end{lemma}

\begin{proof}
Let $I_1$ be the instance for $i_1$ defined as follows.
\begin{itemize}
\item $S_{I_1} = S \cap X_{i_1}$. 
\item $\int{I_1} = \vec{d_{X_{i_1}}}(S)$. 
\item $\ext{I_1} = \{\vec r \mid \exists x \notin V(G_{i_1}) \text{ s.t } \vec{d_{X_{i_1}}}(x) = \vec r \text{ and } \vec{r^-} \in \ext I\}$.
\item $\intint{I_1} = \{(\{\vec{d_{X_{i_1}}}(s_1),\vec{d_{X_{i_1}}}(s_2)\},d(s_1,s_2) \mid s_1,s_2 \in S\}$.
\item $\extext{I_1} = \{(\{\vec r, \vec t\},d) \mid \exists x,y \notin V(G_{i_1}) \text{ s.t } \vec{d_{X_{i_1}}}(x) = \vec r, \vec{d_{X_{i_1}}}(y) = \vec t, d(x,y) = d, \text{ and } (\{\vec{r^-},\vec{t^-}\},d) \in \extext I\}$.
\end{itemize}
Let us show that $I_1$ is compatible with $I$, and that $S$ is a solution for $I_1$.

\begin{claim}
The constructed instance $I_1$ is compatible with $I$.
\end{claim}

\begin{claimproof}
It is clear that condition {\bf (F1)} of \Cref{def:compatible_forget-GS} holds; let us show that the remaining conditions hold as well.\\

\noindent
{\bf (F2)} Consider a vector $\vec r \in \int I$. Then, since $S$ is a solution for $I$, there exists $s \in S$ such that $\vec{d_{X_i}}(s) = \vec r$; but then, $\vec{r|}d(s,v) \in \int{I_1}$ by construction.\\

\noindent
{\bf (F3)} readily follows from the fact that $\ext{I_1} = \{\vec r \mid \exists x \notin V(G_{i_1}) \text{ s.t. } \vec{d_{X_{i_1}}}(x) = \vec r \text{ and } \vec{r^-} \in \ext I\}$.\\

\noindent
{\bf (F4)} Consider an element $(\{\vec{r_1},\vec{r_2}\},d) \in \intint I$. Then, since $S$ is a solution for $I$, there exist $s_1,s_2 \in S$ such that $\vec{d_{X_i}}(s_1) = \vec{r_1}$, $\vec{d_{X_i}}(s_2) = \vec{r_2}$, and $d(s_1,s_2) = d$; but then, $(\{\vec{r_1|}d(s_1,v),\vec{r_2|}d(s_2,v)\},d) \in \intint{I_1}$ by construction.\\

\noindent
{\bf (F5)} readily follows from the fact that $\extext{I_1} = \{(\{\vec r, \vec t\},d) \mid \exists x,y \notin V(G_{i_1}) \text{ s.t. } \vec{d_{X_{i_1}}}(x) = \vec r, \vec{d_{X_{i_1}}}(y) = \vec t, d(x,y) = d, \text{ and } (\{\vec{r^-},\vec{t^-}\},d) \in \extext I\}$.
\end{claimproof}

\begin{claim}
$S$ is a solution for $I_1$.
\end{claim}

\begin{claimproof}
Let us show that every condition of \Cref{def_instance-GS} holds.\\

\noindent
{\bf (S1)} Consider a vertex $x$ of $G_{i_1}$. Then, since $V(G_i) = V(G_{i_1})$ and $S$ is a solution for $I$, either $x$ is covered by $S$, in which case we are done; or (1) $x$ is covered by a vertex $s \in S$ and a vector $\vec r \in \ext I$; (2) or $x$ is covered by an element $(\{\vec r, \vec t\}, d) \in \extext I$. Now, suppose that (1) holds and let $y \notin V(G_i)$ be a vertex with distance vector $\vec r$ to $X_i$ (recall that such a vertex exists by the definition of $\ext I$). Then, by \Cref{lem:covering3}(1), $x$ is covered by $s$ and $\vec{r|}d(y,v)$; but, by construction, $\vec{r|}d(y,v) \in \ext{I_1}$. Suppose next that (2) holds, and let $y,z \notin V(G_i)$ be such that $\vec{d_{X_i}}(y) = \vec r$, $\vec{d_{X_i}}(z) = \vec t$, and $d(y,z) = d$ (recall that such vertices exist by the definition of $\extext I$). Then, by \Cref{lem:covering3}(2), $x$ is covered by $(\{\vec{r|}d(y,v),\vec{t|}d(z,v)\},d(y,z))$, which is an element of $\extext{I_1}$ by construction.\\

\noindent
{\bf (S2)}, {\bf (S3)}, and {\bf (S4)} readily follow from the fact that, by construction, $\int{I_1} = \vec{d_{X_{i_1}}}(S)$, $\intint{I_1} = \{(\{\vec{d_{X_{i_1}}}(s_1),\vec{d_{X_{i_1}}}(s_2)\},d(s_1,s_2) \mid s_1,s_2 \in S\}$, and $S_{I_1} = S \cap X_{i_1}$, respectively.
\end{claimproof}

The lemma now follows from the above two claims.
\end{proof}


To complete the proof of \Cref{thm:algo-diam-tw-GD}, let us now explain how the algorithm proceeds. Given a nice tree decomposition $(T,\mathcal{X})$ of a graph $G$ rooted at node $r \in V(T)$, the algorithm computes the extended geodetic set number for all possible instances in a bottom-up traversal of $T$. It computes the values for leaf nodes using \Cref{lem:leaf_node-GS}, for join nodes using \Cref{lem:join_node-GS}, for introduce nodes using \Cref{lem:introduce_node-GS}, and for forget nodes using \Cref{lem:forget_node-GS}. The correctness of this algorithm follows from these lemmas and the following (recall that $\gs(G)$ is the smallest size of a geodetic set for $G$).

\begin{lemma}
Let $G$ be a graph and let $(T,\{X_i:i\in V(T)\})$ be a nice tree decomposition of $G$ rooted at node $r \in V(T)$. Then,
\[
\gs(G) = \min_{S_r \subseteq X_r} \dim(X_r,S_r,\emptyset,\emptyset,\emptyset,\emptyset).
\]
\end{lemma}

\begin{proof}
Let $S$ be a minimum-size geodetic set of $G$. Then, by \Cref{def_instance-GS}, $S$ is a solution for the {\sc EGS} instance $(X_r,S\cap X_r,\emptyset,\emptyset,\emptyset,\emptyset)$, and so,
\[
\min_{S_r \subseteq X_r} \dim(X_r,S_r,\emptyset,\emptyset,\emptyset,\emptyset) \leq \dim(X_r,S\cap X_r,\emptyset,\emptyset,\emptyset,\emptyset) \leq \gs(G).
\]
Conversely, let $S' \subseteq X_r$ be a set attaining the minimum above, and let $S$ be a minimum-size solution for the {\sc EGS} instance $(X_r,S',\emptyset,\emptyset,\emptyset,\emptyset)$. Then, by \Cref{def_instance-GS}, every vertex of $G_r = G$ is covered by $S$, and so,
\[
\gs(G) \leq \dim(X_r,S',\emptyset,\emptyset,\emptyset,\emptyset) = \min_{S_r \subseteq X_r} \dim(X_r,S_r,\emptyset,\emptyset,\emptyset,\emptyset),
\]
which concludes the proof.
\end{proof}

Now, let $\alpha(k) = 2^{k} \cdot 2^{\diam(G)^{k}} \cdot 2^{\diam(G)^{k}} \cdot 2^{\diam(G)^{2k+1}} \cdot 2^{\diam(G)^{2k+1}}$. To get the announced complexity, observe first that, at each node $i \in V(T)$, there are at most $\alpha(|X_i|)$ possible instances to consider, where $|X_i| = \OO(\tw(G))$; and since $T$ has $\OO(\tw(G) \cdot n)$ nodes, there are in total $\OO(\alpha(\tw(G)) \cdot \tw(G) \cdot n)$ possible instances. The running time of the algorithm then follows from these facts and the following lemma (note that to avoid repeated computations, we can first compute the distance between every pair of vertices of $G$ in $n^{\OO(1)}$ time, as well as all possible distance vectors to a bag from the possible distance vectors to its child/children).

\begin{lemma}
Let $I$ be an {\sc EGS} instance for a node $i \in V(T)$, and assume that, for every child $i_1$ of $i$ and every {\sc EGS} instance $I_1$ for $i_1$ compatible with $I$, $\dim(I_1)$ is known. Then, $\dim(I)$ can be computed in time $\alpha(\OO(|X_i|)) \cdot n^{\OO(1)}$. 
\end{lemma} 

\begin{proof}
If $i$ is a leaf node, then $\dim(I)$ can be computed in constant time by \Cref{lem:leaf_node-GS}. Otherwise, let us prove that one can compute all compatible instances in the child nodes in the announced time (recall that $i$ has at most two child nodes). Given a 6-tuple $I = (X_i,S_I,\int I,\ext I, \intint I,\extext I)$, checking whether it is an {\sc EGS} instance can be done in $\OO(|I|) \cdot n^{\OO(1)}$ time; and the number of such 6-tuples is bounded by $\alpha(|X_i|)$. It is also not difficult to see that checking for compatibility can, in each case, be done in $\OO(|I|) \cdot n^{\OO(1)}$ time. Now, note that, by \Cref{def_instance-GS}, $|I| = \diam(G)^{\OO(|X_i|)}$, and thus, computing all compatible instances can indeed be done in $\alpha(\OO(|X_i|)) \cdot n^{\OO(1)}$ time. Then, since computing the minimum using the formulas of Lemmas \ref{lem:join_node-GS}, \ref{lem:introduce_node-GS}, and \ref{lem:forget_node-GS} can be done in $\alpha(\OO(|X_i|))$ time, the lemma follows. 
\end{proof}

\subsection{(Kernelization) Algorithm for \smdfull}
\label{subsec:algo-vc-SMD}

We prove the following theorem.

\begin{theorem}
\label{thm:smd-algo-vertex-cover}
\smdfull admits 
\begin{itemize}
\item an \FPT\ algorithm running in time
$2^{2^{\mathcal{O}(\vc)}} \cdot n^{\OO(1)}$, and
\item a kernelization algorithm that outputs a kernel with
  $2^{\OO(\vc)}$ vertices.
\end{itemize}
\end{theorem}
\begin{proof}
Given a graph $G$,  let $X\subseteq V(G)$ be a minimum vertex cover of $G$. 
If such a vertex cover is not given, then we can find a $2$-factor approximate vertex
cover in polynomial time. 
Let $I:=V(G)\setminus X$.
By the definition of a vertex cover, the vertices of $I$ are pairwise non-adjacent.
The kernelization algorithm exhaustively applies the following reduction rule.
\begin{reduction rule} 
\label{reduc:twins-SMD}
If there exist three vertices $u,v,  x\in I$ such that $u,v, x$ are false twins, 
then delete $x$ and decrease $k$ by one.
\end{reduction rule}
Since $u,v, x$ are false twins, $N(u)=N(v)=N(x)$. 
This implies that, for any vertex $w\in V(G)\setminus \{u,v, x\}$, $d(w,v)=d(w,u)=d(w, x)$. 
In other words, for any $w \neq v$,
any shortest path from $u$ to $w$
does not contain $v$.
Hence, any strong resolving set that excludes at least two vertices in $\{u, v, x\}$
cannot resolve all three pairs $\langle u,v \rangle$,  $\langle u, x \rangle$,  and $\langle v, x \rangle$.
Hence, we can assume, without loss of generality, that 
any resolving set contains both $u$ and $x$.

Any pair of vertices in $V(G) \setminus \{u, x\}$ that 
is strongly resolved by $x$ is also resolved by $u$.
In other words, if $S$ is a strong resolving set of $G$, then
$S \setminus \{x\}$ is a strong resolving set of $G - \{x\}$.
This implies the correctness of the forward direction.
The correctness of the reverse direction trivially follows from the fact
that we can add $x$ into a strong resolving set of $G - \{x\}$
to obtain a resolving set of $G$.

Consider an instance on which the reduction rule is not applicable.
If the budget is negative, then the algorithm returns a trivial \no-instance
of constant size.
Otherwise, for any $Y\subseteq X$, there are at most two vertices $u,v\in I$ such that $N(u)=N(v)=Y$.
This implies that the number of vertices in the reduced instance
is at most $|X| + 2 \cdot 2^{|X|} =  2^{\vc +1} + \vc$.
The second part of the statement is an
immediate consequence of applying a brute-force
algorithm on the reduced instance.
\end{proof}
\section{Conclusion}\label{sec:conclu}

We have shown (under the \ETH) that three natural metric-based graph problems, \mdfull, \gsfull, and \smdfull, exhibit tight (double-) exponential running times for the standard structural parameterizations by treewidth and vertex cover number. This includes tight double-exponential running times for treewidth plus diameter (\mdfull and \gsfull) and for vertex cover (\smdfull).

Such tight double-exponential running times for \FPT\ 
structural paramaterizations of graph problems had previously been observed only for counting problems and problems complete for classes above \NP. 
Thus, surprisingly, our results show that some natural problems can be in \NP\ and still exhibit such a behavior.

It would be interesting to see whether this phenomenon holds for other graph problems in \NP, and for other structural parameterizations. 
Perhaps one can determine certain properties shared by these metric-based graph problems, that imply such running times, with the goal of generalizing our approach to a broader class of problems. 
In particular, concerning the general versatile technique that we designed to obtain the double-exponential lower bounds, it would be intriguing to see for which other problems in \NP\ our technique works.

In fact, after this paper appeared online, our technique was
successfully applied to an \NP-complete problem in machine learning~\cite{CCMR23} (for $\vc$)
as well as \NP-complete identification
problems~\cite{chakraborty2024tight} (for $\tw$).

\bibliography{bib}

\appendix

\end{document}